\newif\ifabstract
\newif\iffull
\ifabstract \fullfalse \else \fulltrue \fi

\documentclass[11pt]{article}
\usepackage{amsfonts}
\usepackage{amssymb}
\usepackage{amstext}
\usepackage{amsmath}
\usepackage{xspace}
\usepackage{theorem}
\usepackage{graphicx}
\usepackage{url}
\usepackage{graphics}
\usepackage{colordvi}
\usepackage{colordvi}
\usepackage{subfigure}
\usepackage[toc,page]{appendix}

\advance \oddsidemargin by -0.8in
\newcommand{\myparskip}{2pt}
\parskip \myparskip

\newenvironment{proofof}[1]{\noindent{\bf Proof of #1.}}%
        {\hfill\square}

\newcommand{\NDP}{{\sf NDP}\xspace}
\newcommand{\EDP}{{\sf EDP}\xspace}
\newcommand{\NDPwC}{{\sf NDPwC}\xspace}
\newcommand{\EDPwC}{{\sf EDPwC}\xspace}

\newcommand{\OR}{{\sc OR}\xspace}
\newcommand{\true}{{\sc True}\xspace}
\newcommand{\false}{{\sc False}\xspace}
\newcommand{\extra}{{\sc Extra}\xspace}

\newcommand{\yis}{{\sc Yes-Instances}\xspace}
\newcommand{\nis}{{\sc No-Instances}\xspace}
\newcommand{\yi}{{\sc Yes-Instance}\xspace}
\renewcommand{\ni}{{\sc No-Instance}\xspace}

\newcommand{\vargadgetwidth}{\ensuremath{80h+2}}

\newcommand{\Y}{\Upsilon}

\newcommand{\ceil}[1]{\ensuremath{\left\lceil#1\right\rceil}}

\newcommand{\bor}{\vee}

\newcommand{\event}{{\cal{E}}}

\renewcommand{\P}{\mbox{\sf P}\xspace}
\newcommand{\NP}{\mbox{\sf NP}\xspace}

\newcommand{\DTIME}{\mbox{\sf DTIME}}
\newcommand{\ZPTIME}{\mbox{\sf ZPTIME}}

\newcommand{\opt}{\mathsf{OPT}}

\newcommand{\set}[1]{\left\{ #1 \right\}}

\newcommand{\tset}{{\mathcal T}}
\newcommand{\iset}{{\mathcal{I}}}
\newcommand{\pset}{{\mathcal{P}}}
\newcommand{\oset}{{\mathcal{O}}}
\newcommand{\qset}{{\mathcal{Q}}}

\newcommand{\bset}{{\mathcal{B}}}
\newcommand{\aset}{{\mathcal{A}}}
\newcommand{\cset}{{\mathcal{C}}}
\newcommand{\fset}{{\mathcal{F}}}
\newcommand{\mset}{{\mathcal M}}
\newcommand{\hmset}{\hat{\mathcal M}}
\newcommand{\tmset}{\tilde{\mathcal M}}
\newcommand{\tpset}{\tilde{\mathcal P}}
\newcommand{\hpset}{\hat{\mathcal P}}

\newcommand{\wset}{{\mathcal{W}}}

\newcommand{\yset}{{\mathcal{Y}}}
\newcommand{\rset}{{\mathcal{R}}}

\newcommand{\sset}{{\mathcal{S}}}

\newcommand{\row}{\operatorname{row}}
\newcommand{\col}{\operatorname{col}}

\newcommand{\be}{\begin{enumerate}}
\newcommand{\ee}{\end{enumerate}}
\newcommand{\bd}{\begin{description}}
\newcommand{\ed}{\end{description}}
\newcommand{\bi}{\begin{itemize}}
\newcommand{\ei}{\end{itemize}}

\newtheorem{theorem}{Theorem}[section]
\newtheorem{lemma}[theorem]{Lemma}
\newtheorem{observation}[theorem]{Observation}
\newtheorem{corollary}[theorem]{Corollary}
\newtheorem{claim}[theorem]{Claim}

\newtheorem{definition}{Definition}[section]
\newenvironment{proof}{\par \smallskip{\bf Proof:}}{\hfill\stopproof}
\def\stopproof{\square}
\def\square{\vbox{\hrule height.2pt\hbox{\vrule width.2pt height5pt \kern5pt
\vrule width.2pt} \hrule height.2pt}}

\renewcommand{\phi}{\varphi}
\newcommand{\eps}{\epsilon}

\newcommand{\half}{\ensuremath{\frac{1}{2}}}

\newcommand{\poly}{\operatorname{poly}}

\newenvironment{properties}[2][0]
{
\begin{enumerate} \setcounter{enumi}{#1}}{\end{enumerate}}

\newcommand{\evtop}{e_v^{\mbox{\textup{\footnotesize{top}}}}}
\newcommand{\etop}[1]{e_{#1}^{\mbox{\textup{\footnotesize{top}}}}}
\newcommand{\evbot}{e_v^{\mbox{\textup{\footnotesize{bot}}}}}
\newcommand{\ebot}[1]{e_{#1}^{\mbox{\textup{\footnotesize{bot}}}}}

\newcommand{\evleft}{e_v^{\mbox{\textup{\footnotesize{left}}}}}
\newcommand{\eleft}[1]{e_{#1}^{\mbox{\textup{\footnotesize{left}}}}}

\newcommand{\evright}{e_v^{\mbox{\textup{\footnotesize{right}}}}}
\newcommand{\eright}[1]{e_{#1}^{\mbox{\textup{\footnotesize{right}}}}}

\begin{document}

\title{New Hardness Results for Routing on Disjoint Paths}
\author{Julia Chuzhoy\thanks{Toyota Technological Institute at Chicago. Email: {\tt cjulia@ttic.edu}. Supported in part by NSF grants CCF-1318242 and CCF-1616584.}\and David H. K. Kim\thanks{Computer Science Department, University of Chicago. Email: {\tt hongk@cs.uchicago.edu}. Supported in part by NSF grant CCF-1318242.} \and Rachit Nimavat\thanks{Toyota Technological Institute at Chicago. Email: {\tt nimavat@ttic.edu}. Supported in part by NSF grant CCF-1318242.}}

\begin{titlepage}
\maketitle

\thispagestyle{empty}

\begin{abstract}
In the classical Node-Disjoint Paths (\NDP) problem, the input consists of an undirected $n$-vertex graph $G$, and a collection $\mset=\set{(s_1,t_1),\ldots,(s_k,t_k)}$ of pairs of its vertices, called source-destination, or demand, pairs. The goal is to route the largest possible number of the demand pairs via node-disjoint paths. The best current approximation for the problem is achieved by a simple greedy algorithm, whose approximation factor is $O(\sqrt n)$, while the best current negative result is an $\Omega(\log^{1/2-\delta}n)$-hardness of approximation for any constant $\delta$, under standard complexity assumptions. Even seemingly simple special cases of the problem are still poorly understood: when the input graph is a grid, the best current algorithm achieves an $\tilde O(n^{1/4})$-approximation, and when it is a general planar graph, the best current approximation ratio of an efficient algorithm is $\tilde O(n^{9/19})$. The best currently known lower bound on the approximability of both these versions of the problem is APX-hardness.

In this paper we prove that \NDP is $2^{\Omega(\sqrt{\log n})}$-hard to approximate, unless all problems in \NP have algorithms with running time $n^{O(\log n)}$. Our result holds even when the underlying graph is a planar graph with maximum vertex degree $3$, and all source vertices lie on the boundary of a single face (but the destination vertices may lie anywhere in the graph). We extend this result to the closely related Edge-Disjoint Paths problem, showing the same hardness of approximation ratio even for sub-cubic planar graphs with all sources lying on the boundary of a single face.
\end{abstract}

\end{titlepage}
\label{------------------------------------------sec: intro-------------------------}
\section{Introduction}\label{sec: intro}
The main focus of this paper is the Node-Disjoint Paths (\NDP) problem: given an undirected $n$-vertex graph $G$, together with a collection $\mset=\set{(s_1,t_1),\ldots,(s_k,t_k)}$ of pairs of its vertices, called \emph{source-destination}, or \emph{demand} pairs, route the largest possible number of the demand pairs via node-disjoint paths. In other words, a solution to the problem is a collection $\pset$ of node-disjoint paths, with each path connecting a distinct source-destination pair, and the goal is to maximize $|\pset|$.  The vertices participating in the demand pairs of $\mset$ are called \emph{terminals}. \NDP is a classical routing problem, that has been extensively studied in both Graph Theory and Theoretical Computer Science communities. One of the key parts of Robertson and Seymour's Graph Minors series is an efficient algorithm for the special case of the problem, where the number $k$ of the demand pairs is bounded by a constant~\cite{RobertsonS,flat-wall-RS}; the running time of their algorithm is $f(k)\cdot \poly(n)$ for some large function $f$. However, when $k$ is a part of input, the problem becomes \NP-hard~\cite{Karp,EDP-hardness}, even on planar graphs~\cite{npc_planar}, and even on grid graphs~\cite{npc_grid}. The following simple greedy algorithm provides an $O(\sqrt n)$-approximation for \NDP~\cite{KolliopoulosS}: Start with $\pset=\emptyset$. While $G$ contains a path connecting any demand pair, select the shortest such path $P$, add it to $\pset$, and delete all vertices of $P$ from $G$. Surprisingly, despite the extensive amount of work on the problem and its variations, this elementary algorithm remains the best currently known approximation algorithm for the problem, and until recently, this was true even for the special cases where $G$ is a planar graph, or a grid graph. The latter two special cases have slightly better algorithms now: a recent result of Chuzhoy and Kim~\cite{NDP-grids} gives a $\tilde O(n^{1/4})$-approximation for \NDP on grid graphs, and Chuzhoy, Kim and Li~\cite{NDP-planar} provide a $\tilde O(n^{9/19})$-approximation algorithm for the problem on planar graphs. The best current negative result shows that \NDP has no $O(\log^{1/2-\delta}n)$-approximation algorithms for any constant $\delta$, unless $\NP \subseteq \ZPTIME(n^{\poly \log n})$~\cite{AZ-undir-EDP,ACGKTZ}. For the special case of grids and planar graphs only APX-hardness is currently known on the negative side~\cite{NDP-grids}. 

The main result of this paper is that \NDP is $2^{\Omega(\sqrt{\log n})}$-hard to approximate unless $\NP\subseteq \DTIME(n^{O(\log n)})$, even if the underlying graph is a planar graph with maximum vertex degree at most $3$, and all source vertices $\set{s_1,\ldots,s_k}$ lie on the boundary of a single face. We note that \NDP can be solved efficiently on graphs whose maximum vertex degree is $2$\footnote{A graph whose maximum vertex degree is $2$ is a collection of disjoint paths and cycles. It is enough to solve the problem on each such path and cycle separately. The problem is then equivalent to computing a maximum independent set of intervals on a line or on a circle, and can be solved efficiently by standard methods.}.

A problem closely related to \NDP is Edge-Disjoint Paths (\EDP). The input to this problem is the same as to \NDP, and the goal is again to route the largest possible number of the demand pairs. However, the routing paths are now allowed to share vertices, as long as they remain disjoint in their edges. The two problems are closely related: it is easy to see that \EDP is a special case of \NDP, by using the line graph of the \EDP instance to obtain an equivalent \NDP instance (but this transformation may inflate the number of vertices, and so approximation factors depending on $n$ may not be preserved). This relationship is not known for planar graphs, as the line graph of a planar graph is not necessarily planar. 
The current approximability status of \EDP is similar to that of \NDP: the best current approximation algorithm achieves an $O(\sqrt n)$-approximation factor~\cite{EDP-alg}, and the best current negative result is an $\Omega(\log^{1/2-\delta}n)$-hardness of approximation for any constant $\delta$, unless $\NP \subseteq \ZPTIME(n^{\poly \log n})$~\cite{AZ-undir-EDP,ACGKTZ}.  An analogue of the special case of \NDP on grid graphs for the \EDP problem is when the input graph is a wall, and the work of~\cite{NDP-grids} gives an $\tilde O(n^{1/4})$-approximation algorithm for \EDP on wall graphs. However, for planar graphs, no better than $O(\sqrt n)$-approximation is currently known for \EDP, and it is not clear whether the algorithm of~\cite{NDP-planar} can be adapted to this setting in order to break the $O(\sqrt n)$-approximation barrier for \EDP in planar graphs. 
Our hardness result extends to \EDP on planar sub-cubic graphs, where all source vertices lie on the boundary of a single face.

We say that two instances $(G,\mset)$ and $(G',\mset')$ of the \NDP problem are \emph{equivalent} iff the sets of terminals in both instances are the same, $\mset=\mset'$, and for every subset $\tmset\subseteq \mset$ of demand pairs, the pairs in $\tmset$ are routable via node-disjoint paths in $G$ iff they are routable in $G'$. Equivalence of \EDP problem instances is defined similarly, with respect to edge-disjoint routing. We show in Section~\ref{sec: large-degree-NDP} of the Appendix that for every integer $d>3$, there is an instance $(G,\mset)$ of \NDP, where $G$ is a planar graph, such that for every instance $(G',\mset')$ of \NDP that is equivalent to $(G,\mset)$, some vertex of $G'$ has degree at least $d$. Therefore, the class of all planar graphs is strictly more general than the class of all planar graphs with maximum vertex degree at most $3$ for \NDP. In contrast, it is well known that for any instance $(G,\mset)$ of \EDP, there is an equivalent instance $(G',\mset')$ of the problem, with maximum vertex degree at most $4$\footnote{We use a slightly different definition of equivalence for \EDP: namely, we only require that for every subset $\tmset$ of demand pairs, such that every terminal participates in at most one demand pair in $\tmset$, set $\tmset$ is routable in $G$ iff it is routable in $G'$. This definition is more appropriate for the \EDP problem, as we discuss in Section~\ref{sec: EDP-degree-reduction}.}. Moreover, if $G$ is planar then $G'$ can also be made planar.  Informally, graph $G'$ is obtained from $G$ by replacing every large-degree vertex with a grid: if the degree of a vertex $v$ is $d_v>4$, then we replace $v$ with the $(d_v\times d_v)$-grid, and connect the edges incident to $v$ to the vertices on the first row of the grid. This transformation inflates the number of vertices, but can be performed in a way that preserves planarity. However, as we show in Section~\ref{sec: EDP-degree-reduction},  there is an instance $(G,\mset)$ of \EDP, where $G$ is a planar graph, such that for every instance $(G',\mset')$ of \EDP that is equivalent to $(G,\mset)$, such that $G'$ is planar, the maximum vertex degree in $G'$ is at least $4$. This shows that we cannot reduce the degree all the way to $3$.

Interestingly, better algorithms are known for several special cases of \EDP on planar graphs. Kleinberg~\cite{Kleinberg-planar}, building on the work of Chekuri, Khanna and Shepherd~\cite{CKS,CKS-planar1}, has shown an $O(\log^2n)$-approximation algorithm for even-degree planar graphs. Aumann and Rabani~\cite{grids1} showed an $O(\log^2 n)$-approximation algorithm for \EDP on grids, and Kleinberg and Tardos~\cite{grids3,grids4} showed $O(\log n)$-approximation algorithms for broader classes of nearly-Eulerian uniformly high-diameter planar graphs, and nearly-Eulerian densely embedded graphs. Recently, Kawarabayashi and Kobayashi~\cite{KK-planar} gave an $O(\log n)$-approximation algorithm for \EDP on 4-edge-connected planar graphs and on Eulerian planar graphs.  It seems that the restriction of the graph $G$ to be planar Eulerian, or nearly-Eulerian, makes the \EDP problem significantly more tractable. In contrast, the graphs we construct in our hardness of approximation proof are sub-cubic, and far from being Eulerian.

A variation of the \NDP and \EDP problems, where small congestion is allowed, has been a subject of extensive study. In the \NDP with congestion (NDPwC) problem, the input is the same as in the \NDP problem, and we are additionally given an integer $c\geq 1$. The goal is to route as many of the demand pairs as possible with congestion at most $c$: that is, every vertex may participate in at most $c$ paths in the solution. \EDP with Congestion (EDPwC) is defined similarly, except that now the congestion bound is imposed on the graph edges and not vertices. The classical randomized rounding technique of Raghavan and Thompson~\cite{RaghavanT} gives a constant-factor approximation for both problems, if the congestion $c$ is allowed to be as high as $\Theta(\log n/\log\log n)$. A long line of work~\cite{CKS,Raecke,Andrews,RaoZhou,Chuzhoy11,  ChuzhoyL12,ChekuriE13,NDPwC2} has lead to an $O(\poly\log k)$-approximation for both \NDPwC and \EDPwC problems, with congestion bound $c=2$. For planar graphs, a constant-factor approximation with congestion 2 is known for \EDP~\cite{EDP-planar-c2}. Our new hardness results demonstrate that there is a dramatic difference in the approximability of routing problems with congestion $1$ and $2$.

\paragraph{ Our Results and Techniques.}
Our main result is the proof of the following two theorems.

\begin{theorem}\label{thm: main}
There is a constant $c$, such that no efficient algorithm achieves a factor $2^{c\sqrt{\log n}}$-approximation for \NDP, unless $\NP\subseteq \DTIME(n^{O(\log n)})$. This result holds even for planar graphs with maximum vertex degree $3$, where all source vertices lie on the boundary of a single face.
\end{theorem}

\begin{theorem}\label{thm: main-EDP}
There is a constant $c$, such that no efficient algorithm achieves a factor $2^{c\sqrt{\log n}}$-approximation for \EDP, unless $\NP\subseteq \DTIME(n^{O(\log n)})$. This result holds even for planar graphs with maximum vertex degree $3$, where all source vertices lie on the boundary of a single face.
\end{theorem}

We now provide an informal high-level overview of the proof of Theorem~\ref{thm: main}. 
It is somewhat easier to describe the proof of the theorem for the case where the maximum vertex degree is allowed to be $4$ instead of $3$. This proof can then be easily modified to ensure that the maximum vertex degree in the instances we obtain does not exceed $3$, and also extended to the \EDP problem.
We perform a reduction from the 3SAT(5) problem. In this problem, we are given a SAT formula $\phi$ defined over a set of $n$ Boolean variables. The formula consists of $m$ clauses, each of which is an \OR of three literals, where every literal is either a variable or its negation. Every variable of $\phi$ participates in exactly $5$ distinct clauses, and the literals of every clause correspond to three distinct variables. We say that $\phi$ is a \yi, if there is an assignment to its variables that satisfies all clauses, and we say that it is a \ni, if no assignment satisfies more than a $(1-\eps)$-fraction of the clauses, for some fixed constant $0<\eps<\half$. The famous PCP theorem~\cite{AS98,ALMSS} shows that, unless $\P=\NP$, no efficient algorithm can distinguish between the {\sc Yes-} and the \nis of 3SAT(5).

We perform $\Theta(\log n)$ iterations, where in iteration $i$ we construct what we call a \emph{level-$i$} instance. We use two parameters, $N_i$ and $N'_i$, and ensure that, if the reduction is performed from a formula $\phi$ which is a \yi, then there is a solution to the level-$i$  instance of \NDP that routes $N_i$ demand pairs, while if $\phi$ is a \ni, then no solution routes more than $N'_i$ demand pairs. We let $g_i=N_i/N'_i$ be the gap achieved by the level-$i$ instance. Our construction ensures that the gap grows by a small constant factor in every iteration, so $g_i=2^{\Theta(i)}$, while the instance size grows by roughly factor-$\Theta(n\cdot g_{i-1})$ in iteration $i$. Therefore, after $\Theta(\log n)$ iterations, the gap becomes $2^{\Omega(\log n)}$, while the instance size becomes $n'=2^{O(\log^2n)}$, giving us the desired $2^{\Omega(\sqrt{\log n'})}$-hardness of approximation, unless $\NP\subseteq\DTIME(n^{O(\log n)})$. 

In all our instances of \NDP, the underlying graph is a subgraph of a grid, with all sources lying on the top boundary of the grid; all vertices participating in the demand pairs are distinct.
 In the first iteration, a level-$1$ instance is constructed by a simple reduction from 3SAT(5), achieving a small constant gap $g_1$. Intuitively, once we construct a level-$i$ instance, in order to construct a level-$(i+1)$ instance, we replace every demand pair from a level-$1$ instance with a collection of level-$i$ instances. 
 In order to be able to do so, we need the instances to be ``flexible'', so that, for example, we have some freedom in choosing the locations of the source  and the destination vertices of a given level-$i$ instance in the grid.
 
 We achieve this flexibility by defining, for each level $i$, a family of level-$i$ instances. 
 The graph associated with a level-$i$ instance $\iset$ is a subgraph of a large enough grid $G_i$. The construction of the instance consists of two parts. First, we construct a path $Z(\iset)$, and place all source vertices on this path. Second, we construct a vertex-induced subgraph $B(\iset)$ of a relatively small grid $G'_i$, and we call $B(\iset)$ a \emph{box}. We place all the destination vertices inside the box $B(\iset)$. Graphs $Z(\iset)$ and $G'_i$ are completely disjoint from the grid $G_i$ and from each other. 
 In order to construct a specific level-$i$ instance, we select a placement of the path $Z(\iset)$ on the first row of the grid $G_i$, and a placement of the box $B(\iset)$ in $G_i$, far enough from its boundaries (see Figure~\ref{fig: level-i-schematic}). In other words, we choose a sub-path $P$ of the first row of the grid, of the same length as $Z(\iset)$, and map the vertices of $Z(\iset)$ to $P$ in a natural way. We also choose a sub-grid $G''_i$ of $G_i$, of the same dimensions as $G'_i$, and map the vertices of $G'_i$ to the vertices of $G''_i$ in a natural way. Since $B(\iset)\subseteq G'_i$, this also defines a mapping of the vertices of $B(\iset)$ to the vertices of $G_i$. Once these placements are selected, the mapping of the vertices of $Z(\iset)$ to the vertices of $P$  determines the identities of the source vertices, and the mapping of the vertices of $B(\iset)$ to the vertices of $G''_i$ determines the identities of the destination vertices. We delete from $G_i$ all vertices to which the vertices of $G'_i\setminus B(\iset)$ are mapped. In other words, all vertices that were removed from $G'_i$ to construct $B(\iset)$, are also removed from $G_i''$, and hence from $G_i$.
We note that box $B(\iset)$ may be constructed recursively, for example, by placing several boxes $B(\iset')$ corresponding to lower-level instances $\iset'$ inside it. The mapping of the vertices of $B(\iset')$ to the vertices of $B(\iset)$, the placement of the destination vertices, and the removal of the vertices of $B(\iset)$ corresponding to the grid vertices removed from $B(\iset')$ is done similarly. In order to reduce the maximum vertex degree to $3$, we can use wall graphs instead of grid graphs and employ a similar proof. Alternatively, a simple modification of the final instance we obtain can directly reduce its maximum vertex degree to $3$.

\begin{figure}[h]
\begin{center}
\scalebox{0.4}{\includegraphics{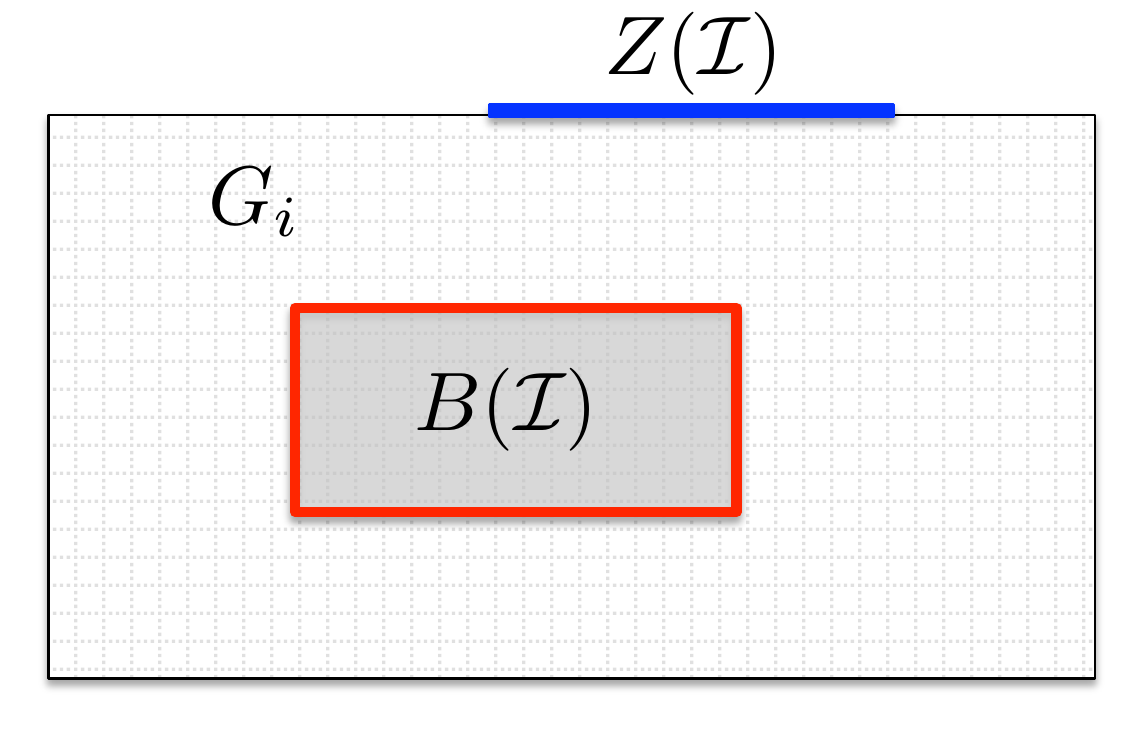}}
\caption{A schematic view of a level-$i$ instance $\iset$. All source vertices lie on $Z(\iset)$, whose location can be chosen arbitrarily on the first row of $G_i$. All destination vertices belong to $B(\iset)$, that can be located anywhere in $G_i$, far enough from its boundaries.\label{fig: level-i-schematic}}
\end{center}
\end{figure}

The most natural intuitive way to think about our construction is the one described above.  An equivalent, and somewhat easier way to define our construction is slightly different: we let a level-$0$ instance be an instance consisting of a single demand pair $(s,t)$, with $s$ lying on the first row of the grid and $t$ lying far from the grid boundary. We then show, for each $i>0$, a procedure that constructs a level-$i$ instance by combining a number of level-$(i-1)$ instances. The latter definition is somewhat more convenient, because it saves us the need to provide a separate correctness proof for level-$1$ instances, which is essentially identical to the proof for higher-level instances. However, we still feel that defining level-$1$ instances explicitly is useful for the sake of intuition. Therefore, we start with preliminaries in Section~\ref{sec: prelims} and describe our construction of level-$1$ instances in Section~\ref{sec: level 1}, together with an intuition for constructing higher-level instances. We only provide a sketch of the correctness proof, as the complete correctness proof appears in the following sections. In Section~\ref{sec: level i}, we define our construction in two steps: by first defining level-$0$ instances, and then showing how to construct level-$i$ instances from level-$(i-1)$ instances. The resulting level-$1$ instances will be similar to those defined in Section~\ref{sec: level 1}. We provide a complete proof of correctness in Sections~\ref{sec: level i}--\ref{sec: NI}. This provides a proof of Theorem~\ref{thm: main} for the case where the maximum vertex degree is allowed to be $4$.
 Section~\ref{sec: EDP} of the Appendix extends our result to \EDP in planar graphs, and show how to reduce the degree of the hard \NDP instances to $3$, completing the proofs of Theorems~\ref{thm: main} and \ref{thm: main-EDP}. Finally, sections~\ref{sec: large-degree-NDP} and~\ref{sec: EDP-degree-reduction} discuss degree reduction in \NDP and \EDP instances, respectively.

\label{-----------------------------------------------sec: prelims--------------------------------------------}
\section{Preliminaries}\label{sec: prelims}
For a pair  $\ell,h> 0$ of integers, we let $G^{\ell,h}$ denote a grid of length $\ell$ and height $h$.
The set of vertices of $G^{\ell,h}$ is $V(G^{\ell,h})=\set{v(i,j)\mid 1\leq i\leq h, 1\leq j\leq \ell}$, and the set of its edges is the union of two subsets: the set of horizontal edges $E^H=\set{(v_{i,j},v_{i,j+1})\mid 1\leq i\leq h, 1\leq j<\ell}$ and the set of vertical edges $E^{V}=\set{(v_{i,j},v_{i+1,j})\mid 1\leq i< h, 1\leq j\leq\ell}$. The subgraph $G^{\ell,h}$ induced by the edges of $E^H$ consists of $h$ paths, that we call the \emph{rows} of the grid; for $1\leq i\leq h$, the $i$th row $R_i$ is the row containing the vertex $v(i,1)$. Similarly, the subgraph induced by the edges of $E^V$ consists of $\ell$ paths that we call the \emph{columns} of the grid, and for $1\leq j\leq \ell$, the $j$th column $W_j$ is the column containing $v(1,j)$. We think of the rows  as ordered from top to bottom and the columns as ordered from left to right. Row $R_{\ceil{h/2}}$ is called the \emph{middle row} of the grid $G^{\ell,h}$. 
Given two vertices $u=v(i,j),u'=v(i',j')$ of the grid, the distance between them is $d(u,u')=|i-i'|+|j-j'|$. Given two vertex subsets $X,Y\subseteq V(G^{\ell,h})$, the distance between them is $d(X,Y)=\min_{u\in X,u'\in Y}\set{d(u,u')}$. Given a vertex $v=v(i,j)$ of the grid, we denote by $\row(v)$ and $\col(v)$ the row and the column, respectively, that contain $v$.

Given a set $\rset$ of consecutive rows of a grid $G=G^{\ell,h}$ and a set $\wset$ of consecutive columns of $G$, we let $\Y(\rset,\wset)$ be the subgraph of $G$ induced by the set $\set{v(j,j')\mid R_j\in \rset, W_{j'}\in \wset}$ of vertices. We say that $\Y=\Y(\rset,\wset)$ is the \emph{sub-grid of $G$ spanned by the set $\rset$ of rows and the set $\wset$ of columns}. 

Assume now that we are given a grid $G$, a sequence $\sset=(G_1,\ldots,G_r)$ of disjoint sub-grids of $G$, and an integer $N$. We say that the grids of $\sset$ are \emph{aligned and $N$-separated} iff the middle row of each grid $G_i$ is a sub-path of the middle row of $G$; the grids in $\set{G_1,\ldots,G_r}$ appear in this left-to-right order inside $G$; every pair of consecutive grids $G_i$ is separated by at least $N$ columns of $G$; and every grid in $\sset$ is separated by at least $N$ columns from the right and the left boundaries of $G$.

Let $R',R''$ be two distinct rows of some grid $G$, and let $\pset$ be a set of node-disjoint paths, such that every path in $\pset$ has one endpoint  (called a source) on row $R'$ and another (called a destination) on row $R''$. We say that the set $\pset$ of paths is \emph{order-preserving} iff the source vertices of the paths in $\pset$ appear on row $R'$ in exactly the same left-to-right order as their destination vertices on $R''$.

Throughout our construction, we use the notion of a box. A box $B$ of length $\ell$ and height $h$ is a vertex-induced subgraph of $G^{\ell,h}$. We denote $U(B)=V(G^{\ell,h})\setminus V(B)$, and we sometimes think of set $U(B)$ as the ``set of vertices deleted from B''. We say that $B$ is a \emph{cut-out box} iff $U(B)$ contains all vertices lying on the left, right, and bottom boundaries of $G^{\ell,h}$; note that $U(B)$ may contain additional vertices of $G^{\ell,h}$. The remaining vertices of the top boundary of $G^{\ell,h}$ that belong to $V(B)$ are called \emph{the opening of $B$}. We sometimes say that the vertices of $B$ that belong to row $R_{\ceil{h/2}}$ of $G^{\ell,h}$ lie on the middle row of $B$.

Given any set  $\mset$ of demand pairs, we let $S(\mset)$ denote the set of all source vertices participating in $\mset$ and $T(\mset)$ the set of all destination vertices. Given a path $P$, the length of the path is the number of vertices on it. We say that two paths $P,P'$ are \emph{internally disjoint} iff every vertex in $P\cap P'$ is an endpoint of both paths. Given a set $\pset$ of paths and some set $U$ of vertices, we say that the paths in $\pset$ are internally disjoint from $U$ iff for every path $P\in \pset$, every vertex in $P\cap U$ is an endpoint of $P$. Given a subgraph $G'$ of a graph $G$ and a set $\pset$ of paths in $G$, we say that the paths in $\pset$ are internally disjoint from $G'$ iff they are internally disjoint from $V(G')$.

As already described in the introduction, for every level $0\leq i\leq \Theta(\log n)$, we construct a level-$i$ instance $\iset$. In fact, it is a family of instances, but it is more convenient to think of it as one instance with different instantiations. A definition of a level-$i$ instance $\iset$ consists of the following ingredients:

\begin{itemize}
\item integral parameters $L_i,L'_i$ and an even integer $H_i$;
\item a path $Z(\iset)$ of length $L_i$;
\item a grid $G'_i$ of length $L'_i$ and height $H_i$, together with a cut-out box $B(\iset)\subseteq G_i'$; and
\item a set $\mset$ of demand pairs, together with a mapping of the vertices of $S(\mset)$ to distinct vertices of $Z(\iset)$ and a mapping of the vertices of $T(\mset)$ to distinct vertices on the middle row of $B(\iset)$.
\end{itemize}

In order to instantiate a level-$i$ instance $\iset$, we select a grid $G_i$ of length at least $2L_i+2L'_i+4H_i$ and height at least $3H_i$, a sub-path $P$ of the first row of $G_i$ of length $L_i$, and a sub-grid $G''_i$ of $G_i$ of height $H_i$ and length $L'_i$, so that the distance from the vertices of $G''_i$ to the vertices lying on the boundary of $G_i$ is at least $H_i$. We map every vertex of $Z(\iset)$ to the corresponding vertex of $P$ in a natural way, and this determines the identities of the source vertices in the instance we construct. We also map every vertex of $G_i'$ to the corresponding vertex of $G_i''$, and this determines the identities of the destination vertices. Finally, for every vertex $u\in U(B(\iset))$, we delete the vertex of $G_i''$ to which $u$ is mapped from $G_i$. This defines an instance of \NDP on a subgraph of $G_i$, where all the sources lie on the top boundary of $G_i$ and all source and destination vertices are distinct.

Assume now that we are given an instantiation of a level-$i$ instance $\iset$ and a set $\pset$ of node-disjoint paths routing a subset $\mset'$ of the demand pairs in that instance. Assume for convenience that $\mset'=\set{(s_1,t_1),\ldots,(s_r,t_r)}$, that the vertices $s_1,\ldots,s_r$ appear in this left-to-right order on $Z(\iset)$, and that $\pset=\set{P_1,\ldots,P_r}$, where path $P_j$ connects $s_j$ to $t_j$. Let $A$ be the set of all vertices of the top row of the grid $G''_i$ that were not deleted (that is, these are the vertices lying on the opening of $B(\iset)$). We say that the set $\pset$ of paths \emph{respects} the box $B(\iset)$ iff for all $1\leq j\leq r$, $P_j\cap A$ is a single vertex, that we denote by $u_j$, and $u_j$ is the $j$th vertex of $A$ from the left. Intuitively, the paths in $\pset$ connect the sources to a set of consecutive vertices on the opening of $B(\iset)$ in a straightforward manner, and the actual routing occurs inside the box $B(\iset)$.

We perform a reduction from the 3SAT(5) problem. In this problem, we are given a SAT formula $\phi$ on a set $\set{x_1,\ldots,x_n}$ of $n$ Boolean variables and a set $\cset=\set{C_1,\ldots,C_m}$ of $m=5n/3$ clauses. Each clause contains $3$ literals, each of which is either a variable or its negation. The literals of each clause correspond to $3$ distinct variables, and each variable participates in exactly $5$ clauses. We denote the literals of the clause $C_q$ by $\ell_{q_1},\ell_{q_2}$ and $\ell_{q_3}$.
A clause is satisfied by an assignment to the variables iff at least one of its literals evaluates to \true. We say that $\phi$ is a \yi if there is an assignment to its variables satisfying all its clauses. We say that it is a \ni with respect to some parameter $\epsilon$, if no assignment satisfies more than $(1-\epsilon)m$ clauses. The following well-known theorem follows from the PCP theorem~\cite{AS98,ALMSS}.

\begin{theorem} \label{thm: PCP}
There is a constant $\epsilon:0<\epsilon<\half$, such that it is NP-hard to distinguish between the \yis and the \nis (defined with respect to $\epsilon$) of the 3SAT(5) problem.
\end{theorem}


Given an input formula $\phi$, we will construct an instance $(G,\mset)$ of the \NDP problem with $|V(G)|=n'=n^{O(\log n)}$, that has the following properties: if $\phi$ is a \yi, then there is a solution to the \NDP instance routing $N$ demand pairs, for some parameter $N$; if $\phi$ is a \ni, then at most $N/g$ demand pairs can be routed, where $g=2^{\Omega(\log n)}=2^{\Omega(\sqrt{\log n'})}$. This will prove that no efficient algorithm can achieve a better than factor $2^{O(\sqrt{\log n})}$-approximation for \NDP, unless $\NP\subseteq \DTIME(n^{O(\log n)})$. The instance we construct is a subgraph of a grid with all source vertices lying on its top boundary, so the hardness result holds for planar graphs with maximum vertex degree $4$, with all sources lying on the boundary of a single face. In Section~\ref{sec: EDP} of the Appendix, we modify this instance to reduce its maximum vertex degree to $3$.


\label{--------------------------------------------------------sec: level 1--------------------------------------------------}
\section{The Level-1 Instance}\label{sec: level 1}
In this section we define our level-$1$ instance $\iset$ and provide intuition for generalizing it to higher-level instances. Since Sections~\ref{sec: level i}--\ref{sec: NI} contain all formal definitions and proofs, including those for the level-1 instance, the description here is informal, and we only provide proof sketches.

 We assume that we are given a 3SAT(5) formula $\phi$ defined over a set $\set{x_1,\ldots,x_n}$ of variables and a set $\cset=\set{C_1,\ldots,C_m}$ of clauses, so $m=5n/3$.  For every variable $x_j$ of $\phi$, we will define a set $\mset(x_j)$ of demand pairs that represent that variable, and similarly, for every clause $C_q\in \cset$ we will define a set $\mset(C_q)$ of demand pairs representing it. We call the demand pairs in set $\mset^V=\bigcup_{j=1}^n\mset(x_j)$ \emph{variable-pairs} and the demand pairs in set $\mset^C=\bigcup_{C_q\in \cset}\mset(C_q)$ \emph{clause-pairs}.

Let $h=1000/\eps$ and $\delta=8\eps^2/10^{12}$, where $\eps$ is the parameter from Theorem~\ref{thm: PCP}. We set $N_1=(200h/3+1)n$ and $N'_1=(1-\delta)N_1$. Our construction will ensure that, if the input formula $\phi$ is a \yi, then for every instantiation of $\iset$, there is a collection $\pset$ of node-disjoint paths that respects the box $B(\iset)$ and routes $N_1$ demand pairs. On the other hand, if $\phi$ is a \ni, then no solution can route more than $N_1'$ demand pairs  in any instantiation of $\iset$. This gives a gap of $1/(1-\delta)$ between the {\sc Yes-} and \ni solution costs. In the following levels we gradually amplify this gap.

We set $L_1=(80h+2)n$, $L'_1=20N_1^3$ and $H_1=20N_1$. 
In order to construct a level-$1$ instance $\iset$, we start with a path $Z(\iset)$ of length $L_1$ and a grid $G_1'$ of length $L'_1$ and height $H_1$. We delete all vertices lying on the bottom, left and right boundaries of $G_1'$ to obtain the initial cut-out box $B(\iset)$; we will later delete some additional vertices from $B(\iset)$.

We define two sub-grids of $G_1'$: grid $B^V$, that will contain all vertices of $T(\mset^V)$ (the destination vertices of the demand pairs in $\mset^V$), and grid $B^C$, that will contain all vertices of $T(\mset^C)$. Both grids have sufficiently large length and height: length $9N_1^3$ and height $16N_1$ for each grid are sufficient. We place both grids inside $G_1'$, so that the middle row of each grid is contained in the middle row of $G_1'$,  there is a horizontal spacing of at least $2N_1$ between the two grids, and both grids are disjoint from the left and the right boundaries of $G_1'$. It is easy to see that at least $2N_1$ rows of $G_1'$ lie above and below both grids (see Figure~\ref{level-1}).

Next, we define smaller sub-grids of the grids $B^V$ and $B^C$. For every variable $x_j$ of $\phi$, we select a sub-grid $B(x_j)$ of $B^V$ of height $H^V=4N_1+3$ and length $L^V=4N_1+(70h+2)$. This is done so that the sequence $B(x_1),\ldots,B(x_n)$ of grids is aligned and $(2N_1)$-separated
in $B^V$ (see Figure~\ref{level-1}). Recall that this means that the middle row of each grid coincides with the middle row of $B^V$, and the horizontal distance between every pair of these grids, and between each grid and the right and the left boundaries of $B^V$ is at least $2N_1$. 
It is easy to verify that grid $B^V$ is large enough to allow this. For each variable $x_j$ of $\phi$, the vertices of $T(\mset(x_j))$ will lie in $B(x_j)$. 
Note that there are at least $2N_1$ rows of $B^V$ above and below each box $B(x_1),\ldots,B(x_n)$.

Similarly, for every clause $C_q\in \cset$, we define a sub-grid $B(C_q)$ of $B^C$ of length $L^C=3h$ and height $H^C=3$. We select the sub-grids $B(C_1),\ldots,B(C_q)$ of $B^C$ so that they are aligned and $(4N_1)$-separated. As before, box $B^C$ is sufficiently large to allow this, and there are at least $2N_1$ rows of $B^C$ both above and below each such grid $B(C_q)$ (see Figure~\ref{level-1}). 

\begin{figure}[h]
\scalebox{0.6}{\includegraphics{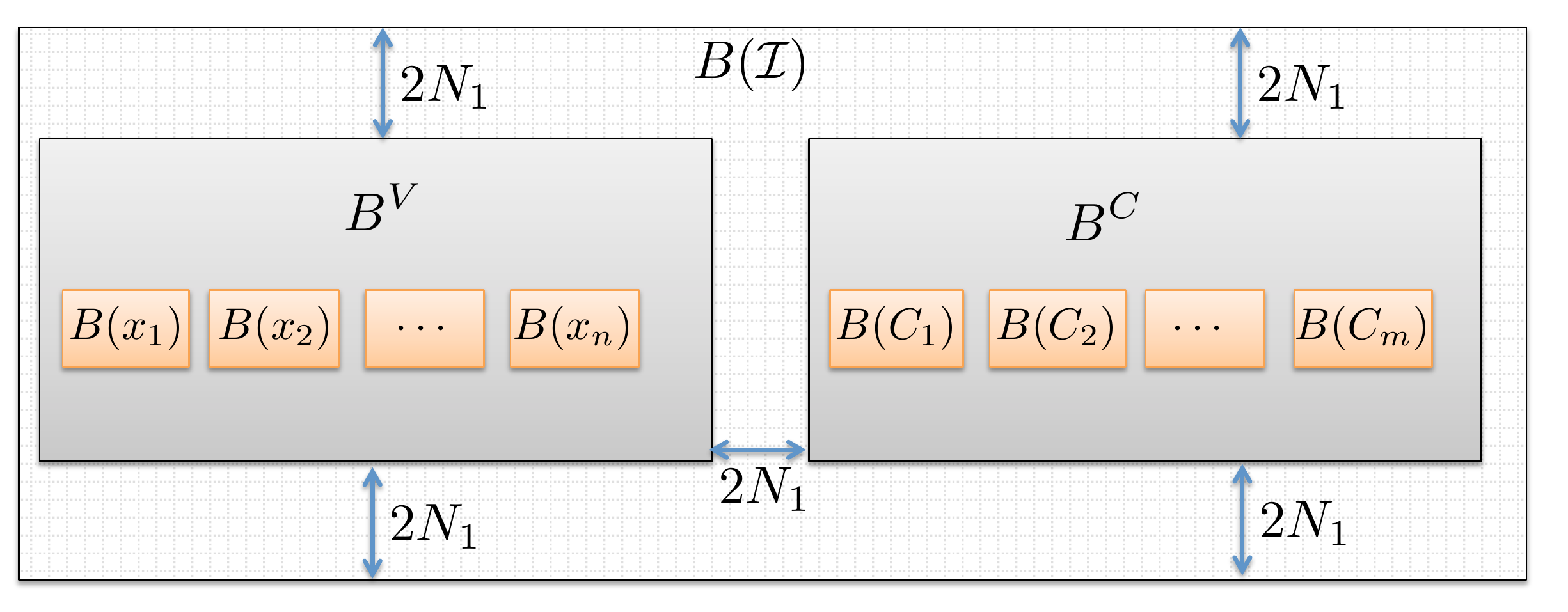}}
\caption{High-level view of the level-1 construction. Boxes $B(x_1),\ldots,B(x_n)$ are at a distance at least $2N_1$ from each other and from  the boundaries of $B^V$, and boxes $B(C_1),\ldots,B(C_m)$ are at a distance at least $4N_1$ from each other and from the left and right boundaries of $B^C$.\label{level-1}}
\end{figure}


Recall that the length of the path $Z(\iset)$ is $L_1=(80h+2)n$. We partition $Z(\iset)$ into $n$ disjoint sub-paths $I(x_1),I(x_2),\ldots,I(x_n)$ of length $80h+2$ each, that we refer to as intervals. For each $1\leq j\leq n$, vertices of $S(\mset(x_j))$ will lie on $I(x_j)$. Additionally, for every clause $C_q$ in which variable $x_j$ participates, path $I(x_j)$ will contain some vertices of $S(\mset(C_q))$.
The remainder of the construction consists of two parts --- variable gadgets and clause gadgets, that we define next, starting with the variable gadgets.



\paragraph{ Variable gadgets.}
Consider some variable $x$ of the formula $\phi$ and its corresponding interval $I(x)$ of $Z(\iset)$. We partition $I(x)$ as follows. Let $I^T(x), I^F(x)\subseteq I(x)$ denote the sub-intervals of $I(x)$ containing the first and the last $(10h+1)$ consecutive vertices of $I(x)$, respectively, and let $I^X(x)$ be the interval containing the remaining $60h$ vertices.
Consider the box $B(x)$, and recall that it has length $4N_1+70h+2$ and height $4N_1+3$. Let $B'(x)\subseteq B(x)$ be a sub-grid of $B(x)$ of length $70h+2$ and height $3$, so that there are exactly $2N_1$ rows of $B(x)$ above and below $B'(x)$, and $2N_1$ columns of $B(x)$ to the left and to the right of $B'(x)$. Notice that the middle row of $B'(x)$, that we denote by $R'(x)$, is aligned with the middle row of $B(x)$ and hence of $B(\iset)$. We delete all vertices of $B'(x)$ that lie on its bottom row, and we will place all destination vertices of the demand pairs in $\mset(x)$ on $R'(x)$. In order to do so, we further partition $R'(x)$ into three intervals: interval $\hat I^T(x)$ contains the first $5h+1$ vertices of $R'(x)$; interval $\hat I^F(x)$ contains the following $5h+1$ vertices of $R'(x)$; and interval  $\hat I^X(x)$ contains the remaining $60h$ vertices of $R'(x)$ (see Figure~\ref{fig: level 1 var gadget}). The set $\mset(x)$ of demand pairs consists of three subsets:

 \begin{itemize}

 \item {\bf (Extra Pairs).} Let $\mset^X(x)=\set{(s_y^X(x),t_y^X(x))}_{y=1}^{60h}$ be a set of $60h$ demand pairs that we call the \extra pairs for $x$.  The vertices $s_1^X(x),\ldots,s_{60h}^X(x)$ appear on $I^X(x)$ in this order, and the vertices $t_1^X(x),\ldots,t_{60h}^X(x)$ appear on $\hat I^X(x)$ in this order.

 \item  {\bf (\true Pairs).} 
We denote the vertices appearing on $I^T(x)$ by $a^T_1,b^T_1,a^T_2,b^T_2,\ldots, b^T_{5h},a^T_{5h+1}$ in this left-to-right order. Let $\mset^T(x)=\set{(s_y^T(x),t_y^T(x))}_{y=1}^{5h+1}$ be a set of $(5h+1)$ demand pairs that we call the \true demand pairs for $x$. For each $1\leq y\leq 5h+1$, we identify $s_y^T(x)$ with the vertex $a^T_y$ of $I^T(x)$, and we let $t_y^T(x)$ be the $y$th vertex on $\hat I^T(x)$.

  \item  {\bf (\false Pairs).} 
Similarly, we denote the vertices appearing on $I^F(x)$ by $a^F_1,b^F_1,a^F_2,b^F_2,\ldots, b^F_{5h},a^F_{5h+1}$ in this left-to-right order. Let $\mset^F(x)= \set{(s_y^F(x),t_y^F(x))}_{y=1}^{5h+1}$ be a set of $(5h+1)$ demand pairs that we call the \false demand pairs for $x$. For each $1\leq y\leq 5h+1$, we identify $s_y^F(x)$ with the vertex $a^F_y$, and we let $t_y^F(x)$ be the $y$th vertex on $\hat I^F(x)$. 
 \end{itemize}

We let $\mset(x)=\mset^X(x)\cup \mset^T(x)\cup \mset^F(x)$ be the set of demand pairs representing $x$.
  
\begin{figure}[h]
\centering
\subfigure[A variable gadget. The $(5h)$ green vertices of $I^F(x)$ are partitioned into $5$ groups of $h$ consecutive vertices; each group is used by a different clause that contains the literal $x$. The green intervals of $I^T$ are dealt with similarly, but they are used by clauses containing $\neg x$. The vertices on the bottom boundary of $B'(x)$ are deleted.]
{\scalebox{0.35}{\includegraphics{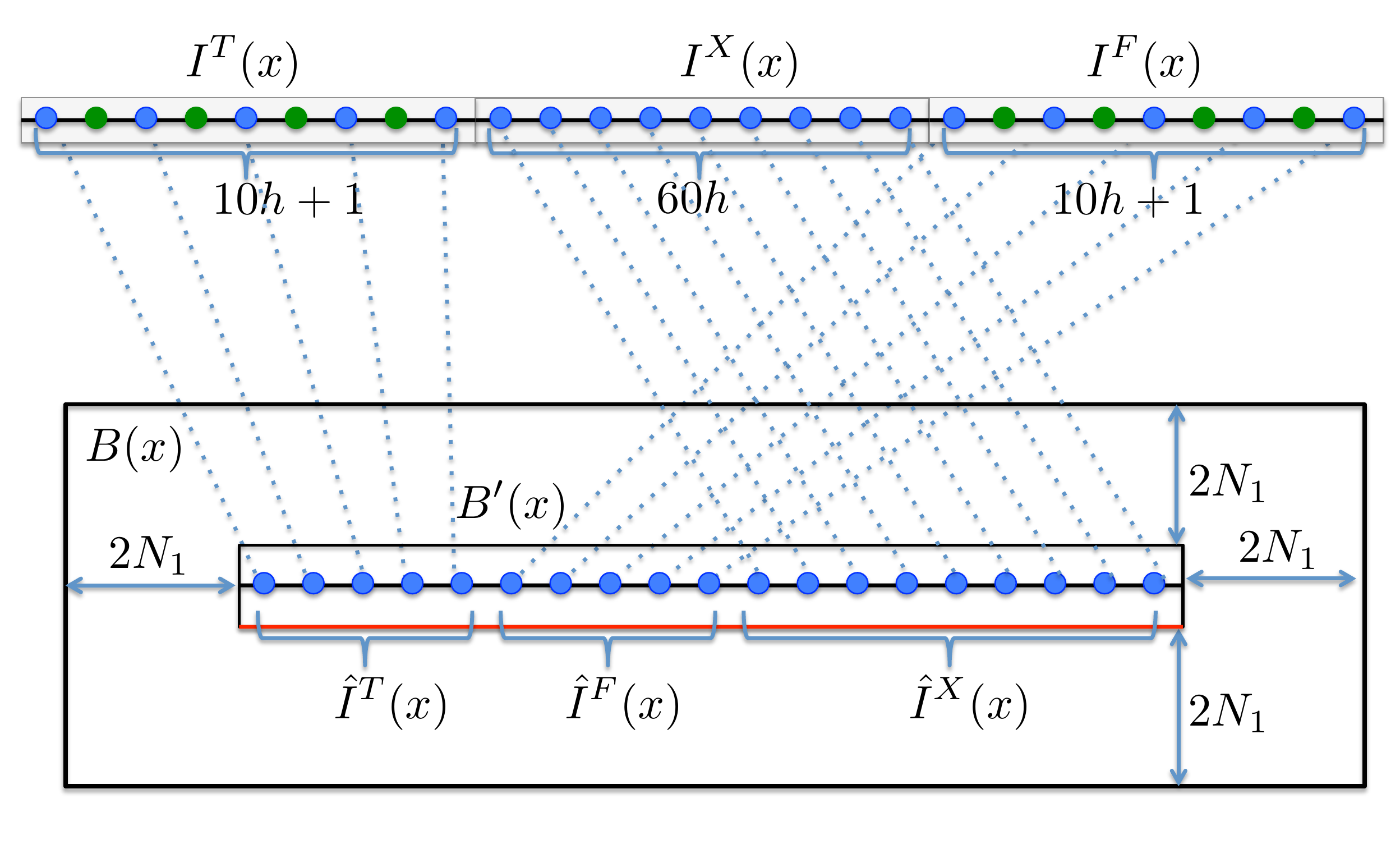}}\label{fig: level 1 var gadget}}
\hspace{0.5cm}
\subfigure[A clause gadget. Vertices of different colors correspond to different literals. The vertices on the bottom boundary of $B(C_q)$ are deleted.]{\scalebox{0.35}{\includegraphics{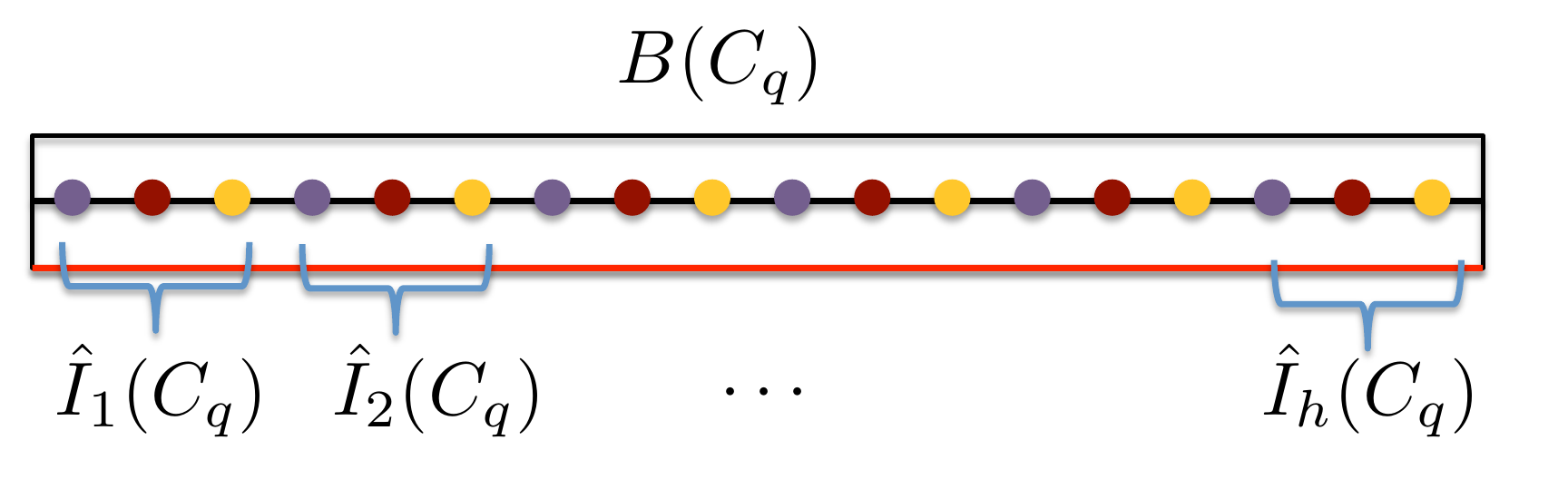}}
\label{fig: level 1 clause gadget}}
\caption{A variable gadget and a clause gadget \label{fig: level-1 var clause gadgets}}
\end{figure}

 Consider the set $\cset(x)\subseteq \cset$ of clauses in which variable $x$ appears without negation. Assume w.l.o.g. that $\cset(x)=\set{C_1,\ldots,C_r}$, where $r\leq 5$. For each $1\leq r'\leq r$, we will create $h$ demand pairs $\set{(s_j(C_{r'},x),t_j(C_{r'},x))}_{j=1}^h$, representing the literal $x$ of $C_{r'}$. Consider the interval $I^F(x)$. We will use its vertices $b_j^F$ as the sources of these demand pairs, where, intuitively, sources corresponding to the same clause-literal pair appear consecutively. Formally, for each $1\leq r'\leq r$ and $1\leq j\leq h$, we identify the vertex $s_j(C_{r'},x)$ with the vertex $b^F_{(r'-1)h+j}$  of $I^F(x)$. Intuitively,  if $x$ is assigned the value \false, then we will route all demand pairs in $\mset^F(x)$. The paths routing these pairs will ``block'' the vertices $b^F_j$, thus preventing us from routing demand pairs that represent clause-literal pairs $(C_{r'},x)$.
We treat the subset $\cset'(x)\subseteq \cset$ of clauses containing the literal $\neg x$ similarly, except that we identify their source vertices with the vertices of $\set{b_1^T,\ldots,b_{5h}^T}$ of $I^T(x)$. 


 \paragraph{ Clause Gadgets.}
 Consider some clause $C_q=(\ell_{q_1}\bor \ell_{q_2}\bor \ell_{q_3})$. For each one of the three literals $\ell\in \set{\ell_{q_1},\ell_{q_2},\ell_{q_3}}$ of $C_q$, let $\mset(C_q,\ell)=\set{(s_j(C_q,\ell),t_j(C_q,\ell))\mid 1\leq j\leq h}$ be a set of $h$ demand pairs representing the literal $\ell$ for clause $C_q$. 
 Recall that $B(C_q)$ is a grid of height $3$ and length $3h$. We delete all vertices that appear on the bottom row of this grid, and we let $R(C_q)$ be the middle row of $B(C_q)$. Partition $R(C_q)$ into $h$ intervals $\hat I_1(C_q),\ldots,\hat I_h(C_q)$, each containing three consecutive vertices (see Figure~\ref{fig: level 1 clause gadget}). 
%
%
Fix some $1\leq j\leq h$. We identify the three vertices of $\hat I_j(C_q)$ with the destination vertices $t_j(C_q,\ell_{q_1})$, $t_j(C_q,\ell_{q_2})$, and $t_j(C_q,\ell_{q_3})$ in this order. For each $1\leq z\leq 3$, the corresponding source vertex $s_j(C_q,\ell_{q_z})$ has already been defined as part of the definition of the variable gadget corresponding to the literal $\ell_{q_z}$. Let $\mset(C_q)=\bigcup_{z=1}^{3}\mset(C_q,\ell_{q_z})$ be the set of all demand pairs representing $C_q$, so $|\mset(C_q)|=3h$. Let $\mset^C=\bigcup_{q=1}^m\mset(C_q)$ be the set of all clause-pairs, and let $\mset^V=\bigcup_{j=1}^n\mset(x_j)$ be the set of all variable-pairs. Our final set of demand pairs is $\mset=\mset^V\cup \mset^C$. This concludes the definition of the level-$1$ instance. We now proceed to analyze it.

\paragraph{ Yes-Instance Analysis.}
%
Assume that $\phi$ is a \yi. We show that for every instantiation of the level-$1$ instance $\iset$, there is a set $\pset$ of node-disjoint paths routing $N_1=(200h/3+1)n$ demand pairs, that respect the box $B(\iset)$. Assume that we are given some instantiation of $\iset$. We first select the set $\hat\mset\subseteq\mset$  of demand pairs to route, and then define the routing. Fix some assignment $\aset$ to the variables of $\phi$ that satisfies all the clauses. For every variable $x$, if $\aset$ assigns the value \true to $x$, then we let $\hat\mset(x)=\mset^T(x)\cup\mset^X(x)$, and otherwise, we let $\hat\mset(x)=\mset^F(x)\cup\mset^X(x)$. Notice that in either case, $|\hmset(x)|=65h+1$. For each clause $C_q=(\ell_{q_1}\vee\ell_{q_2}\vee\ell_{q_3})\in\cset$, let $\ell_{q_z}$ be a literal which evaluates to \true by $\aset$ (if there are several such literals, select one of them arbitrarily). We let $\hat\mset(C_q)=\mset(C_q,\ell_{q_z})$. Let $\hat\mset^V=\bigcup_{j=1}^n\hat\mset(x_j)$ and $\hat\mset^C=\bigcup_{q=1}^m\hat\mset(C_q)$. We then set $\hmset=\hmset^V\cup \hmset^C$, so $|\hmset|=(65h+1)n+mh=(200h/3+1)n=N_1$, as $m=5n/3$. 
We now show that all demand pairs in $\hmset$ can be routed by a set $\pset$ of paths that respects the box $B(\iset)$. We only provide an intuitive description of the routing here; a formal proof appears in Section~\ref{sec: YI}.

 Consider some variable $x$, and assume that it is assigned the value \true. Let $\hmset'(x)\subseteq \hmset^C$ be the set of all clause-pairs, whose source vertices lie on the interval $I(x)$. Since the current assignment to $x$ must satisfy their corresponding clauses, all vertices of $S(\hmset'(x))$ lie on $I^F(x)$. Similarly, if $x$ is assigned the value \false, then all vertices of $S(\hmset'(x))$ lie on $I^T(x)$. Therefore, for every variable $x$, the sources of the demand pairs in $\hmset(x)$ appear consecutively on $Z(\iset)$ (relatively to the vertices of $S(\hmset)$), and the same holds for the sources of the demand pairs in $\hmset(C_q)$, for every clause $C_q$.
 
 We build the paths in $\pset$ gradually, growing them from the source vertices.
 We start by selecting, for every variable $x$, a set $A(x)$ of $|\hmset(x)|$ vertices on the top boundary of $B(x)$, and similarly, for every clause $C_q$, a set $A(C_q)$ of $|\hmset(C_q)|$ vertices on the top boundary of $B(C_q)$. We discuss this selection later.
  In the first step, we route the paths from their source vertices to these newly selected vertices, so that for every variable $x$, all paths originating from the  vertices of $S(\hmset(x))$ terminate at the vertices of $A(x)$ and they are order-preserving, and similarly  for every clause $C_q\in \cset$, all paths originating from the vertices of $S(\hmset(C_q))$ terminate at the vertices of $A(C_q)$ and they are order-preserving.
 In order to execute this step, we carefully select a set $\Gamma$ of $|\hmset|$ vertices on the top boundary of $B^V$; a set $\Gamma''$ of  $|\hmset^C|$ vertices on the bottom boundary of $B^V$; and the set $\Gamma'''$ of $|\hmset^C|$ left-most vertices on the top boundary of $B^C$.  We first connect all vertices in $S(\hmset)$ to the $|\hmset|$ leftmost vertices on the opening of $B(\iset)$ via a set of order-preserving node-disjoint paths, and then extend them to the vertices of $\Gamma$ via node-disjoint order-preserving paths inside $B(\iset)$. This part of the routing is straightforward. For every variable $x$, the paths originating at the vertices of $S(\hmset(x))$ then continue to their corresponding boxes $B(x)$, while for each clause $C_q$, the paths originating at the vertices of $S(\hmset(C_q))$ continue to the bottom boundary of $B^V$ and terminate at a consecutive set of vertices of $\Gamma''$ (see Section~\ref{sec: YI} and Figure~\ref{fig: yi-routing} for more details). We select the vertices of $\Gamma$ and $\Gamma''$ in a way that ensures that every path that we route inside the box $B^V$ is a sub-path of a column of $B^V$. We then connect the vertices of $\Gamma''$ to the vertices of $\Gamma'''$ via a set of node-disjoint order-preserving paths, that are internally disjoint from the boxes $B^V$ and $B^C$; these paths exploit the spacing between the two boxes. Finally, we complete the routing inside the box $B^C$, ensuring that for every clause $C_q$, the paths originating at the vertices of $S(\hmset(C_q))$ terminate at the vertices of $A(C_q)$. This is done via a standard snake-like routing. This routing critically uses the fact that the endpoints of the paths that we have constructed so far, which originate at the vertices of $S(\hmset(C_q))$, appear consecutively on $\Gamma'''$, for every clause $C_q$.
 
 By appropriately choosing, for every clause $C_q\in \cset$, the set $A(C_q)$ of vertices on the top boundary of $B(C_q)$, it is easy to extend the paths originating at the vertices of $S(\hmset(C_q))$, so that they terminate at the vertices in $T(\hmset(C_q))$. Since the resulting paths are order-preserving, we will route all demand pairs in $\hmset(C_q)$.
 
 We now consider some variable $x$ and show how to complete the routing of the demand pairs in $\hmset(x)$ inside the box $B(x)$. Assume first that $x$ is assigned the value \true.  Then the paths routing the demand pairs in $\hmset(x)$ arrive at the top boundary of $B(x)$ in the same order as the ordering of their source vertices on $Z(\iset)$. The order of their corresponding destination vertices on the second row of $B'(x)$ is identical, and so it is immediate to extend the paths inside $B(x)$ to complete the routing.
 
 Assume now that $x$ is assigned the value \false. Let $J$ and $J'$ be the intervals of the top boundary of $B(x)$ where the paths routing the pairs in $\hmset^F(x)$ and $\hmset^X(x)$ arrive, respectively. Unfortunately, $J$ lies to the right of $J'$, while interval $\hat I^F(x)$ lies to the left of the interval $\hat I^X(x)$ on the second row of $B'(x)$. Therefore, we need to ``switch'' the ordering of these two sets of paths before we can complete the routing. It is easy to do so by exploiting the ample spacing between the box $B'(x)$ and the boundaries of the box $B(x)$ (see Figure~\ref{fig: var-routing}).

\noindent {\bf No-Instance Analysis.}
Assume now that $\phi$ is a \ni, and that we are given some instantiation of the level-$1$ instance $\iset$ and a set $\tpset$ of node-disjoint paths routing some subset $\tmset\subseteq \mset$ of demand pairs. Our goal is to prove that $|\tmset|\leq N'_1=(1-\delta)N_1=(1-\delta)(200h/3+1)n$. Assume for contradiction that $|\tmset|>N'_1$. 
In order to analyze the \ni, it is convenient to view the construction slightly differently. Let $\cset'$ be the set of clauses obtained by adding, for each clause $C_q\in \cset$, $h$ copies $C^1_q,\ldots,C^h_q$ of $C_q$ to $\cset'$. We will refer to the clauses in $\cset$ as the \emph{original clauses}, and the clauses in $\cset'$ as the \emph{new clauses}. Notice that $|\cset'|=mh$, and it is easy to verify that no assignment to the variables of $\phi$ can satisfy more than $(1-\eps)hm$ clauses of $\cset'$. We will reach a contradiction by defining an assignment to the variables of $\phi$ that  satisfies more than $(1-\eps)hm$ clauses of $\cset'$. For each new clause $C^j_q\in \cset'$, we let $\mset(C^j_q)\subseteq \mset$ be  the set of all demand pair whose destinations lie on interval $\hat I_j(C_q)$; we view these demand pairs as representing the clause $C^j_q$.

For every variable $x$ of $\phi$, let $\tilde\mset(x)=\tilde\mset\cap\mset(x)$, and let $\tilde\mset^T(x),\tilde\mset^F(x),\tilde\mset^X(x)$ denote  $\mset^T(x)\cap \tmset,\mset^F(x)\cap \tmset$ and $\mset^X(x)\cap \tmset$, respectively. For every new clause $C^j_q\in \cset'$, let $\tilde\mset(C_q^j)=\tmset\cap \mset(C^j_q)$. We also denote by $\tilde\mset^V=\bigcup_{j=1}^n\tilde\mset(x_j)$ and by $\tilde\mset^C=\bigcup_{C^j_q\in \cset'}\tilde\mset(C^j_q)$, the sets of the variable-pairs and the clause-pairs, respectively, that are routed by $\tpset$. 
We use the following claim.

\begin{claim}\label{claim: lvl1 var gadget} For each variable $x$ of $\phi$, at least one of the sets $\tilde\mset^T(x),\tilde\mset^F(x),\tilde\mset^X(x)$ is empty.
\end{claim}

The proof is omitted here; we prove a somewhat stronger claim in Section~\ref{sec: NI} (see also Observation~\ref{obs: cylinder} from which the claim follows immediately). Notice that from the above claim, $|\tmset^V|\leq (65h+1)n$.

Consider now some variable $x$. Assume first that $\tmset^F(x)=\emptyset$. We then assign $x$ the value \true. We say that an index $1\leq j\leq 5h+1$ is \emph{bad} for $x$, iff the pair $(s_j^T(x),t_j^T(x))\not \in \tmset(x)$. Otherwise, if $\tmset^F(x)\neq \emptyset$, then we assign $x$ the value \false. In this case, we say that an index $1\leq j\leq 5h+1$ is bad for $x$, iff $(s_j^F(x),t_j^F(x))\not \in \tmset(x)$. We later show that the total number of pairs $(x,j)$, where $j$ is a bad index for variable $x$, is small.

Consider some new clause $C_q^j$. We say that it is an \emph{interesting} clause if $|\tmset(C_q^j)|\geq 1$ (in other words, at least one pair in the set $\set{(s_j(C_q,\ell_{q_z}),t_j(C_q,\ell_{q_z}))}_{z=1}^3$ is in $\tmset$), and we say that it is uninteresting otherwise. We say that clause $C_q^j$ is \emph{troublesome} iff $|\tmset(C_q^j)|>1$. The proof of the following simple observation is omitted here; we prove a more general statement in Section~\ref{sec: NI}.

\begin{observation}\label{obs: problematic level 1}
For each clause $C_q$, at most three of its copies are troublesome.
\end{observation}

 We conclude that $|\tmset^C|\leq m(h+6)= 5n(h+6)/3$. 
Let $\cset_1'\subseteq \cset'$ be the set of all interesting new clauses. A simple accounting shows that, if $|\tmset|\geq (1-\delta)(200h/3+1)n$, then $|\cset_1'|\geq (1-\eps/10) hm$ must hold. Notice that for each new clause $C_q^j\in \cset_1'$, at least one demand pair from the set $\set{(s_j(C_q,\ell_{q_z}),t_j(C_q,\ell_{q_z}))}_{z=1}^3$ is in $\tmset$. We select any literal $\ell\in \set{\ell_{q_1},\ell_{q_2},\ell_{q_3}}$ such that 
$(s_j(C_q,\ell),t_j(C_q,\ell))\in \tmset$, and we say that clause $C_q^j$ \emph{chooses} the literal $\ell$. Let $x$ be the variable corresponding to the literal $\ell$. We say that $C_q^j$ is a \emph{cheating} clause iff the assignment that we chose for $x$ is not consistent with the literal $\ell$: that is, if $\ell=x$, then $\aset(x)=\mbox{\false}$, and if $\ell=\neg x$, then $\aset(x)=\mbox{\true}$. Notice that, if $C_q^j$ is an interesting and a non-cheating clause, then the current assignment satisfies $C_q^j$. Therefore, in order to compete the analysis, it is enough to prove the following claim.

\begin{claim}\label{claim: cheating level 1}
The number of cheating clauses in $\cset'_1$ is bounded by $\eps mh/2$.
\end{claim}

We prove a stronger claim in Section~\ref{sec: NI} (see Lemma~\ref{lem: cheating clauses}), and provide a proof sketch here. Let $C^j_q\in \cset'$ be a cheating clause, and suppose it has chosen the literal $\ell$, whose corresponding variable is $x$. We say that $C^j_q$ is a \emph{bad} cheating clause, iff at least one of the indices $j,j+1$ is a bad index for variable $x$ (recall that $j$ is a bad index for $x$ if $\aset(x)=\mbox{\true}$ and $(s_j^T(x),t_j^T(x))\not \in \tmset$, or $\aset(x)= \mbox{\false}$ and $(s_j^F(x),t_j^F(x))\not \in \tmset$). Otherwise, we say that $C_q^j$ is a good cheating clause. A simple accounting shows that the number of pairs $(x,j)$, where $j$ is a bad index for $x$ is bounded by $\eps mh/16$. Each such pair $(x,j)$ may contribute to at most two bad cheating clauses, and so there are at most   $\eps mh/8$ bad cheating clauses.

Our final step is to show that the number of good cheating clauses is bounded by $\eps mh/4$. We show that for each original clause $C_q$, at most $3$ copies of $C_q$ are good cheating clauses. It then follows that the total number of good cheating clauses is at most $3m<\eps mh/4$, since $h=1000/\eps$. Consider some original clause $C_q$. It is enough to show that for each literal $\ell$ of $C_q$, the number of copies of $C_q$ that choose $\ell$ and are good cheating clauses is at most $1$. Assume for contradiction that there are two such copies $C_q^j$ and $C_q^{j'}$. Assume w.l.o.g. that the variable $x$ that corresponds to $\ell$ is assigned the value \true, so $\ell=\neg x$. Then vertex $s_j(C_q,\ell)$ lies on interval $I^T(x)$, between $s_j^T(x)$ and $s_{j+1}^T(x)$, while vertex $s_{j'}(C_q,\ell)$ lies on interval $I^T(x)$, between $s_{j'}^T(x)$ and $s_{j'+1}^T(x)$. Assume w.l.o.g. that $j<j'$. Consider the plane with only the top boundary of the grid $G_1$, the bottom boundary of the box $B'(x)$, and the images of the paths of $\tpset$ routing the pairs $(s_j^T(x),t_j^T(x)),(s_{j+1}^T(x),t_{j+1}^T(x))$, and $(s_{j'+1}^T(x),t_{j'+1}^T(x))$ present. Let $f,f'$ be the two faces of this drawing that differ from the outer face, such that $f$ has $s_j(C_q,\ell)$ on its boundary and $f'$ has $s_{j'}(C_q,\ell)$ on its boundary. Then $f\neq f'$, and the bottom boundary of $B(C_q)$ must belong to a single face of the resulting drawing. Assume w.l.o.g. that this face is $f^*\neq f$. Then $t_j(C_q,\ell)$ lies on $f^*$, and so it is impossible that a path of $\tpset$ connects $s_j(C_q,\ell)$ to $t_j(C_q,\ell)$. We conclude that the current assignment satisfies at least $(1-\eps/10)hm-\eps hm/2>(1-\eps)hm$ clauses of $\cset'$, a contradiction.

\paragraph{ Generalization to Higher Levels and the Hardness Gap.}
Assume now that we are given a construction of a level-$i$ instance, and we would like to construct a level-$(i+1)$ instance. Intuitively, we would like to start with the level-$1$ instance described above, and then replace each source-destination pair $(s,t)$ with a distinct copy of a level-$i$ instance $\iset'$. So we would replace the vertex $s$ with the path $Z(\iset')$, and the vertex $t$ with the cut-out box $B(\iset')$. We say that a level-$i$ instance $\iset'$ is routed by a solution $\pset$ to this resulting instance, iff $\pset$ routes a significant number of the demand pairs in $\iset'$.
The idea is that, due to the level-$1$ instance analysis, the number of such level-$i$ instances that can be routed in the {\sc Yes-} and the \ni cases differ by a constant factor, while within each such instance we already have some gap $g_i$ between the {\sc Yes-} and the \ni solutions, and so the gap grows by a constant factor in every level. Unfortunately this idea does not quite work. If we consider, for example, a level-$1$ instance $\iset'$, then its destination vertices appear quite far -- at distance $\Theta(N_1)$ -- from the bottom boundary of the box $B(\iset')$. In general, in a level-$i$ instance, this distance needs to be roughly $\Theta(N_i)$, to allow the routing in the \yi case (recall that $N_i$ is the number of the demand pairs that can be routed in the \yi case). Therefore, if we replace each level-$1$ demand pair by a level-$i$ instance, some of the paths in the new level-$(i+1)$ instance may ``cheat'' by passing through the boxes $B(\iset')$ of level-$i$ instances $\iset'$, and exploiting the spacing between the destination vertices and the bottom boundary of each such box. For example, it is now possible that in a variable gadget, we will be able to route many demand pairs from each set $\mset^X(x),\mset^T(x)$ and $\mset^F(x)$ simultaneously. A simple way to get around this problem is to create more level-$i$ instances, namely: we replace each source-destination pair from a level-$1$ instance by a collection of $c_{i+1}$ level-$i$ instances. The idea is that, if the number of the demand pairs we try to route in many such $c_{i+1}$-tuples of level-$i$ instances is large enough, then on average only a small fraction of the routing paths may cheat by exploiting the spacing between the destination vertices and their corresponding box boundaries, and this will not affect the overall accounting by too much. However, if the formula $\phi$ is a \ni, then we will only attempt to route $N'_i$ demand pairs from each level-$i$ instance, and therefore we need to ensure that $c_{i+1}N'_i\gg N_i$ in order for the gap to grow in the current level. In other words, the number of copies of the level-$i$ instances that we use in the level-$(i+1)$ instance construction should be proportional to the gap between the {\sc Yes-} and the \ni cost at level $i$ (times $n$). A simple calculation shows that, if we follow this approach, we will obtain a gap of $2^{\Omega(i)}$ in level-$i$ instances, with construction size roughly $n^{\Theta(i)}\cdot 2^{\Theta(i^2)}$. Therefore, after $i^*=\Theta(\log n)$ iterations, we obtain a gap of $2^{\Omega(\sqrt{\log n'})}$, where $n'$ is the size of the level-$i^*$ instance. This rapid growth in the instance size is the main obstacle to obtaining a  stronger hardness of approximation factor using this approach.

\label{------------------------------------------------sec: level i--------------------------------------}
\section{The Full Construction}\label{sec: level i}

In this section we provide a full description of our construction. The resulting graphs will have maximum vertex degree $4$. We show in Section~\ref{sec: EDP} of the Appendix how to modify the resulting instances in order to obtain the proof of Theorem~\ref{thm: main} for sub-cubic graphs.
We start with setting the parameters.

\paragraph{Parameters.}
The two main parameters that we use are $h=1000/\eps$ and $\delta=8\eps^2/10^{12}$, where $\eps$ is the constant from Theorem~\ref{thm: PCP}.
We define the remaining parameters in terms of these two parameters.

For every level $i\geq 0$ of our construction, we use two parameters, $N_i$ and $N'_i$. We will ensure that for every instantiation of the level-$i$ instance $\iset$, if the initial 3SAT(5) formula $\phi$ is a \yi, then  there is a solution to $\iset$ routing $N_i$ demand pairs, that respects the box $B(\iset)$, and if $\phi$ is a \ni, then no solution to $\iset$ can route more than $N'_i$ demand pairs. We define the parameters $N_i,N'_i$ inductively, starting with $N_0=N'_0=1$. Assume now that for some $i\geq 0$, we are given the values of $N_i$ and $N_{i}'$. Let $g_i=N_i/N_i'$ be the gap between the {\sc Yes}- and the \ni solution values at level $i$, and let $c_{i+1}=10^8h^2g_i=O(g_i)$. Parameter $c_{i+1}$ will be used in our construction of the level-$(i+1)$ instance. We then set $N_{i+1}=n c_{i+1} (200h/3+1)N_i$ and $N_{i+1}'=(1-\delta)n c_{i+1} (200h/3+1)N'_i$. It is immediate to verify that $g_{i+1}=\frac{g_i}{1-\delta}$, and so for all $i\geq 0$, $g_i=\left(\frac{1}{1-\delta}\right)^i$, and $N_i=O(n\cdot g_{i-1})\cdot N_{i-1}=(\rho n)^i\cdot 2^{O(i^2)}$, for some absolute constant $\rho$. We set the parameters $L_i,L'_i$ and $H_i$ below, but we will ensure that each of these parameters is bounded by $20N_i^3$. Our construction has $i^*=\log n$ levels, giving us a gap of $2^{\Omega(\log n)}$ between the {\sc Yes}- and the \ni solution values. For our final level-$i^*$ instance, we can choose the grid $G_{i^*}$ to be of size $(Q\times Q)$, where $Q=2L_{i^*}+2L'_{i^*}+4H_{i^*}=O(N_{i^*}^3)$, and so the instance size is bounded by $n'=O(N_{i^*}^6)=n^{O(\log n)}\cdot 2^{O(\log^2 n)}=2^{O(\log^2n)}$. Overall, we obtain a factor $2^{\Omega(\sqrt{\log n'})}$-hardness of approximation, unless all problems in \NP have deterministic algorithms running in time $n^{O(\log n)}$.

For $i\geq 0$, we set the parameter $H_i=20N_i$. The following bound on $H_i$ follows immediately from the definitions of our parameters, and we use it several times in our analysis:

\begin{equation}
H_i=20N_i=20 g_i N'_i=\frac{20 c_{i+1}N'_i}{10^8h^2}=\frac{2c_{i+1}\eps^2N'_i}{10^{13}} \label{bound on H_i}
\end{equation}

For all $i\geq 0$, we set $L'_i=20N_i^3$. Parameter $L_i$ is defined as follows: $L_0=1$, and for $i>0$, $L_i=(80h+2)c_iL_{i-1}n\leq  (80h+2)c_{i}\cdot 20N_{i-1}^3n\leq 20N_i^3$.

\paragraph{Level-$0$ Instance.}
A level-$0$ instance $\iset$ consists of a single demand pair $(s,t)$. In order to be consistent with our definitions, we let $Z(\iset)$ be a path of length $L_0=1$, with $s$ mapped to the unique vertex of $Z(\iset)$. Recall that $N_0=N'_0=1$, $H_0=20N_0=20$, and $L'_0=20N_0^3=20$. Let $G'_0$ be a grid of length $L'_0=20$ and height $H_0=20$. We obtain the box $B(\iset)$ from $G'_0$ by deleting all vertices lying on the left, bottom, and right boundaries of $G'_0$.  Let $R'$ be the middle row of $G'_0$. We map $t$ to any vertex of $B(\iset)$ that belongs to $R'$. It is immediate to verify that for every instantiation of this level-$0$ instance, there is a solution that routes one demand pair and respects $B(\iset)$, regardless of whether we are in the {\sc Yes} or the \ni. 

From now on we focus on constructing instances of levels $i>0$. As already mentioned, the construction we obtain for level $i=1$ is similar to that described in Section~\ref{sec: level 1}.

\subsection{Level-$(i+1)$ Construction}

A level-$(i+1)$ instance is obtained by combining a number of level-$i$ instances.
We start with a quick overview of the level-$i$ construction.

\paragraph{Level-$i$ Construction Overview.}
Recall that a definition of a level-$i$ instance $\iset$ consists of a path $Z(\iset)$ of length $L_i$, a grid $G'_i$ of height $H_i$ and length $L'_i$, together with a cut-out box $B(\iset)\subseteq G'_i$, and a collection $\mset$ of demand pairs, such that all vertices of $S(\mset)$ are mapped to vertices of $Z(\iset)$, while all vertices of $T(\mset)$ are mapped to distinct vertices of $B(\iset)\cap R'$, where $R'$ is the middle row of $G'_i$. Path $Z(\iset)$ will eventually become a sub-path of the first row of the larger grid $G_i$, and box $B(\iset)$ will be placed inside $G_i$, within distance at least $H_i$ from its boundaries.

For every destination vertex $t$, we draw a straight line $Q_t$ from $t$ to the bottom of $B(\iset)$. This line contains at most $H_i/2$ vertices of the graph. We will use these lines in the analysis of the \ni case of the level-$(i+1)$ construction. 

It will sometimes be useful to place several level-$i$ instances side-by-side. For an integer $c>0$, a \emph{$c$-wide level-$i$ instance} $\iset$ is constructed as follows. Intuitively, we construct $c$ disjoint level-$i$ instances $\iset_1,\ldots,\iset_c$, placing their intervals $Z(\iset_j)$ side-by-side on $Z(\iset)$ and placing their boxes $B(\iset_j)$ side-by-side inside $B(\iset)$. Formally, for each $1\leq j\leq c$, let $\mset_j$ be the set of the demand pairs of the level-$i$ instance $\iset_j$, and let $G^j$ be the corresponding grid $G'_i$ for that instance. The set of the demand pairs of instance $\iset$ is $\mset=\bigcup_{j=1}^c\mset_j$.
We let $Z(\iset)$ be a path of length $c\cdot L_i$, partitioned into $c$ equal-length intervals $A_1,\ldots,A_c$. We let $G'$ be a grid of length $cL'_i$ and height $H_i$, that we partition into $c$ sub-grids of length $L'_i$ and height $H_i$ each. For $1\leq j\leq c$, we map the vertices of $Z(\iset_j)$ to the vertices of $A_j$ in a natural way. This defines the mapping of the vertices of $S(\mset)$ to the vertices of $Z(\iset)$. For each $1\leq j\leq c$, we map the vertices of $G^j$ to the $j$th sub-grid of $G'$, and delete from $G'$ all vertices to which the vertices of $G^j\setminus B(\iset_j)$ are mapped. The resulting subgraph of $G'$ becomes the box $B(\iset)$, and the above mapping defines the mapping of the destination vertices in $T(\mset)$ to the vertices of $B(\iset)$. Note that if $R'$ denotes the  middle row of $G'$, then all vertices of $T(\mset)$ lie on $R'$. 
In order to instantiate this instance, we need to select a grid $G$ of length at least $c(2L_i+2L'_i+4H_i)$ and height at least $3cH_i$, a sub-path $P$ of the first row of $G$ of length $cL_i$, to which $Z(\iset)$ will be mapped, and a sub-grid $G''$ of $G$ of the same dimensions as $G'$, to which the vertices of $G'$ will be mapped. The vertices of $G''$ must be at a distance at least $cH_i$ from the boundaries of $G$.
Clearly, for any instantiation of this instance, in the \yi case, there is a solution $\pset$ routing $cN_i$ demand pairs, such that, if we denote, for each $1\leq j\leq c$, by $\pset_j\subseteq \pset$ the set of paths routing demand pairs in $\mset_j$, then $|\pset_j|=N_i$ and $\pset_j$ respects the box $B(\iset_j)$. On the other hand, in the \ni case, no solution to $\iset$ can route more than $cN'_i$ demand pairs.

We now assume that we are given a construction of a level-$i$ instance, for $i\geq 0$, and describe a construction of a level-$(i+1)$ instance $\iset$. 
For convenience, we denote $c_{i+1}$ by $c$. 
 We use parameters $L_{i+1},H_{i+1}, L'_{i+1}$ described above, so $H_{i+1}=20N_{i+1}$, $L'_{i+1}=20N_{i+1}^3$, and $L_{i+1}=(\vargadgetwidth)cL_in$. 

 In order to construct the box $B(\iset)$, we start with a grid $G'_{i+1}$ of length $L'_{i+1}$ and height $H_{i+1}$. 
 We define two sub-grids of $B(\iset)$, of length $9N_{i+1}^3$ and height $16N_{i+1}$ each: grid $B^V$ that will contain all destination vertices of the demand pairs representing the variables of the formula $\phi$, and grid $B^C$ that will contain all destination vertices of the demand pairs representing the clauses of the formula $\phi$. In order to construct both grids, let $\rset$ be the set of all rows of $G'_{i+1}$, excluding the top $2N_{i+1}$ and the bottom $2N_{i+1}$ rows, so that $|\rset|=H_{i+1}-4N_{i+1}=16N_{i+1}$. Let $\wset$ be a consecutive set of $9N_{i+1}^3$ columns of $G'_{i+1}$, starting from the second column, and let $\wset'$ be a consecutive set of $9N_{i+1}^3$  columns of $G'_{i+1}$, terminating at the penultimate column. We then let $B^V$ be the sub-grid of $G_{i+1}'$ spanned by the rows in $\rset$ and the columns in $\wset$, and we let $B^C$ be the sub-grid of $G_{i+1}'$ spanned by the rows in $\rset$ and the columns in $\wset'$ (see Figure~\ref{fig: whole box}). Notice that at least $2N_{i+1}$ columns of $G_{i+1}'$ separate the two grids. We delete the bottom, left, and right boundaries of $G_{i+1}'$ to turn it into a cut-out box, that we refer to as $B(\iset)$ from now on. We will later delete some additional vertices from $B(\iset)$.
 
\begin{figure}[h]
\scalebox{0.6}{\includegraphics{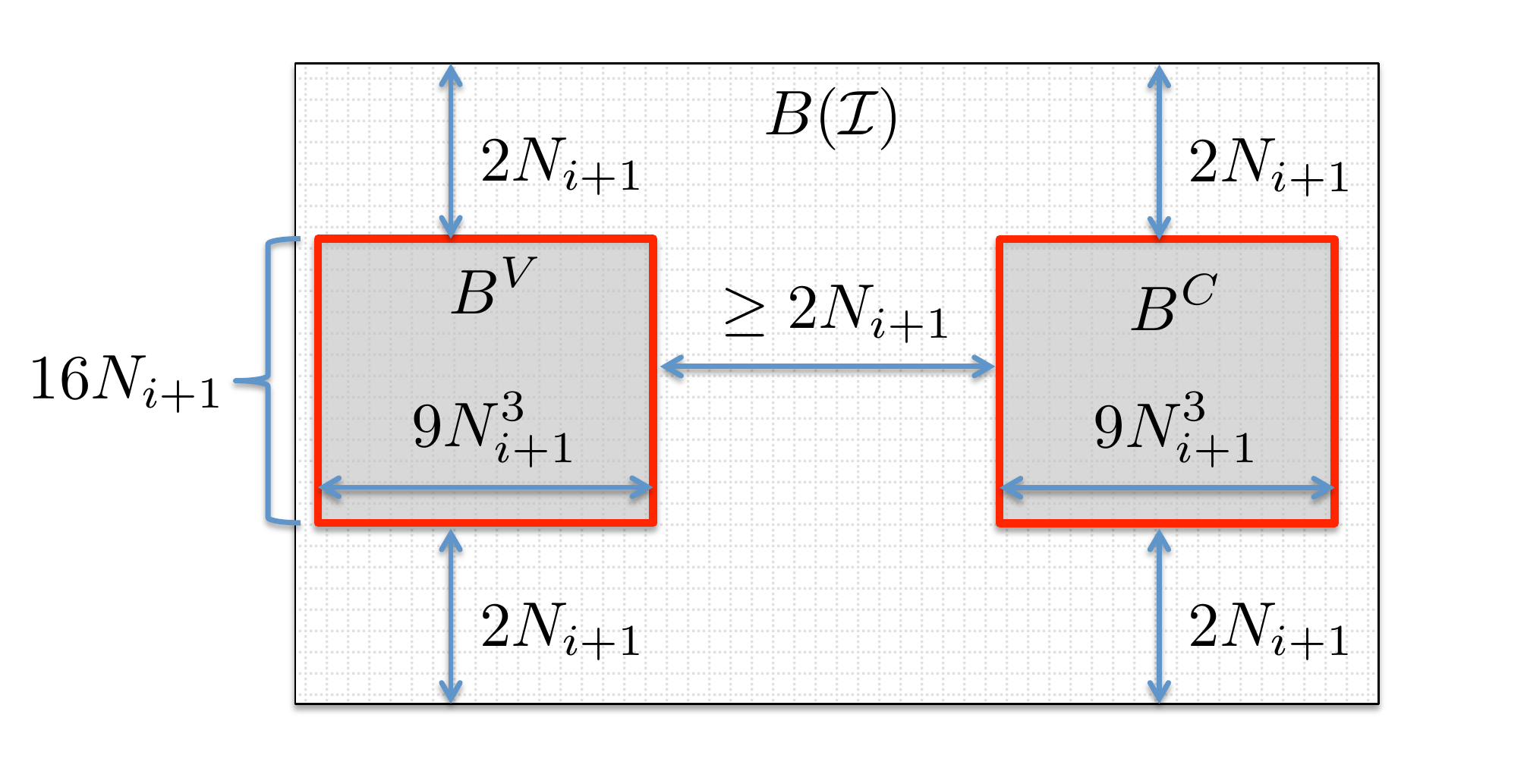}}
\caption{High-level view of box $B(\iset)$\label{fig: whole box}}
\end{figure} 
 
Next, we define smaller sub-grids of the two grids $B^V$ and $B^C$. For every variable $x_j$ of $\phi$, we define a sub-grid $B(x_j)$ of $B^V$, of length $L^V=4N_{i+1}+(70h+2)cL'_i$ and height $H^V=H_i+2N_{i+1}$. This box will contain all destination vertices of the demand pairs that represent the variable $x_j$. We place the boxes $B(x_1),\ldots,B(x_n)$ inside grid $B^V$, so that they are aligned and $2N_{i+1}$-separated. In other words, the middle row of each box is contained in the middle row of $B^V$, and the horizontal distance between every pair of these boxes, and between each box and the left and right boundaries of $B^V$ is at least $2N_{i+1}$. 
Since $n\cdot L^V+(n+1)\cdot 2N_{i+1}\leq 7nN_{i+1}+(70h+2)cL'_in\leq 7nN_{i+1}+1500hc N_i^3n<9N_{i+1}^3$, we can find such grids $B(x_1),\ldots,B(x_n)$. Since $H^V=H_i+2N_{i+1}=20N_i+2N_{i+1}<3N_{i+1}$, there are at least $2N_{i+1}$ rows of $B^V$ above and below these new grids (see Figure~\ref{fig: big var box}).

\begin{figure}[h]
\scalebox{0.6}{\includegraphics{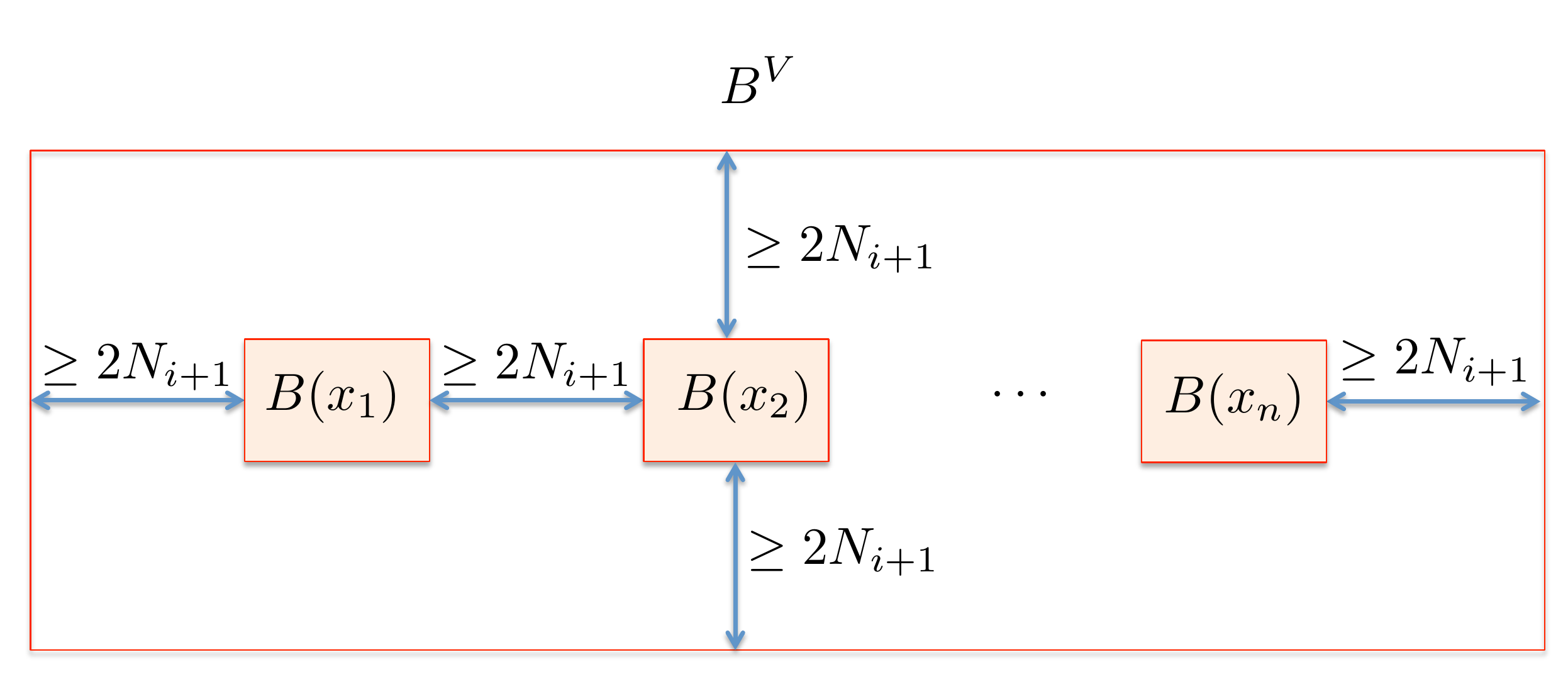}}
\caption{Box $B^V$. Each box $B(x_j)$ has length $L^V=4N_{i+1}+(70h+2)cL'_i$ and height $H^V=H_i+2N_{i+1}$.  \label{fig: big var box}}
\end{figure}

We similarly define sub-grids $B(C_1),\ldots,B(C_m)$ of $B^C$. Each such sub-grid has height $H^C=H_i$ and width $L^C=3chL'_i$. For each clause $C_q\in \cset$, box $B(C_q)$ will contain all destination vertices of the demand pairs that represent this clause. 
We let $B(C_1),\ldots,B(C_m)$ be sub-grids of $H^C$ that are aligned and $4N_{i+1}$-separated. Since $m\cdot L^C+(m+1)\cdot 4N_{i+1}\leq 20nN_{i+1}+15hcL'_in\leq 20nN_{i+1}+300hc N_i^3n<9N_{i+1}^3$, we can find such grids (see Figure~\ref{fig: big clause box}). 
 Since $H^C=H_i=20N_i<N_{i+1}$, there are at least $2N_{i+1}$ rows of $B^C$ above and below these new grids. 

\begin{figure}[h]
\scalebox{0.6}{\includegraphics{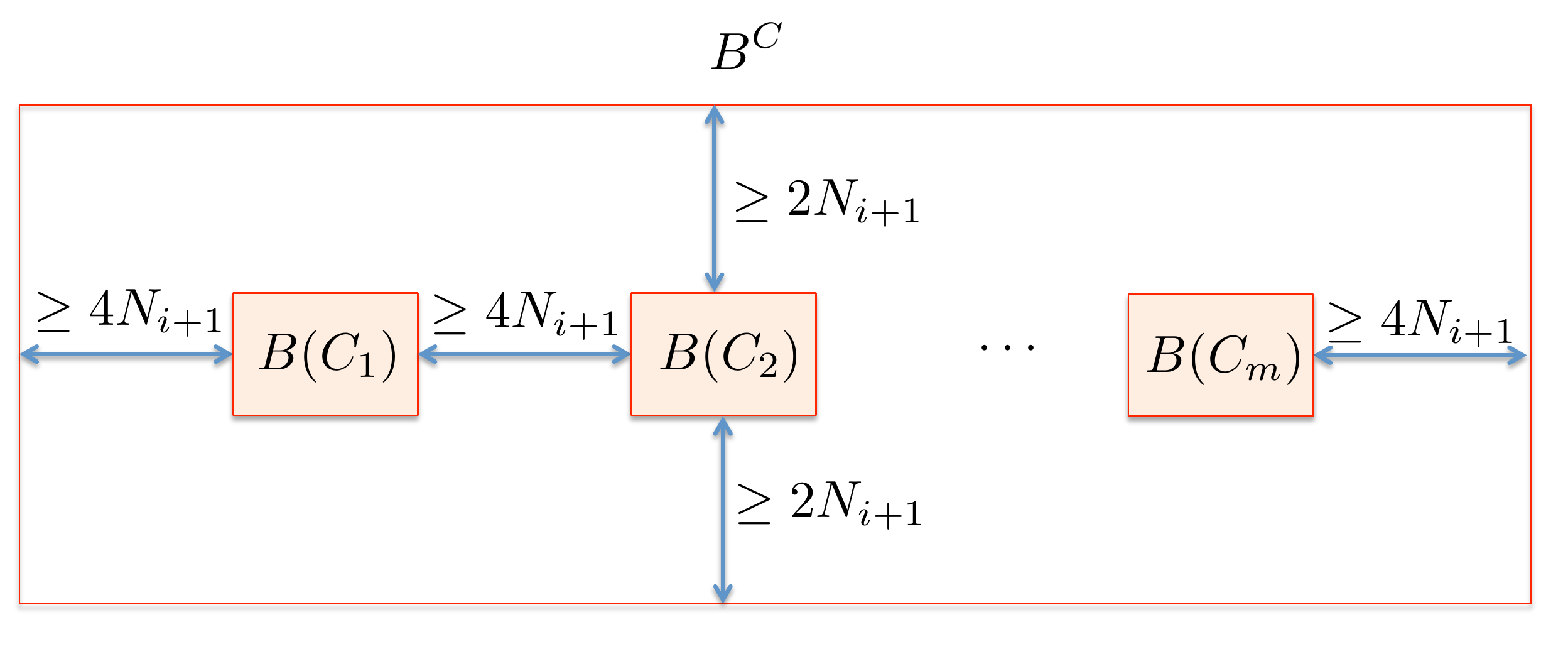}}
\caption{Box $B^C$. Each box $B(C_q)$ has length $L^C=3chL'_i$ height $H^C=H_i$. \label{fig: big clause box}}
\end{figure}

Our construction consists of two parts, called variable gadgets and clause gadgets. For each variable $x_j$, we construct a number of level-$i$ instances $\iset'$, whose corresponding boxes $B(\iset')$ are placed inside $B(x_j)$. Whenever we do so, we delete the corresponding vertices of $B(x_j)$ as described in the preliminaries. We also construct clause gadgets similarly.

\paragraph{Variable Gadgets}
Let $Z(\iset)$ be a path of length $L_{i+1}$, and let $\Pi$ be a partition of $Z(\iset)$ into disjoint contiguous sub-paths (that we sometimes refer to as \emph{intervals}) of length $cL_i$ each. For each $1\leq j\leq n$, we let $I(x_j)$ be a sub-path of $Z(\iset)$, containing exactly $\vargadgetwidth$ consecutive intervals of $\Pi$, so that $I(x_1),I(x_2),\ldots,I(x_n)$ appear on $Z(\iset)$ in this left-to-right order.

Consider some variable $x$ of the 3SAT(5) formula $\phi$ and the corresponding interval $I(x)$ of $Z(\iset)$, containing $\vargadgetwidth$ consecutive intervals of $\Pi$. We further partition $I(x)$ as follows. Let $I^T(x), I^F(x)\subseteq I(x)$ denote the subpaths of $I(x)$ containing the first $(10h+1)$ and the last $(10h+1)$ consecutive intervals of $\Pi$, respectively. Let $I^X(x)$ denote the union of the remaining $60h$ consecutive intervals  of $\Pi$ (see Figure~\ref{fig: level 2 var gadget}).

 \begin{itemize}
 
 \item {\bf (Extra Pairs).} We use $60h$ copies of $c$-wide level-$i$ instances, that we denote by $\iset^X_j(x)$, for $1\leq j\leq 60h$. For each $1\leq j\leq 60h$, we let $Z(\iset_j^X(x))$ be the $j$th interval of $I^X(x)$. We place the corresponding boxes $B(\iset^X_1(x)),\ldots,B(\iset^X_{60h}(x))$ side-by-side, obtaining one box $B_X(x)$ of width $60hcL_i'$ and height $H_i$. We define the placement of this box inside $B(x)$ later. We denote by $\mset^X(x)$ the resulting set of demand pairs, and we refer to them as \emph{the \extra demand pairs of $x$}.
 

 \item  {\bf (\true Pairs).} 
We denote the intervals of $\Pi$ appearing on $I^T(x)$ by: $$A^T_1,Y^T_1,A^T_2,Y^T_2,\ldots, Y^T_{5h},A^T_{5h+1},$$ and we assume that they appear on $I^T(x)$ in this left-to-right order.
 We use  $(5h+1)$ copies of the $c$-wide level-$i$ instance, that we denote by $\iset^T_j(x)$, for $1\leq j\leq 5h+1$. For each $1\leq j\leq 5h+1$, we let $Z(\iset_j^T)$ be the interval $A_j^T$. We place the corresponding boxes $B(\iset^T_1(x)),\ldots,B(\iset^T_{5h+1}(x))$ side-by-side, obtaining one box $B_T(x)$ of width $(5h+1)cL_i'$ and height $H_i$. We denote by $\mset^T(x)$ the resulting set of demand pairs, and we refer to them as \emph{the \true demand pairs of $x$}.

  \item  {\bf (\false Pairs).} 
Similarly, we denote the intervals of $\Pi$ appearing on $I^F(x)$ by:

 $$A^F_1,Y^F_1,A^F_2,Y^F_2,\ldots, Y^F_{5h},A^F_{5h+1},$$ 
 
 and we assume that they appear on $I^F(x)$ in this left-to-right order.
 We use  $(5h+1)$ copies of the $c$-wide level-$i$ instance, that we denote by $\iset^F_j(x)$, for $1\leq j\leq 5h+1$. For each $1\leq j\leq 5h+1$, we let $Z(\iset_j^F)$ be the interval $A_j^F$. We place the corresponding boxes $B(\iset^F_1(x)),\ldots,B(\iset^F_{5h+1}(x))$ side-by-side, obtaining one box $B_F(x)$ of width $(5h+1)cL_i'$ and height $H_i$. We denote by $\mset^F(x)$ the resulting set of demand pairs, and we refer to them as \emph{the \false demand pairs of $x$}.
 \end{itemize}


We let $\mset(x)=\mset^X(x)\cup \mset^T(x)\cup \mset^F(x)$. We call the demand pairs in $\mset(x)$ \emph{variable-pairs representing $x$}.

\begin{figure}[h]
\scalebox{0.6}{\includegraphics{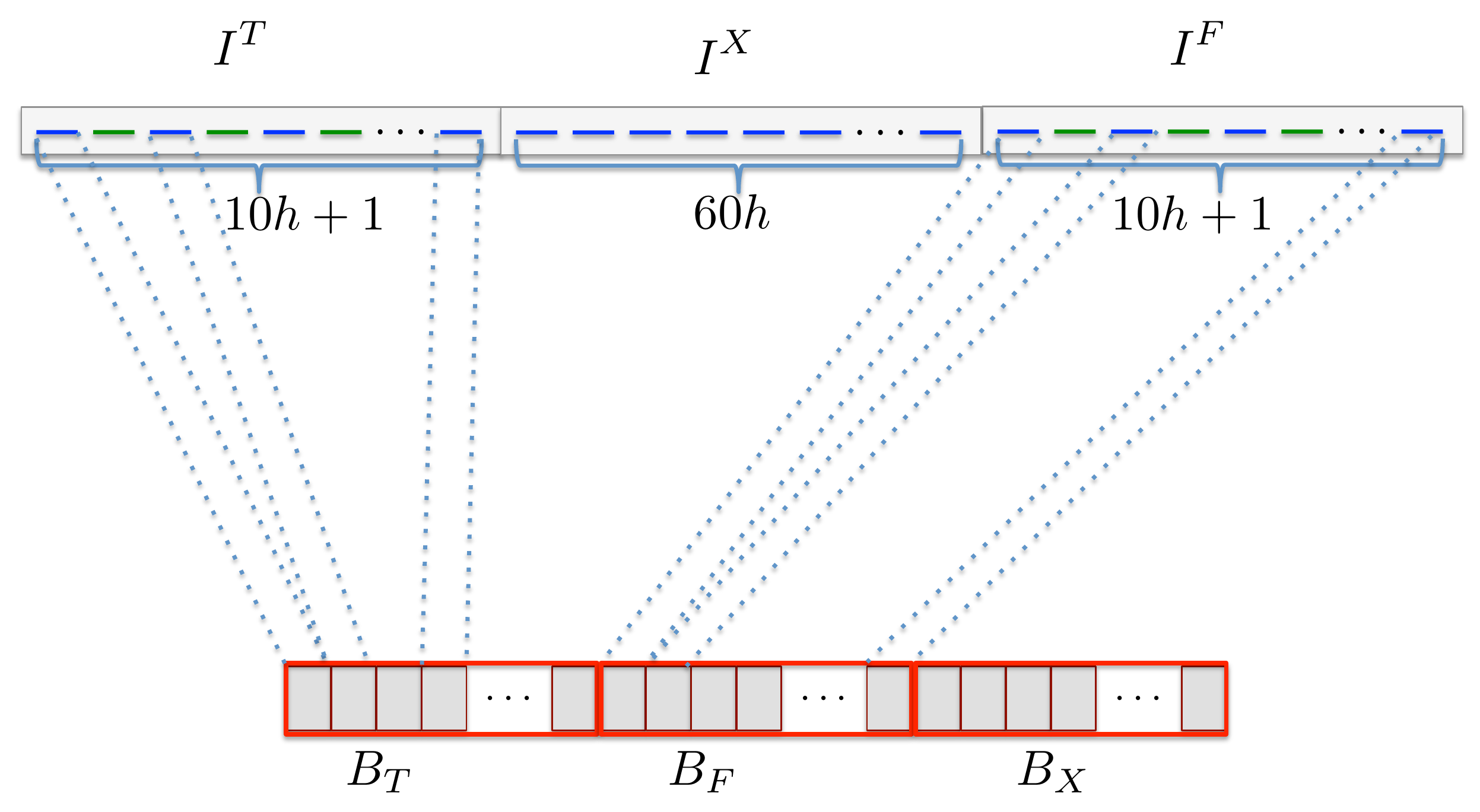}}
\caption{Variable gadget for level-$(i+1)$ instance. Index $x$ is omitted for convenience. Blue sub-intervals of $I^F$ belong to instances $\iset_j^F$ whose destinations lie in $B_F$; green sub-intervals belong to instances $\iset_j(C,x)$ associated with clauses $C\in \cset(x)$. The $(5h)$ green intervals are partitioned into $5$ groups of $h$ consecutive intervals each, and each group belongs to a distinct clause. 
Blue and green intervals of $I^T$ are dealt with similarly.
\label{fig: level 2 var gadget}}
\end{figure}

Recall that the length of box $B_X(x)$ is $60hcL'_i$, while boxes $B_T(x),B_F(x)$ have length $(5h+1)cL'_i$ each. The height of each box is $H_i$. Recall also that box $B(x)$ has length $L^V=  4N_{i+1}+(70h+2)cL'_i$ and height $H^V=H_i+2N_i$.

 We place the boxes $B_T(x),B_F(x)$ and $B_X(x)$ side-by-side inside $B(x)$ in this order, so that the middle row of each box is contained in the middle row of $B(x)$, and there is a horizontal spacing of  $2N_{i+1}$ between the left boundaries of $B_T(x)$ and $B(x)$, and between the right boundaries of $B_X(x)$ and $B(x)$ (see Figure~\ref{fig: var box}). Notice that there is no horizontal spacing between $B_T(x),B_F(x)$ and $B_X(x)$, and all destination vertices lying in $B(x)$ belong to the middle row of $B(x)$, and hence to the middle row of $B(\iset)$.

\begin{figure}[h]
\scalebox{0.5}{\includegraphics{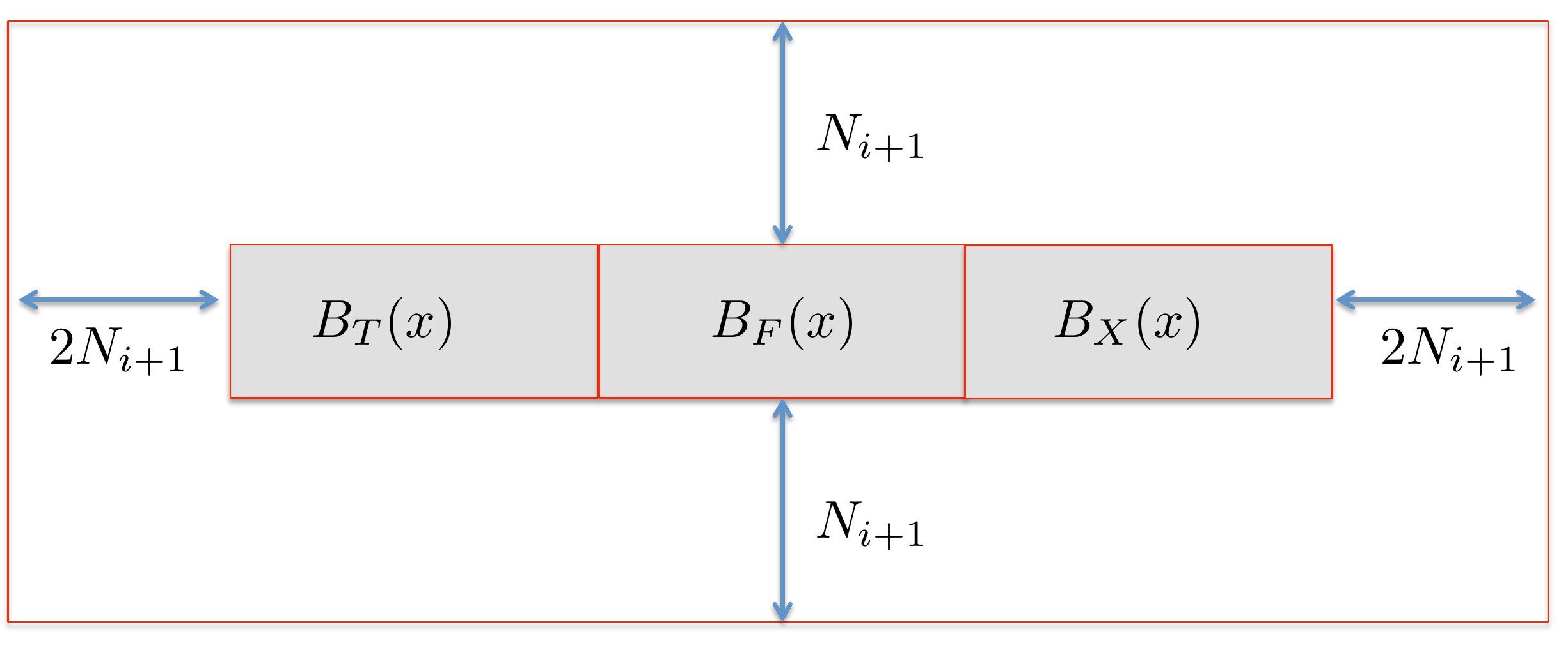}}
\caption{Box $B(x)$. Height: $H^V=H_i+2N_{i+1}$, length: $L^V=4N_{i+1}+(70h+2)cL'_i$. \label{fig: var box}}
\end{figure}

 Consider the set $\cset(x)\subseteq \cset$ of clauses in which variable $x$ appears without negation. Assume without loss of generality that $\cset(x)=\set{C_1,\ldots,C_r}$, where $r\leq 5$. For each $1\leq r'\leq r$, we will create $h$ level-$i$ instances of width $c$, that represent the variable $x$ of $C_{r'}$. We denote these instances by $\iset_j(C_{r'},x)$, for $1\leq j\leq h$. Consider the interval $I^F(x)$. We will use the sub-intervals $Y_j^F$ of $I^F(x)$ as intervals $Z(\iset_{j}(C_{r'},x))$, where, intuitively, intervals corresponding to the same clause-variable pair appear consecutively. Formally, for each $1\leq r'\leq r$, for each $1\leq j\leq h$, we use the interval $Y^F_{(r'-1)h+j}$  of $I^F(x)$ as $Z(\iset_j(C_{r'},x))$, and we say that it is the sub-interval of $I^F(x)$ that belongs to instance $\iset_j(C_{r'},x)$. Intuitively,  if $x$ is assigned the value \false, then we will route a large number of demand pairs in $\mset^F(x)$. The paths routing these pairs will ``block'' the intervals $Y^F_j$ of $I^F(x)$, thus preventing us from routing demand pairs that belong to instances $\iset_j(C_{r'},x)$, for $1\leq j\leq h$ and $C_{r'}\in \cset(x)$.

We treat the subset $\cset'(x)\subseteq \cset$ of clauses containing the negation of $x$ similarly, except that we assign to each resulting instance an  interval $Y_j^T$ of $I^T(x)$. 


 \paragraph{Clause Gadgets.}
 Consider some clause $C_q=(\ell_{q_1}\bor \ell_{q_2}\bor \ell_{q_3})$. For each one of the three literals of $C_q$, we construct $h$ level-$i$ width-$c$ instances, with instances $\set{\iset_j(C_q,\ell_{q_z})}_{j=1}^{h}$ representing the literal $\ell_{q_z}$, for $1\leq z\leq 3$. 
 Recall that $B(C_q)$ is a grid of height $H^C=H_i$ and length $L^C=3ch L_i'$. We partition $B(C_q)$ into $h$ sub-grids $B^1(C_q),\ldots,B^h(C_q)$, each of which has height $H_i$ and length $3cL'_i$. For each $1\leq j\leq h$, we place the boxes $B(\iset_j(C_q,\ell_{q_1})),B(\iset_j(C_q,\ell_{q_2})),B(\iset_j(C_q,\ell_{q_3}))$ inside $B^j(C_q)$ side-by-side in this order (see Figure~\ref{fig: clause-gadget}). 
 

 \begin{figure}[h]
\scalebox{0.6}{\includegraphics{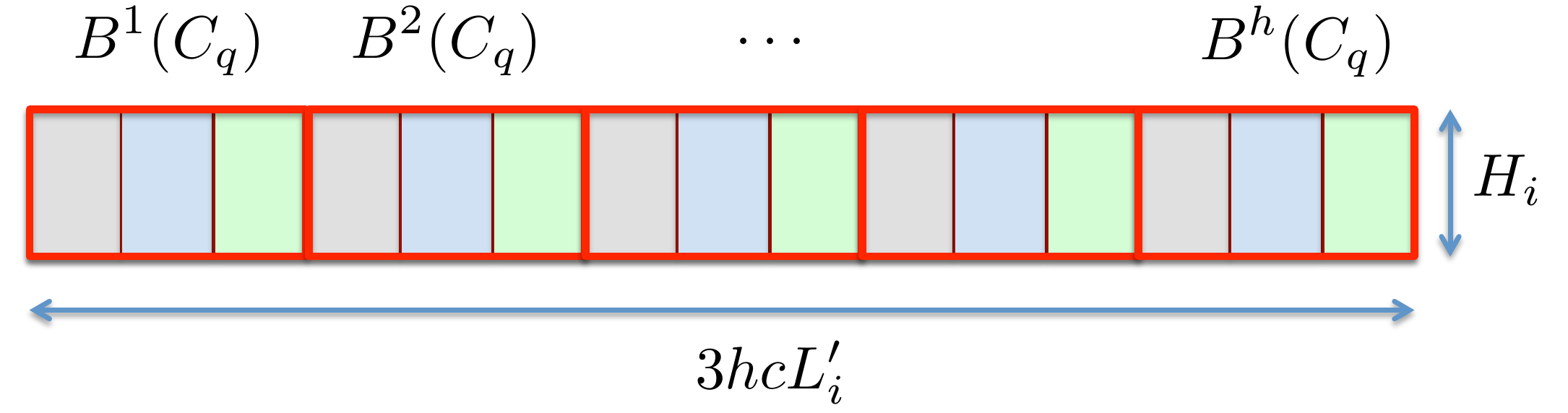}}
\caption{Box $B(C_q)$. Different colors show boxes that represent the three different literals. \label{fig: clause-gadget}}
\end{figure}
 
  The intervals $Z(\iset_j(C_q,\ell_{q_z}))$ are the same as the ones defined in the constructions of the variable gadgets. We denote by $\mset(C_q)$ the set of all demand pairs whose destinations lie in $B(C_q)$, and we call them \emph{clause-pairs representing $C_q$}. For each $1\leq z\leq 3$, we denote by $\mset(C_q,\ell_{q_z})$ the set of all demand pairs that belong to instances $\iset_j(C_q,\ell_{q_z})$, for $1\leq j\leq h$, and we sometimes say that they represent literal $\ell_{q_z}$ of clause $C_q$. We then let $\mset^C=\bigcup_{q=1}^m\mset(C_q)$ be the set of all clause-pairs, and $\mset^V=\bigcup_{j=1}^n\mset(x_j)$ the set of all variable-pairs. Our final set of demand pairs is $\mset=\mset^V\cup \mset^C$. This completes the definition of the level-$(i+1)$ instance.

\label{-------------------------------------------sec: YI-----------------------------------}
 \section{Yes-Instance Analysis}\label{sec: YI}
 

 The goal of this section is to prove the following theorem.
 
 \begin{theorem}\label{thm: yi}
 Assume that the input 3SAT(5) formula $\phi$ is a \yi. Then for all $i\geq 0$, for every instantiation of the level-$i$ instance $\iset$, there is a solution routing $N_i$ demand pairs, that respects the box $B(\iset)$.
 \end{theorem}
 
 The remainder of this section is devoted to proving this theorem. The proof proceeds by induction. For $i=0$, $N_0=1$, and it is easy to see that for any instantiation of the level-$0$ instance $\iset$, there is a solution that routes the unique demand pair of this instance and respects the box $B(\iset)$. We now assume that the theorem holds for some $i\geq 0$, and prove it for a level-$(i+1)$ instance, that we denote by $\iset$. We assume that we are given an instantiation of instance $\iset$, that consists of a grid $G_{i+1}$ of length at least $2L_{i+1}+2L_{i+1}'+4H_{i+1}$ and height at least $3H_{i+1}$, the placement  of the path $Z(\iset)$ on the top boundary of $G_{i+1}$, and the placement of the box $B(\iset)$ inside $G_{i+1}$, at distance at least $H_{i+1}$ from its boundaries. We denote the resulting graph by $G$, and the resulting set of demand pairs by $\mset$. Recall that our construction combines a number of level-$i$ instances.  From the induction hypothesis, for each such instance $\iset'$, for every instantiation of instance $\iset'$, there is a set $\pset(\iset')$ of disjoint paths, routing a set of $N_i$ demand pairs of $\iset'$, such that the paths in $\pset(\iset')$ respect the box $B(\iset')$. It is easy to verify that, if a set $\mset^*(\iset')$ of demand pairs of $\iset'$ has a routing that respects $B(\iset')$ in one instantiation of $\iset'$, then it has such a routing in every instantiation of $B(\iset')$. Therefore, for every level-$i$ instance $\iset'$, we can fix one such set $\mset^*(\iset')$ of demand pairs with $|\mset^*(\iset')|=N_i$, and assume that $\mset^*(\iset')$ has a routing that respects $B(\iset')$ in every instantiation of $\iset'$.
 
 Recall that our construction combines level-$i$ instances into $c_{i+1}$-wide level-$i$ instances. Let $\iset''$ be any such $c_{i+1}$-wide level-$i$ instance, and assume that it consists of level-$i$ instances $\iset_1,\ldots,\iset_{c_{i+1}}$. We set $\mset^*(\iset'')=\bigcup_{j=1}^{c_{i+1}}\mset^*(\iset_j)$, so $|\mset^*(\iset'')|=c_{i+1}N_i$. It is easy to see that  for any instantiation of $\iset''$, there is a routing of all demand pairs in $\mset^*(\iset'')$, such that for each $1\leq j\leq c_{i+1}$, the demand pairs in $\mset^*(\iset_j)$ are routed via paths that respects the box $B(\iset_j)$.

Consider now the given instantiation $(G,\mset)$ of the level-$(i+1)$ instance $\iset$.
 We first select the set $\hmset\subseteq \mset$ of demand pairs that we route, and then compute their routing. We fix some assignment $\aset$ to the variables $\set{x_1,\ldots,x_n}$ of $\phi$, that satisfies all clauses. 
 
 \paragraph{Variable Pairs.}
 Let $x$ be some variable, and let $\hmset^X(x)=\bigcup_{j=1}^{60h}\mset^*(\iset^X_j(x))$ --- the set of all demand pairs that are routed by the solutions to the $c_{i+1}$-wide level-$i$ instances $\iset_1^X(x),\ldots,\iset_{60h}^X(x)$. Notice that $|\hmset^X(x)|=60 h c_{i+1}N_i$. 
 If $x$ is assigned the value \true, then we let $\hmset^T(x)=\bigcup_{j=1}^{5h+1}\mset^*(\iset^T_j(x))$, and we set $\hmset^F(x)=\emptyset$. Notice that $|\hmset^T(x)|=(5h+1)c_{i+1}N_i$ in this case.
  Otherwise, we let  $\hmset^F(x)=\bigcup_{j=1}^{5h+1}\mset^*(\iset^F_j(x))$, so $|\hmset^F(x)|=(5h+1)c_{i+1}N_i$,  and we set $\hmset^T(x)=\emptyset$. 
  We denote $\hmset(x)=\hmset^X(x)\cup \hmset^T(x)\cup \hmset^F(x)$, and we let $\hmset^V=\bigcup_{x}\hmset(x)$, so $|\hmset^V|=nc_{i+1}(65h+1)N_i$.
 
\paragraph{Clause Pairs.} Let $C_q\in \cset$ be a clause, and let $\ell_q$ be a literal of $C_q$ whose value is \true (if there are several such literals, we select any one of them arbitrarily). We say that clause $C_q$ \emph{chooses} the literal $\ell_q$. For simplicity, we denote $\hmset_j(C_q)=\mset^*(\iset_j(C_q,\ell_q))$, for all $1\leq j\leq h$, and we let $\hmset(C_q)=\bigcup_{j=1}^{h}\hmset_j(C_q)$. Let $\hmset^C= \bigcup_{q=1}^m\hmset(C_q)$. Clearly, for each $1\leq q\leq m$, $|\hmset(C_q)|=hc_{i+1}N_i$, and overall, $|\hmset^C|=mhc_{i+1} N_i=5nhc_{i+1}N_i/3$.  
 
 Finally, we let $\hmset=\hmset^V\cup \hmset^C$, so $|\hmset|= nc_{i+1}\cdot(65h+1)N_i+5nc_{i+1}h N_i/3=n\cdot N_ic_{i+1} (200h/3+1)=N_{i+1}$.
 It is now enough to prove the following lemma.
 
 \begin{lemma}\label{lem: routing yi}
 There is a set $\pset$ of node-disjoint paths in graph $G$, routing all demand pairs in $\hmset$, such that $\pset$ respects box $B(\iset)$.
 \end{lemma}
 

For convenience, we denote the first row of the grid $G_{i+1}$ by $R$. 
 Let $S'\subseteq S(\hmset)$ be any subset of the source vertices of the demand pairs in $\hmset$. We say that the sources of $S'$ appear \emph{consecutively} on $R$, iff there is a sub-path $P$ of $R$ that contains all the vertices of $S'$ and does not contain any vertex of $S(\hmset)\setminus S'$. We let $\oset$ be the left-to-right ordering of the vertices of $S(\hmset)$ on row $R$.

  Consider some variable $x$, and denote by $\hmset'(x)$ the set of all demand pairs in $\hmset^C$ whose sources lie on the interval $I(x)$ --- these are the demand pairs representing the clauses $C_q$ that chose either $x$ or $\neg x$ as their literal. Assume first that $x$ is assigned the value \true. Then all sources of the demand pairs in $\hmset(x)$ lie on the intervals $I^X(x)$ and $I^T(x)$, while all sources in set $S(\hmset'(x))$ lie on $I^F(x)$: indeed, since $x$ is assigned the value \true,  the corresponding clauses must contain $x$ without negation and so their sources lie on $I^F(x)$. Therefore, the sources of each set $\hmset(x)$ and $\hmset'(x)$ are consecutive on $R$. If $x$ is assigned the value \false, then similarly all sources of the demand pairs in $\hmset(x)$ are consecutive on $R$, while all sources of the demand pairs in $\hmset'(x)$ appear on $I^T(x)$ and are therefore also consecutive on $R$. In either case, for every clause $C_q\in \cset$, the vertices of $S(\hmset(C_q))$ are consecutive on $R$.

 In order to construct the routing, it is convenient to use special subgraphs of $G$ that we call snakes, in which routing can be done easily. We start by defining a corridor, and then combine several corridors to define a snake.
 
 Recall that our graph $G$ is a subgraph of a grid $G_{i+1}$. Recall also that, given a set $\rset$ of consecutive rows of $G_{i+1}$ and a set $\wset$ of consecutive columns of $G_{i+1}$, we denoted by $\Y(\rset,\wset)$ the subgraph of $G_{i+1}$ induced by the set $\set{v(j,j')\mid R_j\in \rset, W_{j'}\in \wset}$ of vertices. We say that $\Y=\Y(\rset,\wset)$ is a \emph{corridor} iff every vertex of $\Y$ belongs to $G$. For convenience, we will say that $\Y$ is a corridor spanned by the rows in $\rset$ and the columns of $\wset$.  Let $R'$ and $R''$ be the first and the last row of $\rset$ respectively, and let $W'$ and $W''$ be the first and the last column of $\wset$ respectively. Each of the four paths $\Y\cap R',\Y\cap R'',\Y\cap W'$ and $\Y\cap W''$ is called a \emph{boundary edge} of $\Y$, and their union is called the \emph{boundary} of $\Y$. We say that two corridors $\Y,\Y'$ are \emph{internally disjoint}, iff every vertex $v\in \Y\cap \Y'$ belongs to the boundaries of both corridors. We say that two internally disjoint corridors $\Y,\Y'$ are \emph{neighbors} iff $\Y\cap \Y'\neq \emptyset$.
 
 We are now ready to define snakes. A snake $\yset$ of length $\ell$ is a sequence $\Y_1,\Y_2,\ldots,\Y_{\ell}$ of $\ell$ corridors that are pairwise internally disjoint. Moreover, for all $1\leq \ell',\ell''<\ell$, $\Y_{\ell'}$ is a neighbor of $\Y_{\ell''}$ iff $|\ell'-\ell''|=1$. We say that the width of the snake is $w$ iff for each $1\leq \ell'\leq \ell$, $\Y_{\ell'}$ is spanned by a set of at least $w$ rows and by a set of at least $w$ columns, and for all $1\leq \ell'<\ell$, $\Y_{\ell'}\cap \Y_{\ell'+1}$ contains at least $w$ vertices.
 We use the following simple claim, whose proof is deferred to the Appendix, for routing in snakes.
 
 \begin{claim}\label{claim: routing in a snake}
 Let $\yset=(\Y_1,\ldots,\Y_{\ell})$ be a snake of width $w$, and let $A,A'$ be two sets of vertices with $|A|=|A'|\leq w-2$, such that the vertices of $A$ lie on a single boundary edge of $\Y_1$ and the vertices of $A'$ lie on a single boundary edge of $\Y_{\ell}$. Then there is a set $\qset$ of node-disjoint paths contained in $\bigcup_{\ell'=1}^{\ell}\Y_{\ell'}$, that  connect every vertex of $A$ to a distinct vertex of $A'$. 
 \end{claim}
 
 Our routing consists of two steps. In the first step, we connect each source vertex $s\in S(\hmset)$ to the top boundary of the unique box in $\bset=\set{B(x_1),\ldots,B(x_n),B(C_1)\ldots,B(C_m)}$, that contains its corresponding destination vertex. We will later select specific vertices on the top boundary of each such box to which the sources are routed. The resulting paths will be internally disjoint from the boxes in $\bset$. In the second step, we complete the routing inside each box. The first step is summarized in the following claim.
 
 \begin{claim}\label{claim: step 1 of routing}
 Suppose we are given, for every variable $x_j$, a set $A(x_j)$ of $|\hmset(x_j)|$ vertices on the top boundary of $B(x_j)$, and for every clause $C_q\in \cset$, a set $A(C_q)$ of $|\hmset(C_q)|$ vertices on the top boundary of $B(C_q)$. Then there is a collection $\pset'$ of node-disjoint paths in $G$ with the following properties:
 
 \begin{itemize}
 \item the paths of $\pset'$ are internally disjoint from all boxes in $\bset=\set{B(x_1),\ldots,B(x_n),B(C_1)\ldots,B(C_m)}$;
 
 \item for every variable $x_j$ of $\phi$, there is a subset $\pset(x_j)\subseteq \pset'$ of paths, that connect every vertex in $S(\hmset(x_j))$ to a vertex of $A(x_j)$, so that the paths in $\pset(x_j)$ are order-preserving; and
 
 \item for every clause $C_q\in \cset$, there is a subset $\pset(C_q)\subseteq \pset'$ of paths, connecting every vertex in $S(\hmset(C_q))$ to a vertex of $A(C_q)$, so that the paths in $\pset(C_q)$ are order-preserving.
 \end{itemize}
 \end{claim}
 
 \begin{proof}
 We provide here a high-level sketch of the proof. Turning it into a formal proof is straightforward but somewhat tedious; we defer the formal proof to the Appendix. We construct $5$ sets of node-disjoint paths, $\pset_0,\ldots,\pset_4$. Let $Z'$ be the set of $|\hmset|$ leftmost vertices on the opening of $B(\iset)$. Set $\pset_0$ of paths routes all vertices in $S(\hmset)$ to the vertices of $Z'$ via node-disjoint paths that are order-preserving and internally disjoint from $B(\iset)$, in a straightforward manner. Next, we select a  subset $\Gamma$ of  $|\hmset|$ vertices on the top row of $B^V$. Set $\pset_1$ of paths connects every vertex in $Z'$ to a distinct vertex of $\Gamma$, so that the paths in $\pset_1$ are node-disjoint, order-preserving, internally disjoint from $B^V\cup B^C$, and contained in $B(\iset)$. In order to route the paths in $\pset_1$, we construct an appropriate snake inside $B(\iset)$ (see Figure~\ref{fig: yi-routing}). If we combine the paths in $\pset_0$ and $\pset_1$, we obtain a set of node-disjoint order-preserving paths, connecting every vertex of $S(\hmset)$ to a distinct vertex of $\Gamma$. Paths originating at the vertices corresponding to variable demand pairs then continue to their corresponding boxes $B(x_j)$, while paths originating at the vertices corresponding to clause demand pairs continue directly to the bottom boundary of $B^V$, exploiting the spacing between the boxes $B(x_j)$. The vertices of $\Gamma$ are selected in such a way that each such path can be implemented as a sub-path of a column of $B^V$ (see Figure~\ref{fig: yi-routing}). The resulting set of paths is denoted by $\pset_2$. We then connect the endpoints of the paths in $\pset_2$ that appear on the bottom boundary of $B^V$ to a subset $\Gamma'''$ of vertices on the top boundary of $B^C$, by defining an appropriate snake that exploits the spacing between $B^V$ and $B^C$. This latter set of paths is denoted by $\pset_3$. Let $\pset''$ be the set of paths obtained by combining the paths in $\pset_0,\ldots,\pset_3$. For every variable $x_j$ of $\phi$, there is a subset $\pset(x_j)\subseteq \pset''$ of paths, connecting every vertex of $S(\hmset(x_j))$ to a vertex of $A(x_j)$. For every clause $C_q\in \cset$, there is a subset $\pset'(C_q)\subseteq \pset''$ of paths, connecting every vertex of $S(\hmset(C_q))$ to a distinct vertex  of $\Gamma'''$. We denote by $\Gamma(C_q)\subseteq \Gamma'''$ the set of vertices that serve as endpoints of the paths in $\pset'(C_q)$. Note that for each clause $C_q\in \cset$, the vertices of $\Gamma(C_q)$ appear consecutively on the top boundary of $B^C$.

 \begin{figure}[h]
 \scalebox{0.6}{\includegraphics{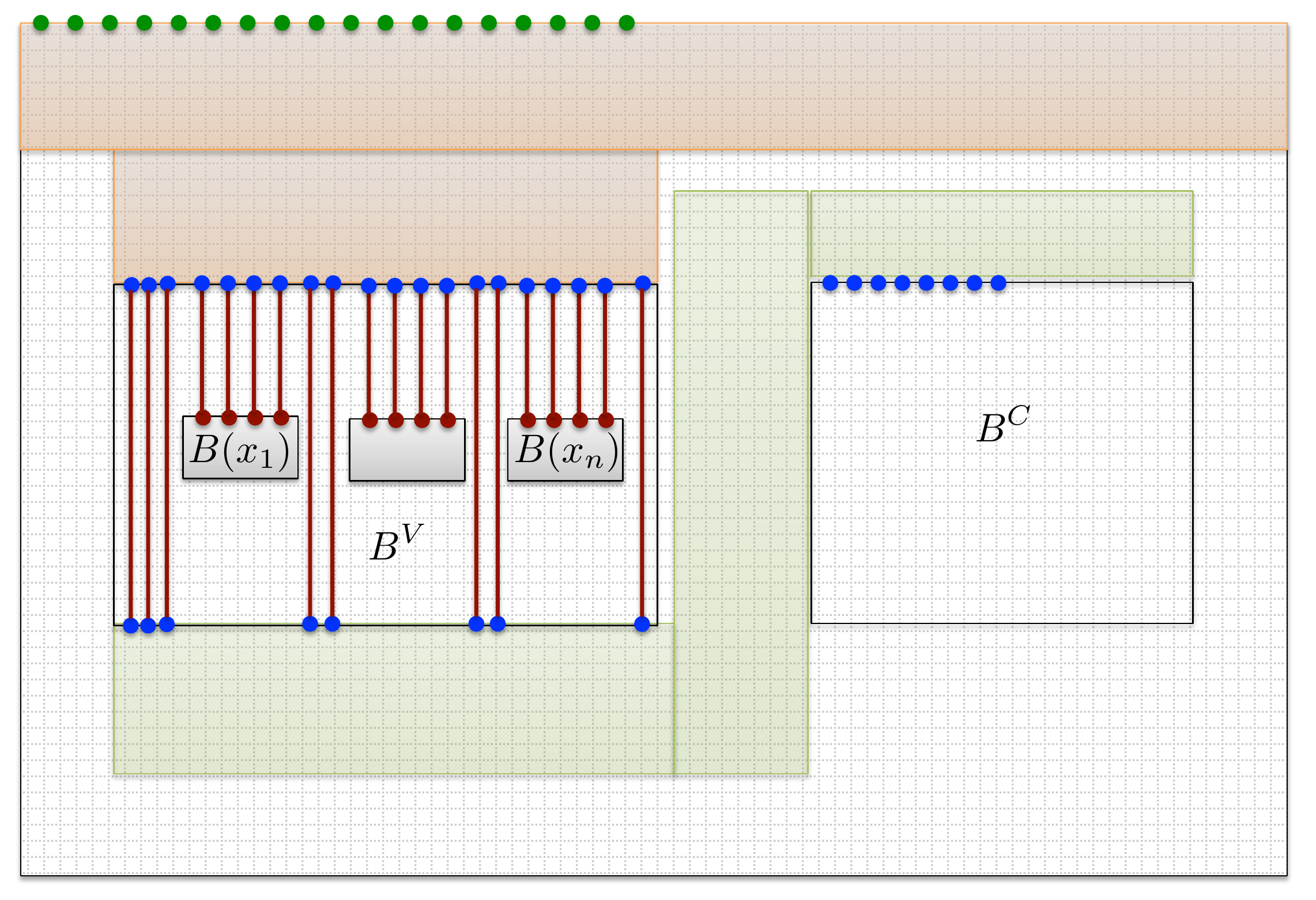}}
 \caption{Routing the sets $\pset_1,\pset_2$ and $\pset_3$ of paths inside $B(\iset)$. The paths in $\pset_2$ are shown in red; the paths of $\pset_1$ are routed inside the orange snake, and the paths of $\pset_3$ are routed inside the green snake.\label{fig: yi-routing}}
 \end{figure}

 We can now define an ordering $\oset'$ of the clauses in $\cset$ as follows: clause $C_q$ appears before clause $C_{q'}$ in this ordering iff the vertices of $\Gamma(C_q)$ appear to the left of the vertices of $\Gamma(C_{q'})$ on the top row of $B^C$. Our final step is to define a set $\pset_4$ of node-disjoint paths, that are contained in $B^C$ and are internally disjoint from all boxes $\set{B(C_q)\mid C_q\in \cset}$, such that for each clause $C_q\in \cset$, there is a subset $\pset'(C_q)\subseteq \pset_4$ of order-preserving paths, connecting the vertices of $\Gamma(C_q)$ to the vertices of $A(C_q)$.
 Set $\pset_4$ is constructed by using a standard snake-like routing inside the box $B^C$. For each clause $C_q$, we extend the box $B(C_q)$ by $N_{i+1}$ columns to the right and to the left, obtaining a larger box $B'(C_q)$. The idea is to route the paths, so that they visit the boxes $B'(C_1),\ldots,B'(C_m)$ in this order. For all $1\leq q\leq m$, only paths corresponding to the clauses $C_q,C_{q+1},\ldots,C_m$ visit the box $B'(C_q)$. The paths corresponding to clause $C_q$ enter the box $B(C_q)$ and terminate at its top boundary. For each clause $C\in \set{C_{q+1},\ldots,C_m}$, if $C$ appears before $C_q$ in the ordering $\oset'$, then the paths corresponding to $C$ traverse the box $B'(C_q)$ to the left of $B(C_q)$, and otherwise they traverse it to the right of $B(C_q)$. The routing itself is implemented by constructing an appropriate collection of snakes, where each snake is used to connect the vertices on the bottom boundary of $B'(C_q)$ to the vertices on the top boundary of $B'(C_{q+1})$ and exploits the spacing between the two boxes (see Section~\ref{subsec: proof of routing claim} of Appendix for more details).
 \end{proof} 
  
 We now define the sets $\set{A(x_j) \mid 1\leq j\leq n}\cup \set{A(C_q)\mid C_q\in \cset}$ of vertices and complete the routing.

Consider first some clause $C_q\in \cset$, and  assume that it has chosen the literal $\ell_q$. Fix some index $1\leq j\leq h$, and consider the width-$c_{i+1}$ level-$i$ instance $\iset'=\iset_j(C_q,\ell_q)$. Instance $\iset'$ consists of $c_{i+1}$ level-$i$ instances, that we denote by $\iset_1',\ldots,\iset'_{c_{i+1}}$. Then set $\hmset(C_q)$ contains $c_{i+1}N_{i}$ demand pairs that belong to instance $\iset'$, with exactly $N_{i}$ pairs from each instance $\iset_{r}'$. For each such instance $\iset_{r}'$, we select $N_{i}$ leftmost vertices on the opening of the box $B(\iset'_r)$. We then let $A(C_q)$ be the set of all such vertices we have selected, for all $1\leq j\leq h$ and $1\leq r\leq c_{i+1}$, so $|A(C_q)|=hc_{i+1}N_{i}=|\hmset(C_q)|$.  Notice that, if we are given any set $\pset(C_q)$ of $|\hmset(C_q)|$ node-disjoint order-preserving paths, that connect the vertices of $S(\hmset(C_q))$ to the vertices of $A(C_q)$, such that the paths in $\pset(C_q)$ are internally disjoint from the box $B(C_q)$, then we can extend the paths in $\pset(C_q)$ inside $B(C_q)$, so that they route the set $\hmset(C_q)$ of demand pairs. In order to do so, for each $c_{i+1}$-wide level-$i$ instance $\iset'=\iset_j(C_q,\ell_q)$, for each level-$i$ instance $\iset'_r$ that was used to construct $\iset'$, we use the routing that respects the box $B(\iset'_r)$, which is guaranteed by the induction hypothesis and by our choice of the demand pairs to route, in order to connect the vertices we have selected on the opening of $B(\iset'_r)$ to the corresponding destination vertices.
 
 Consider now some variable $x_j$, for $1\leq j\leq n$. We now show how to select a set $A(x_j)$ of $|\hmset(x_j)|$ vertices on the top boundary of box $B(x_j)$. We start with $A(x_j)=\emptyset$ and then add vertices to it.

 Assume first that $x_j$ is assigned the value \true. Let $\iset'$ be any of the $(65h+1)c_{i+1}$ level-$i$ instances whose box $B(\iset')$ is contained in $B_X(x)\cup B_T(x)$. Recall that $\hmset(x_j)$ contains $N_i$ demand pairs that belong to this instance. Let $A'(\iset')$ be the set of $N_i$ leftmost vertices on the opening of the box $B(\iset')$. For each vertex $v\in A'(\iset')$, we add to $A(x_j)$ the vertex $v'$ that belongs to the top boundary of $B(x_j)$ and lies on the same column $W$ as $v$. We let $P_v$ be the sub-path of $W$ between $v$ and $v'$. Notice that $|A(x_j)|=|\hmset(x_j)|$, and, given any set $\pset(x_j)$ of node-disjoint order-preserving paths connecting the vertices of $S(\hmset(x_j))$ to the vertices of $A(x_j)$, such that the paths in $\pset(x_j)$ are internally disjoint from $B(x_j)$, we can extend these paths inside the box $B(x_j)$, so that they route the set $\hmset(x_j)$ of demand pairs. This is done using the box-respecting routing of each level-$i$ instance as before, together with the set $\set{P_v\mid v'\in A(x_j)}$ of paths.

 Finally, assume that $x_j$ is assigned the value \false. 
 Notice that the sources of the \extra demand pairs for $x_j$ appear to the left of the sources of the \false demand pairs of $x_j$, while the box $B_F(x_j)$ appears to the left of the box $B_X(x_j)$. Therefore, the straightforward routing as above cannot be employed here and we need to ``switch'' the two sets of paths. In order to do so, we exploit the spacing between the boxes $B_T(x_j),B_F(x_j),B_X(x_j)$ and the boundaries of $B(x_j)$ (see Figure~\ref{fig: var box}).
 
 For each of the $(5h+1)c_{i+1}$ level-$i$ instances whose box $B(\iset')$ is contained in $B_F(x)$, we define the set $A'(\iset')$ of $N_i$ vertices on the opening of the box $B(\iset')$, and for each such vertex $v$, the corresponding path $P_v$ and vertex $v'$ that is added to $A(x_j)$ exactly as before. Let $J$ denote the set of vertices added to $A(x_j)$ so far.  We also add to $A(x_j)$ the set $J'$ of $60hc_{i+1}N_i$ leftmost vertices on the top boundary of $B(x_j)$. This ensures that $|A(x_j)|=(65h+1)hc_{i+1}N_i=|\hmset(x_j)|$. Assume now that we are given any set $\pset(x_j)$ of node-disjoint order-preserving paths connecting the vertices of $S(\hmset(x_j))$ to the vertices of $A(x_j)$, such that the paths in $\pset(x_j)$ are internally disjoint from $B(x_j)$. We now show how to extend  these paths inside the box $B(x_j)$, so that they route the set $\hmset(x_j)$ of demand pairs. We partition the set $\pset(x_j)$ into two subsets: set $\pset^X$ contains all paths that originate at the sources of the demand pairs in $\hmset^X(x_j)$ --- that is, the \extra demand pairs of $x_j$, and $\pset^F$ contains all remaining demand pairs, that must originate at the source vertices of the \false demand pairs for $x_j$. Since the set $\pset(x_j)$ of paths is order-preserving, the paths in $\pset^X$ terminate at the vertices of $J'$, while the paths in $\pset^F$ terminate at the vertices of $J$. We extend the paths in set $\pset^F$ exactly as before, using the paths in set $\set{P_v\mid v'\in  J}$, and then using the box-preserving routing of each corresponding level-$i$ instance to connect each source vertex of $S(\hmset^F(x_j))$ to its destination.
 
 In order to extend the paths in $\pset^X$ we do the following. For each level-$i$ instance $\iset'$ whose box $B(\iset')$ is contained in $B_X(x_j)$, we let $A'(\iset')$ be the set of $N_{i}$ left-most vertices on the opening of $B(\iset')$. Let $A'$ be the set of all such vertices, for all such instances $\iset'$, so $|A'|=|J'|$. It is now enough to construct a set $\qset$ of order-preserving node-disjoint paths contained in $B(x_j)$ that connect every vertex in $J$ to a distinct vertex of $A'$, such that the paths in $\qset$ are internally disjoint from $B_T(x_j)\cup B_F(x_j)\cup B_X(x_j)$ and are completely disjoint from $\set{P_v\mid v'\in J}$. We show the existence of the set $\qset$ of paths by constructing a snake $\yset$, that consists of four corridors. The first corridor, $\Y_1$, is the set of the first $N_{i+1}$ columns of $B(x_j)$. The third corridor, $\Y_3$, is the set of the last $N_{i+1}$ columns of $B(x_j)$. Let $W'$ be the last column of $\Y_1$ and let $W''$ be the first column of $\Y_3$. The second corridor, $\Y_2$, is spanned by the bottom $N_{i+1}$ rows of $B(x_j)$ and the columns of $B(x_j)$ from $W'$ to $W''$, including these two columns. Let $W'''$ be the leftmost column of $G_{i+1}$ that intersects $B_X(x_j)$; let $R'$ be the topmost row of $G_{i+1}$ that intersects $B_X(x_j)$, and let $\rset$ be the set of $N_{i+1}$ consecutive rows lying above $R'$, including $R'$. The last corridor, $\Y_4$, is spanned by the set $\rset$ of rows, and the set of all columns from $W'''$ to $W''$ (see Figure~\ref{fig: var-routing}). Using Claim~\ref{claim: routing in a snake}, it is immediate to verify that the desired set $\qset$ of paths exists inside the resulting snake. We then extend the routing inside each box $B(\iset')$ of each corresponding level-$i$ instance $\iset'$ exactly as before.

 \begin{figure}[h]
 \scalebox{0.4}{\includegraphics{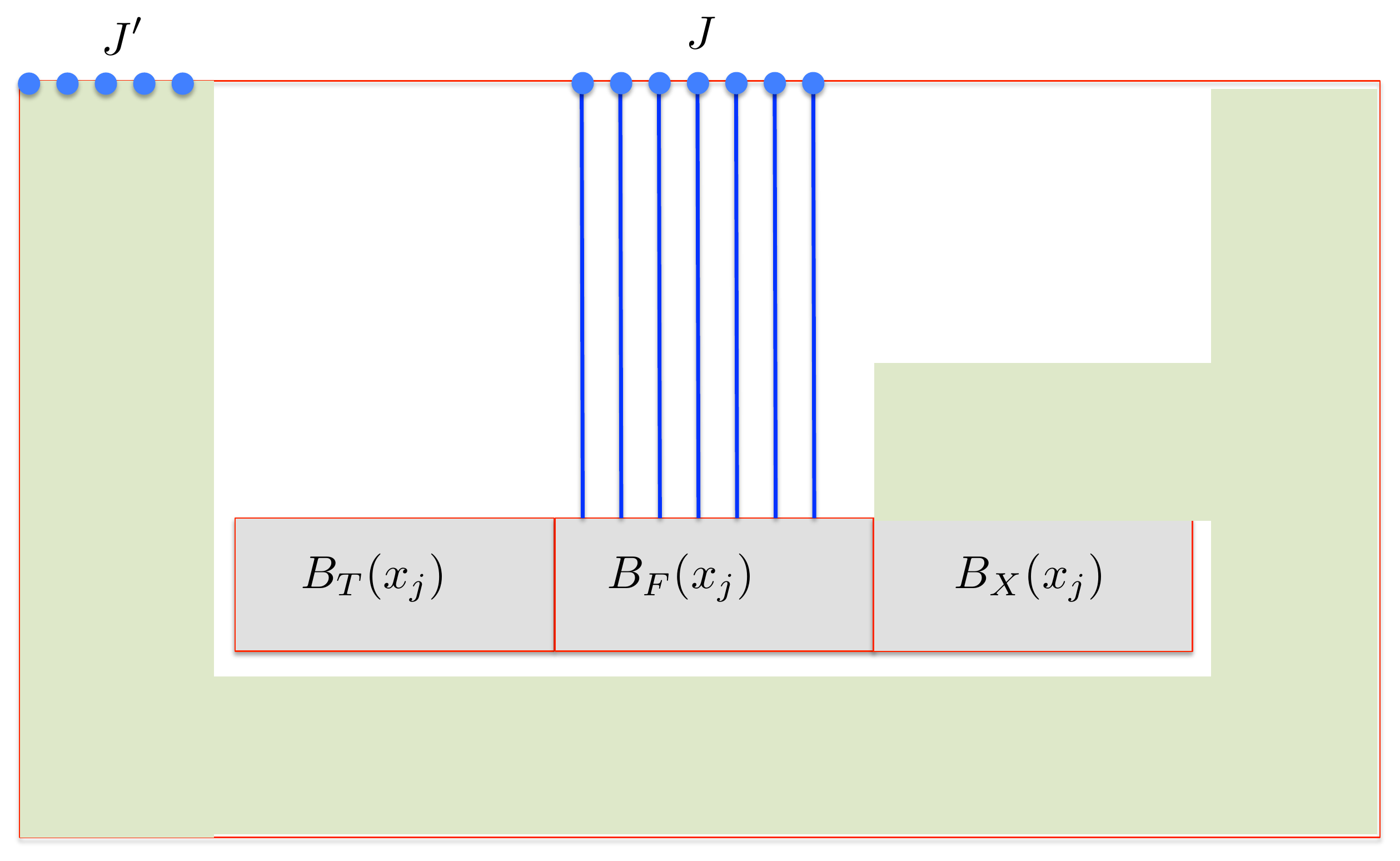}}
 \caption{Routing inside box $B(x_j)$ if $x_j$ is assigned value $F$.\label{fig: var-routing}}
 \end{figure}

 Using Claim~\ref{claim: step 1 of routing} with our definitions of the sets $\set{A(x_j)\mid 1\leq j\leq n}\cup \set{A(C_q)\mid 1\leq q\leq m}$ of vertices, and the above discussion, we can route all demand pairs in $\hmset$ via a set $\pset$ of node-disjoint paths that respects the box $B(\iset)$.

\label{-------------------------------------------sec: NI-----------------------------------}
 \section{No-Instance Analysis}\label{sec: NI}
In this section we analyze the \ni case, by proving the following theorem.

\begin{theorem}\label{thm: ni}
Assume that $\phi$ is a \ni. Then for every integer $i\geq 0$, for every instantiation of the level-$i$ instance $\iset$, and for every solution $\pset$ to this instance, $|\pset|\leq N'_i$.
\end{theorem}

The proof is again by induction. For the base case of $i=0$, $N'_i=1$, and the corresponding level-$0$ instance contains a single demand pair, so the theorem clearly holds. We now assume that the theorem holds for some value $i\geq 0$ and prove it for a level-$(i+1)$ instance $\iset$. We assume that we are given some instantiation of $\iset$, and from now on our goal is to prove that no solution to this instance  of \NDP can route more than $N_{i+1}'=(1-\delta)nc_{i+1}\cdot  (200h/3+1)N'_i$ demand pairs, where $\delta=8\eps^2/10^{12}$. We assume for contradiction that this is not the case, and we let $\tpset$ be a collection of more than $N'_{i+1}$ node-disjoint paths, routing a set $\tmset\subseteq \mset$ of demand pairs. For every demand pair $(s,t)\in \tmset$, we let $P(s,t)\in \tpset$ be the path routing this pair in the solution. 

Recall that our construction of a level-$(i+1)$ instance $\iset$ consists of a number of copies of $c_{i+1}$-wide level-$i$ instances: For every variable $x$ of $\phi$, we have constructed $(70h+2)$ such instances ($60h$ instances for the extra pairs, and $(5h+1)$ instances each for \true and \false pairs); for every clause $C\in \cset$, we have constructed $3h$ such instances. Therefore, overall we use $(70h+2)n+3hm=75nh+2n$ copies of $c_{i+1}$-wide level-$i$ instances (we have used the fact that $m=5n/3$). We assume (by induction) that at most $c_{i+1}N'_i$ pairs from each such instance are in $\tmset$.
We say that a $c_{i+1}$-wide level-$i$ instance $\iset'$ is \emph{interesting} iff at least $25H_i$ demand pairs of $\mset(\iset')$ belong to $\tmset$; otherwise we say that it is uninteresting. We let $\hat \mset\subseteq \tmset$ be the set of all demand pairs that belong to uninteresting instances, and we call them \emph{excess pairs}. We need the following simple observation.

\begin{observation}\label{obs: excess pairs}
$|\hat \mset|\leq \delta (200h/3+1)nc_{i+1}N'_i$.
\end{observation}

\begin{proof}
As observed above, there are at most $75nh+2n$ uninteresting instances, each of which contributes at most $25H_i\leq \frac{50c_{i+1}\eps^2N'_i}{10^{13}}$ excess demand pairs. Therefore, it is enough to show that:

\[(75nh+2n)\cdot  \frac{50c_{i+1}\eps^2N'_i}{10^{13}}\leq \delta\left(\frac{200 h}{3}+1\right ) nc_{i+1}N'_i,\]

which is immediate to verify, substituting $\delta=8\eps^2/10^{12}$.
\end{proof}

It would be convenient for us to assume that no excess pairs exist. In order to do so, we discard all excess pairs from $\tmset$. From Observation~\ref{obs: excess pairs}, $|\tmset|\geq (1-2\delta)nc_{i+1}(200h/3+1)N'_i$ still holds.

For every variable $x$, we let $\tmset(x)=\tmset\cap \mset(x)$, and for every clause $C_q$, we let $\tmset(C_q)=\tmset\cap \mset(C_q)$. We also denote by $\tmset^V=\bigcup_{j=1}^n\tmset(x_j)$ and by $\tmset^C=\bigcup_{q=1}^m\tmset(C_q)$.

 For the sake of the \ni-analysis, it is convenient to view our construction slightly differently. Let $\phi$ be the input 3SAT(5) formula, and recall that $\cset=\set{C_1,\ldots,C_m}$ is the set of its clauses. For each clause $C_q\in \cset$, we create $h$ new clauses $C^1_q,\ldots,C^h_q$, each of which is a copy of the original clause. We let $\cset'=\set{C^j_q\mid 1\leq q\leq m, 1\leq j\leq h}$ be the resulting set of clauses, and $\phi'$ the corresponding 3SAT formula. In order to avoid confusion, we refer to the clauses in $\cset$ as the \emph{original} clauses, to the clauses of $\cset'$ as the \emph{new clauses}, and for each $1\leq q\leq m$, $1\leq j\leq h$, we call $C^j_q$ the \emph{$j$th copy of clause $C_q$}. Recall that the clause gadget for $C_q\in \cset$ contains $h$ boxes $B^1(C_q),\ldots,B^h(C_q)$, where box $B^j(C_q)$ is the union of three boxes: $B(\iset_j(C_q,\ell_{q_1})),B(\iset_j(C_q,\ell_{q_2}))$ and $B(\iset_j(C_q,\ell_{q_3}))$ (see Figure~\ref{fig: clause-gadget}). We think of the box $B^j(C_q)$ as representing the new clause $C^j_q\in \cset'$. For convenience, we denote by $\tmset(C_q^j)\subseteq \tmset(C_q)$ the set of all demand pairs routed by our solution whose destinations lie in $B^j(C_q)$. This set is further partitioned into three subsets, $\tmset(C_q^j,\ell_{q_1}),\tmset(C_q^j,\ell_{q_2}),\tmset(C_q^j,\ell_{q_3})$, each of which contains demand pairs from the instances $\iset_j(C_q,\ell_{q_1}),\iset_j(C_q,\ell_{q_2})$, and $\iset_j(C_q,\ell_{q_3})$ respectively.
The following observation is immediate:

\begin{observation}\label{obs: ni-satisfied new clauses}
If $\phi$ is a \ni, then for any assignment to its variables, at most $(1-\eps)mh$ clauses of $\cset'$ are satisfied.
\end{observation}

 \paragraph{Encircling and its Resolution}
Let $(s,t)\in \tmset$ be any demand pair routed by the solution. Recall that $(s,t)$ belongs to some level-$i$ instance $\iset'$, and we have defined a line $Q_t$ containing at most $H_i/2$ vertices of the graph, that connects $t$ to the bottom of the box $B(\iset')$ (which is a cut-out box). We say that a demand pair $(s',t')\in \tmset$ \emph{encircles} pair $(s,t)$ iff path $P(s',t')$ contains a vertex lying on $Q_t$. Since $Q_t$ contains at most $H_i/2$ vertices, at most $H_i/2$ demand pairs may encircle $(s,t)$. We repeatedly use the following simple lemma.

\begin{lemma}\label{lem: enc resolution}
Let $S_1,\ldots,S_r$ be a collection of disjoint subsets of $\tmset$, such that for all $1\leq j\leq r$, $|S_j|\geq r^2H_i/2$. Then there is a collection $\mset'=\set{(s_1,t_1),\ldots,(s_r,t_r)}$ of demand pairs, such that for all $1\leq j\leq r$, $(s_j,t_j)\in S_j$, and for all distinct $(s,t),(s',t')\in \mset'$, pair $(s',t')$ does not encircle pair $(s,t)$.
\end{lemma}

\begin{proof}
We perform $(r-1)$ iterations. At the beginning of iteration $\ell$, we are given a set $\mset'=\set{(s_1,t_1),\ldots,(s_{\ell-1},t_{\ell-1})}$ of demand pairs, where for all $1\leq j\leq \ell-1$, $(s_j,t_j)\in S_j$, and no two pairs in $\mset'$ encircle each other. Additionally, for each $\ell\leq j\leq r$, we are given a subset $S'_j\subseteq S_j$ of $(r^2-r(\ell-1))H_i/2$ demand pairs, such that no pair in $S'_j$ encircles a pair in $\mset'$, and no pair in $S'_j$ is encircled by a pair in $\mset'$. Therefore, at the end of iteration $(r-1)$, set $\mset'$ contains one pair from each set $S_1,\ldots,S_{r-1}$, and $S_r\neq \emptyset$. We add an arbitrary pair of $S_r$ to $\mset'$ to obtain the desired output.

At the beginning of the algorithm, $\mset'=\emptyset$, and for each $1\leq j\leq r$, set $S'_j$ contains any subset of $r^2H_i/2$ demand pairs of $S_j$. We now describe an execution of some iteration $\ell$. Let $\mset''=\bigcup_{j=\ell+1}^rS'_j$, and let $U$ be the set of all vertices appearing on lines $Q_t$, where $t$ is a destination vertex of a demand pair in $\mset''$. Then $|\mset''|\leq (r-1)(r^2-r(\ell-1))H_i/2$ and $|U|\leq (r-1)(r^2-r(\ell-1))H_i^2/4$,       while $S_{\ell}$ contains $(r^2-r(\ell-1))H_i/2$ demand pairs. Therefore, there is some demand pair $(s_{\ell},t_{\ell})\in S_{\ell}$ that contains at most $(r-1)H_i/2$ vertices of $U$. We add $(s_{\ell},t_{\ell})$ to $\mset'$. Consider now some set $S'_j$, for some $\ell+1\leq j\leq r$. Then pair $(s_{\ell},t_{\ell})$ may encircle at most $(r-1)H_i/2$ pairs of $S'_j$, and at most $H_i/2$ pairs in $S'_j$ may encircle $(s_{\ell},t_{\ell})$. We discard from $S_j'$ all pairs that either encircle $(s_{\ell},t_{\ell})$ or are encircled by it. At the end of this procedure, $|S_j'|\geq (r^2-r(\ell-1))H_i/2-rH_i/2\geq (r^2-r\ell)H_i/2$. If $|S'_j|>(r^2-r\ell)H_i/2$, then we discard arbitrary pairs from $S'_j$ until the equality holds.
\end{proof}

%
 
 \paragraph{Variable Gadget Analysis}
 
 We fix some variable $x$ and consider its corresponding gadget. We start with the following lemma.

\begin{lemma}\label{lem: var gadget analysis}
Let $\tmset_X=\tmset\cap \mset^X(x),\tmset_T=\tmset \cap\mset^T(x)$, and $\tmset_F=\tmset \cap \mset^F(x)$ be the subsets of $\mset^X(x),\mset^T(x)$, and $\mset^F(x)$, respectively, that are routed by our solution. Then at least one of the sets $\tmset_X,\tmset_T,\tmset_F$ is empty.\end{lemma}

\begin{proof}
Assume otherwise. Since we have discarded all excess pairs, each one of the three sets $\tmset_X,\tmset_T,\tmset_F$ contains at least $25H_i$ demand pairs. From Lemma~\ref{lem: enc resolution},  we can find three pairs, $(s_X,t_X)\in \tmset_X$, $(s_T,t_T)\in \tmset_T$ and $(s_F,t_F)\in \tmset_F$, with neither pair encircling the other. We now claim that it is impossible to route all three pairs simultaneously via node-disjoint paths. In order to do so, we use the following simple observation.

\begin{observation}\label{obs: cylinder} Let $\Sigma$ be the sphere, and let $D,D'$ be two disjoint closed simple discs on $\Sigma$, whose boundaries are denoted 
by $\Gamma$ and $\Gamma'$, respectively. Let $\Sigma'$ be the cylinder obtained from $\Sigma$ by removing the open discs $D\setminus \Gamma$, $D'\setminus \Gamma'$ from it. Assume further that we have three distinct points $a_1,a_2,a_3$ appearing on $\Gamma$ in this circular order, and three distinct points $a_1',a_3',a_2'$ appearing on $\Gamma'$ in this circular order (see Figure~\ref{fig: cylinder}). Let $\gamma_1,\gamma_2,\gamma_3$ be three simple curves on $\Sigma'$, where for each $1\leq j\leq 3$, $\gamma_j$ connects $a_j$ to $a_j'$. Then the three curves cannot be pairwise disjoint.
\end{observation}

\begin{proof}
Assume otherwise. We can twist the cylinder so that the curve $\gamma_1$ becomes a straight vertical line, and then cut the cylinder along this line, obtaining a square, with $a_2$ appearing to the left of $a_3$ on the top of the square, and $a_3'$ appearing to the left of $a_2'$ on the bottom of the square. It is now immediate to see that it is impossible to connect $a_2$ to $a_2'$ and $a_3$ to $a_3'$ via disjoint curves.

\begin{figure}[h]
\begin{center}
\scalebox{0.5}{\includegraphics{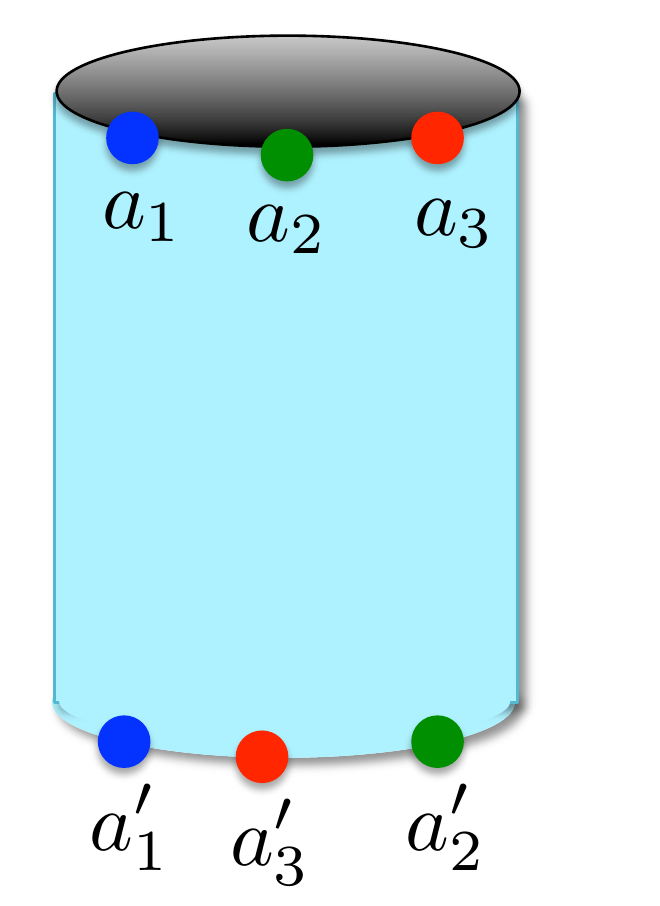}}
\caption{Illustration to Observation~\ref{obs: cylinder}.\label{fig: cylinder}}
\end{center}
\end{figure}
\end{proof}

Assume for contradiction that there are three node-disjoint paths, routing demand pairs $(s_X,t_X)$, $(s_T,t_T)$, and $(s_F,t_F)$. We embed the grid $G_{i+1}$ onto a sphere in a natural way, and then construct a curve $\gamma_X$, by concatenating the image of the path $P(s_X,t_X)$ with the line $Q_{t_X}$, denoting by $b_X$ the endpoint of this curve that is distinct from $s_X$. If this curve is not simple, we delete all loops from it until it becomes simple. We define curves $\gamma_T$, $\gamma_F$, and their endpoints $b_T$ and $b_F$ for the pairs  $(s_T,t_T)$, and $(s_F,t_F)$, respectively, in a similar way. Let $\Gamma$ be the top boundary of the grid, and let $\Gamma'$ be the union of the bottom boundaries of the boxes $B_T(x),B_F(x)$ and $B_X(x)$. By slightly thickening $\Gamma$ and $\Gamma'$, we can create two disjoint discs $D$ and $D'$ in the plane, with $\Gamma$ being the boundary of $D$ and $\Gamma'$ the boundary of $D'$, so that the curves $\gamma_X,\gamma_T$ and $\gamma_F$ are internally disjoint from $D$ and $D'$. Vertices $s_T,s_X$ and $s_F$ appear on $\Gamma$ in this circular order, while vertices $b_T,b_F$ and $b_X$ appear on $\Gamma'$ in this circular order. But then we obtain a collection of three disjoint curves on a cylinder that contradicts Observation~\ref{obs: cylinder}.
\end{proof}

The following corollary is now immediate.

\begin{corollary}\label{cor: var pairs}
$|\tmset^V|\leq (65h+1)\cdot n\cdot c_{i+1}N'_i$.
\end{corollary}

Consider some variable $x$ of $\phi$. If $\tmset\cap \mset^T(x)=\emptyset$, then we assign it the value \false, and otherwise we assign it the value \true.

Fix some variable $x\in X$ and some index $1\leq j\leq 5h+1$. We say that index $j$ is bad for variable $x$ if either (i) $x$ is assigned the value \true, and instance $\iset_j^T(x)$ is uninteresting; or (ii) $x$ is assigned the value \false, and instance $\iset_j^F(x)$ is uninteresting. We will later show that the total number of pairs $(x,j)$ where $j$ is a bad index for $x$ is bounded by $12n$.


\paragraph{Clause Gadget Analysis}
Consider a new clause $C^j_q\in \cset'$ and its three literals $\ell_{q_1},\ell_{q_2},\ell_{q_3}$. 
We say that clause $C_q^j$ is a \emph{troublesome clause}, or a \emph{troublesome copy of $C_q$},  iff there are at least two values $1\leq z<z'\leq 3$, for which instances $\iset_j(C_q,\ell_{q_z})$, $\iset_j(C_q,\ell_{q_{z'}})$ are both interesting.

\begin{lemma}\label{lemma: troublesome clauses}
For every original clause $C_q\in \cset$, at most three of its copies are troublesome.
\end{lemma}

\begin{proof}
Assume otherwise, and let $C_q\in \cset$ be a clause, such that at least four of its copies are troublesome. Then there are two indices $1\leq z<z'\leq 3$, and two indices $1\leq j<j'\leq h$, such that each one of the four instances $\iset_j(C_q,\ell_{q_z}),\iset_j(C_q,\ell_{q_{z'}}), \iset_{j'}(C_q,\ell_{q_z}),\iset_{j'}(C_q,\ell_{q_{z'}})$ is interesting, and so each of the
 four sets $\tmset(C^j_q,\ell_{q_z}),\tmset(C^j_q,\ell_{q_{z'}}), \tmset(C^{j'}_q,\ell_{q_z}),\tmset(C^{j'}_q,\ell_{q_{z'}})$ contains at least $25H_i$ demand pairs. From Lemma~\ref{lem: enc resolution}, we can find demand pairs $(s_1,t_1)\in \tmset(C^j_q,\ell_{q_z})$,  $(s_2,t_2)\in \tmset(C^j_q,\ell_{q_{z'}})$, $(s_1',t_1')\in \tmset(C^{j'}_q,\ell_{q_{z}})$, and $(s_2',t_2')\in \tmset(C^{j'}_q,\ell_{q_{z'}})$, such that none of the four resulting pairs encircles another. Notice that the sources of these pairs appear on the top boundary of the grid in one of the following two orders: $(s_1,s_1',s_2,s_2')$, if the variable  corresponding to $\ell_{q_z}$ appears before the variable corresponding to $\ell_{q_{z'}}$ in the ordering $(x_1,x_2,\ldots,x_n)$ of the variables, or $(s_2,s_2',s_1,s_1')$ otherwise.

We draw a curve $\gamma_1$ by concatenating the image of the path $P(s_1,t_1)$ and the line $Q_{t_1}$. We draw curves $\gamma_1'$, $\gamma_2$ and $\gamma_2'$ for the remaining pairs similarly. On the one hand, none of the the resulting curves may cross each other, while on the other hand the vertices lying on the bottom boundary of the box $B(C_q)$ have been deleted, which is impossible.
\end{proof}

\begin{corollary}\label{cor: clause pairs}
For each original clause $C_q\in \cset$, $|\tmset(C_q)|\leq (6+h)N'_ic_{i+1}$, and overall $|\tmset^C|\leq 5n(6+h)N'_ic_{i+1}/3$.
\end{corollary}

\begin{proof}
Consider some original clause $C_q\in \cset$. Each of its non-troublesome copies $C^j_q$ contributes at most $c_{i+1}N'_i$ demand pairs to $\tmset^C$, while a troublesome copy may contribute at most $3c_{i+1}N'_i$ demand pairs. Since at most three copies are troublesome, overall $|\tmset(C_q)|\leq (6+h)N'_ic_{i+1}$ and $|\tmset^C|\leq m(6+h)N'_ic_{i+1}= 5n(6+h)N'_ic_{i+1}/3$, since $m=5n/3$.
\end{proof}

In the rest of our proof, we will reach a contradiction by proving that the current assignment to the variables of $\phi$ satisfies more than $(1-\eps)hm$ clauses in $\cset'$. In order to do so, we gradually discard clauses from $\cset'$, until we obtain a large enough subset of clauses that is guaranteed to be satisfied by the current assignment.

Our first step is to define uninteresting clauses. Recall that for each new clause $C_q^j\in \cset'$, there are three corresponding $c_{i+1}$-wide level-$i$ instances, $\iset_j(C_q,\ell_{q_1}),\iset_j(C_q,\ell_{q_2})$, and $\iset_j(C_q,\ell_{q_3})$. We say that clause $C_q^j$ is interesting iff at least one of these three instances is interesting, and we say that it is uninteresting otherwise. Let $\cset_0'\subseteq \cset'$ be the set of all  uninteresting clauses.

\begin{claim}\label{claim: uninteresting}
$|\cset_0'|\leq 12n$.
\end{claim}
\begin{proof}
Assume otherwise. From Corollary~\ref{cor: var pairs}, $|\tmset^V|\leq (65h+1)n c_{i+1}N'_i$, while from Corollary~\ref{cor: clause pairs}, it is easy to see that $|\tmset^C|\leq \frac{5n(6+h)N_i'c_{i+1}}{3}-|\cset_0'|N'_ic_{i+1}\leq  \frac{5n(6+h)N_i'c_{i+1}}{3}-12 n N'_ic_{i+1}=nN_i'c_{i+1}(\frac{5h}{3}-2)$. But then:

\[|\tmset|\leq (65h+1)n c_{i+1}N'_i + \left(\frac{5h}{3}-2\right ) nN_i'c_{i+1}\leq \left(\frac{200h}{3}-1\right )nN_i'c_{i+1}<(1-2\delta)\left(\frac{200h}{3}+1\right )nN_i'c_{i+1},\]
 since $h=1000/\eps$ and $\delta=8\eps^2/10^{12}$,
a contradiction.
\end{proof}

Consider  some clause $C^j_q\in \cset'$ that is interesting. Then there is an index $z\in \set{1,2,3}$, such that instance $\iset_j(C_q,\ell_{q_z})$ is interesting. If there are several such indices $z$ (if $C^j_q$ is troublesome), then we choose one of them arbitrarily. We say that clause $C_q^j$ \emph{chooses} the literal $\ell_{q_z}$.
We say that $C^j_q$ is a \emph{cheating clause} iff the variable $x_{q_z}$ corresponding to literal $\ell_{q_z}$ is assigned the opposite value: In other words, if $\ell_{q_z}=x_{q_z}$, then $x$ is assigned the value \false, and otherwise, $\ell_{q_z}=\neg x_{q_z}$, and $x_{q_z}$ is assigned the value \true. We further say that it is a \emph{bad cheating clause} iff at least one of the indices $j,j+1$ is bad for the variable $x_{q_z}$, and we say that it is a \emph{good cheating clause} otherwise. Let $\cset_1'\subseteq \cset'\setminus \cset_0'$ be the set of all the cheating clauses. In the following lemma we bound the number of the cheating clauses.

\begin{lemma}\label{lem: cheating clauses}
There are at most $24n$ bad cheating clauses, and at most $3m$ good cheating clauses.
\end{lemma}

We provide the proof of the lemma below, after we complete the analysis of the \ni using it. Notice that if a clause $C_q^j$ is an interesting non-cheating clause, then the current assignment must satisfy it. From Claim~\ref{claim: uninteresting} and Lemma~\ref{lem: cheating clauses}, there are at least $hm-12n-24n-3m=hm-123m/5=(1-123\eps/5000)hm>(1-\eps)hm$ such clauses, contradicting Observation~\ref{obs: ni-satisfied new clauses}. It now remains to complete the proof of Lemma~\ref{lem: cheating clauses}.

\begin{proofof}{Lemma~\ref{lem: cheating clauses}}
In order to bound the number of bad cheating clauses, we first bound the total number of bad indices for variables. Let $\beta$ denote the total number of pairs $(x,j)$, where $x$ is a variable of $\phi$ and $1\leq j\leq 5h+1$ is an index, such that $j$ is a bad index for $x$. Recall that, assuming that $x$ is assigned the value \true, this means that instance $\iset_j^T(x)$ is uninteresting, and if $x$ is assigned the value \false, then instance $\iset_j^F(x)$ is uninteresting. For every variable $x$, let $\beta(x)$ be the total number of indices $1\leq j\leq 5h+1$ that are bad for $x$. Then $|\tmset(x)|\leq (65h+1)N'_ic_{i+1}-\beta(x)\cdot N'_ic_{i+1}$. Therefore, overall,

\[\begin{split}
|\tmset|&\leq n \cdot (65h+1)N'_ic_{i+1}-\beta\cdot N'_ic_{i+1}+|\tmset^C|\\
&\leq  n \cdot (65h+1)N'_ic_{i+1}-\beta\cdot N'_ic_{i+1}+5n(h+6)N'_ic_{i+1}/3\\
&=(200h/3+11)nN'_ic_{i+1} -\beta\cdot N'_ic_{i+1}.\end{split}\]

Assume now for contradiction that $\beta>12n$. Then $|\tmset|<(200h/3-1)nN'_ic_{i+1}<(1-2\delta)(200h/3+1)nN'_ic_{i+1}$, a contradiction.
It is immediate to verify that each pair $(x,j)$ where $j$ is a bad index for variable $x$ may be responsible for at most two bad cheating clauses, and so the total number of bad cheating clauses is bounded by $24n$.

We now turn to bound the number of good cheating clauses. In order to do so, it is enough to prove the following claim:

\begin{claim}\label{claim: good cheating}
For each original clause $C_q\in \cset$ and index $z\in \set{1,2,3}$, at most one copy of $C_q$  that chooses the literal $\ell_{q_z}$ is a good cheating clause.
\end{claim}

Notice that from the above claim, for each original clause $C_q\in \cset$, there are at most $3$ copies of $C_q$ that are good cheating clauses, and so overall there are at most $3m$ good cheating clauses in $\cset'\setminus\cset'_0$. It now remains to prove Claim~\ref{claim: good cheating}.

\begin{proofof}{Claim~\ref{claim: good cheating}}
Let $C_q\in \cset$ be an original clause, and let $\ell\in \set{\ell_{q_1},\ell_{q_2},\ell_{q_3}}$ be one of its literals. Assume for contradiction that two distinct copies of $C_q$, that we denote by $C_q^j$ and $C_q^{j'}$ are good cheating clauses and that they both select the literal $\ell$ of $C_q$. Let $x$ be the variable corresponding to $\ell$. We assume without loss of generality that $x$ appears in $C_q$ without negation (so $\ell=x$), but we assigned $x$ the value \false; the other case is dealt with similarly. We also assume without loss of generality that $j<j'$.

Since both $C_q^j,C_q^{j'}$ are good cheating clauses, indices $j,j+1$ and $j'+1$ are good indices for $x$. Therefore, each of the instances $\iset_j^F(x),\iset_{j+1}^F(x)$ and $\iset_{j'+1}^F(x)$ is an interesting instance, and $\tmset$ contains at least $25H_i$ demand pairs that belong to each of the three instances. Moreover, $\tmset$ contains at least $25H_i$ demand pairs that belong to each of the two instances $\iset_j(C_q,\ell)$ and $\iset_{j'}(C_q,\ell)$, from the choice of the literal $\ell$ by each corresponding copies of $C_q$.

We let $\mset_1$ be the set of all demand pairs of $\iset_{j}(C_q,\ell)$ that belong to $\tmset$, and we define $\mset_2$ for instance $\iset_{j'}(C_q,\ell)$ similarly.
We also let $\mset_3,\mset_4$ and $\mset_5$ be the sets of the demand pairs from instances $\iset_{j}^F(x)$, $\iset_{j+1}^F(x)$, and $\iset_{j'+1}^F(x)$ that belong to $\tmset$, respectively. From the above discussion, each of the five sets contain at least $25H_i$ demand pairs. From Lemma~\ref{lem: enc resolution}, we can find, for each $1\leq y\leq 5$, a demand pair $(s_y,t_y)\in \mset_y$, such that none of the five resulting demand pairs encircle each other. 

Consider the plane into which the grid $G_i$ is embedded with only the following curves. First, we add the boundary of the grid. We also add the bottom boundary of the box $B_F(x)$, and for every $3\leq y\leq 5$, a curve $\gamma_y$, obtained by concatenating the image of $P(s_y,t_y)$,  with the line $Q_{t_y}$. Notice that the three curves cannot intersect. Let $f$ and $f'$ be the faces of the resulting drawing, that are distinct from the outer face, and are incident to the intervals $Y_{j}^F$ and $Y_{j'}^F$, respectively. Let $\Gamma$ be the bottom boundary of the box $B(C_q)$. Then $\Gamma$ must lie in the same face of the drawing as $t_1$ and $t_2$, since none of the five pairs encircles the other. Assume without loss of generality that this face is $f^*\neq f$. Then the source of the pair $(s_1,t_1)$ lies on the part $Y_{j}^F$ of boundary of $f$ that is also incident to the infinite face, while its destination lies strictly inside another face (and this other face is disjoint from $Y_{j}^F$), which is impossible. \end{proofof} \end{proofof}

\newpage


\appendix

\label{--------------------------------------------Appendix-----------------------------------------------------}

\section*{Appendix}
\label{-----------------------------------------sec: proofs from YI sec ----------------------------}
\section{Proofs Omitted from Section~\ref{sec: YI}}
\label{sec: proofs from YI sec}

 \subsection{Proof of Claim~\ref{claim: routing in a snake}}\label{subsec: proof of routing in snake claim}
 The proof is by induction on $\ell$. For the base case, assume that $\ell=1$, and let $A,A'$ any pair of vertex sets (that are not necessarily disjoint), such that the vertices of each set lie on a single boundary edge of the corridor $\Y=\Y_1$, and $|A|=|A'|=w'\leq w-2$. We show that there is a set $\qset$ of node-disjoint paths in $\Y$, connecting every vertex of $A$ to a vertex of $A'$, with a slightly stronger property: namely, the paths in $\qset$ are internally disjoint from the boundary of $\Y$. Let $\Y'$ be the graph obtained from $\Y$, by deleting all vertices lying on the boundary of $\Y$, except for the vertices of $A\cup A'$. It is enough to show that there is a set $\qset$ of $w'$ node-disjoint paths in $\Y'$, connecting vertices of $A$ to vertices of $A'$. Assume for contradiction that such a set of paths does not exist. Then from Menger's theorem, there is a set $J$ of $w'-1$ vertices, such that $\Y'\setminus J$ has no path connecting a vertex of $A$ to the vertex of $A'$. We consider three cases.
 
 The first case is when $A$ and $A'$ lie on the opposite boundary edges of $\Y$. Assume without loss of generality that the vertices of $A$ lie on the top boundary edge, and the vertices of $A'$ on the bottom boundary edge of $\Y$. Let $\rset$ and $\wset$ be the sets of rows and columns of $G_{i+1}$, respectively, that span $\Y$. Let $\wset',\wset''\subseteq \wset$ be the sets of columns of $\Y$ where the vertices of $A$ and $A'$ lie, respectively, so $|\wset'|=|\wset''|=w'$. Since $|J|<w'$, there is a column $W'\in \wset'$ and a column $W''\in \wset'$, with $W'\cap J,W''\cap J=\emptyset$. There is also some row $R'\in \rset$ of $\Y'$, that is not its top or bottom row, such that $R'\cap J=\emptyset$. But $W'\cup W''\cup R'$ is a connected subgraph of $\Y'\setminus J$, that contains a vertex of $A$ and a vertex of $A'$, a contradiction.
 
 The other two cases, when $A$ and $A'$ lie on adjacent boundary edges of $\Y$, and when $A$ and $A'$ lie on the same boundary edge of $\Y$ are analyzed similarly.

Assume now that the claim holds for some value $\ell\geq 0$. We now prove it for $\ell+1$. Denote $|A|=|A'|=w'$, and let $U$ be any set of $w'$ vertices in $\Y_{\ell}\cap \Y_{\ell+1}$, so the vertices of $U$ lie on the boundaries of both corridors. Notice that the vertices of $U$ must belong to a single  boundary edge of $\Y_{\ell}$, and a single boundary edge of $\Y_{\ell+1}$. Using the induction hypothesis, there is a set $\pset_1$ of $w'$ node-disjoint paths in $\Y_1\cup\cdots\cup \Y_{\ell}$, connecting every vertex of $A$ to a distinct vertex of $U$. From our analysis of the base case, there is a set $\pset_2$ of $w'$ node-disjoint paths in $\Y_{\ell+1}$, connecting every vertex of $U$ to a distinct vertex of $A'$. Moreover, the paths in $\pset_2$ are internally disjoint from the boundary of $\Y_{\ell+1}$. We obtain the desired set of paths by concatenating the paths in $\pset_1$ and the paths in $\pset_2$.

\subsection{Proof of Claim~\ref{claim: step 1 of routing}}\label{subsec: proof of routing claim}
 
 Consider two consecutive variables $x_j,x_{j+1}$, for $1\leq j<n$. Recall that the vertices $S(\hmset(x_j))$ appear consecutively on $R$, and so do the vertices of $S(\hmset(x_{j+1}))$. Let $V_j\subseteq S(\hmset^C)$ be the set of all source vertices that correspond to clause-pairs, and lie between the vertices of $S(\hmset(x_j))$ and the vertices of $S(\hmset(x_{j+1}))$ on $R$. Notice that the vertices of $V_j$ may only correspond to clauses in which $x_j$ or $x_{j+1}$ participate, so $|V_j|\leq 10hN_ic_{i+1}<N_{i+1}$. 
Similarly, we let $V_0,V_n\subseteq S(\hmset^C)$ be the sets of all source vertices that correspond to clause-pairs and lie before the vertices of $S(\hmset(x_1))$ and after the vertices of $S(\hmset(x_n))$ on $R$, respectively. We still have $|V_0|,|V_n|\leq N_{i+1}$. For all $0\leq j\leq n$, let $M'_j=|V_j|$. For all $1\leq j\leq n$, let $M_j=|\hmset(x_j)|=c_{i+1}N_i(65h+1)<N_{i+1}$. 
 
 We now select a set $\Gamma$ of vertices on the top row of $B^V$, as follows. First, for every variable $x_j$ of $\phi$, for every vertex $a\in A(x_j)$, we select a vertex $a'$ on the top row of $B^V$, lying in the same column as $a$, and  we let $P_a$ be the sub-path of the column $\col(a)$ between $a$ and $a'$.   Let $\Gamma_j$ denote the resulting set of vertices that we have selected for $x_j$.
 
 For $1\leq j<n$, let $\wset_j$ be the set of columns of $B^V$ that lie between the boxes $B(x_j)$ and $B(x_{j+1})$. We also let $\wset_0$ be the set of all columns of $B^V$ that lie before $B(x_1)$, and $\wset_n$ the set of all columns lying after $B(x_n)$. For all $0\leq j\leq n$,  we select an arbitrary set $\Gamma'_j$ of $M'_j$ distinct vertices on the top row of $B^V$ that belong to the columns of $\wset_j$. Since $|\wset_j|\geq N_{i+1}$, while $M'_j\leq N_{i+1}$, such a set of vertices exists. For each selected vertex $v\in \Gamma'_j$, we let $P_v$ be the column of $B^V$ in which $v$ lies, and we let $v'$ be the other endpoint of the column. Let $\Gamma''_j=\set{v'\mid v\in \Gamma'_j}$. 
 We denote $\Gamma=\left(\bigcup_{j=1}^n\Gamma_j\right )\cup \left(\bigcup_{j=0}^n\Gamma'_j\right )$, and $\Gamma''=\bigcup_{j=0}^n\Gamma''_j$.
 Finally, we let $\Gamma'''$ be the set of $|\hmset^C|$ consecutive left-most vertices on the top boundary of $B^C$. 
 
 We construct five sets of node-disjoint paths, $\pset_0,\ldots,\pset_4$. The final set $\pset'$ of paths will be obtained by combining these  sets of paths. 
 
 \paragraph{Set $\pset_0$:} Let $Z'$ be the set of $|\hmset|$ leftmost vertices on the opening of $B(\iset)$. Set $\pset_0$ consists of $|\hmset|$ node-disjoint paths, connecting every vertex of $S(\hmset)$ to a vertex of $Z'$, so that the paths in $\pset_0$ are order-preserving and  internally disjoint from $B(\iset)$. Since the distance between $B(\iset)$ and the boundaries of $G_{i+1}$ is at least $H_{i+1}\geq N_{i+1}$, while $|\hmset|=N_{i+1}$, it is easy to verify that such a set of paths exists.

\paragraph{Set $\pset_1$:} This set of paths connects every vertex of $Z'$ to a distinct vertex of $\Gamma$ in a node-disjoint and order-preserving manner. In order to construct it, we use a snake $\yset^1$, consisting of two corridors, $\Y_1^1$ and $\Y^1_2$. The first corridor $\Y_1^1$ is simply the set of the top $N_{i+1}+2$ rows of $B(\iset)$. In order to construct the second corridor, $\Y_2^1$, we denote by $\wset^V$ the set of all columns of $G_{i+1}$ that intersect the box $B^V$, and we denote by $\hat R$ the row of $G_{i+1}$ that contains the topmost row of $B^V$. Let $\hat{\rset}$ be a consecutive set of rows of $B(\iset)$, starting from row $R_{N_{i+1}+2}$ and ending at row $\hat R$. Let $\Y_2^1$ be the corridor spanned by the set $\wset^V$ of columns and the set $\hat {\rset}$ of rows (see Figure~\ref{fig: yi-routing-appendix}). By combining the two corridors, we we obtain a snake $\yset^1$ of width at least $N_{i+1}+2$. From Claim~\ref{claim: routing in a snake}, there is a set $\pset_1$ of node-disjoint paths, connecting every vertex of $Z'$ to a distinct vertex of $\Gamma$ inside the snake. It is immediate to verify that the set $\pset_1$ of paths must be order-preserving.

 \begin{figure}[h]
 \scalebox{0.6}{\includegraphics{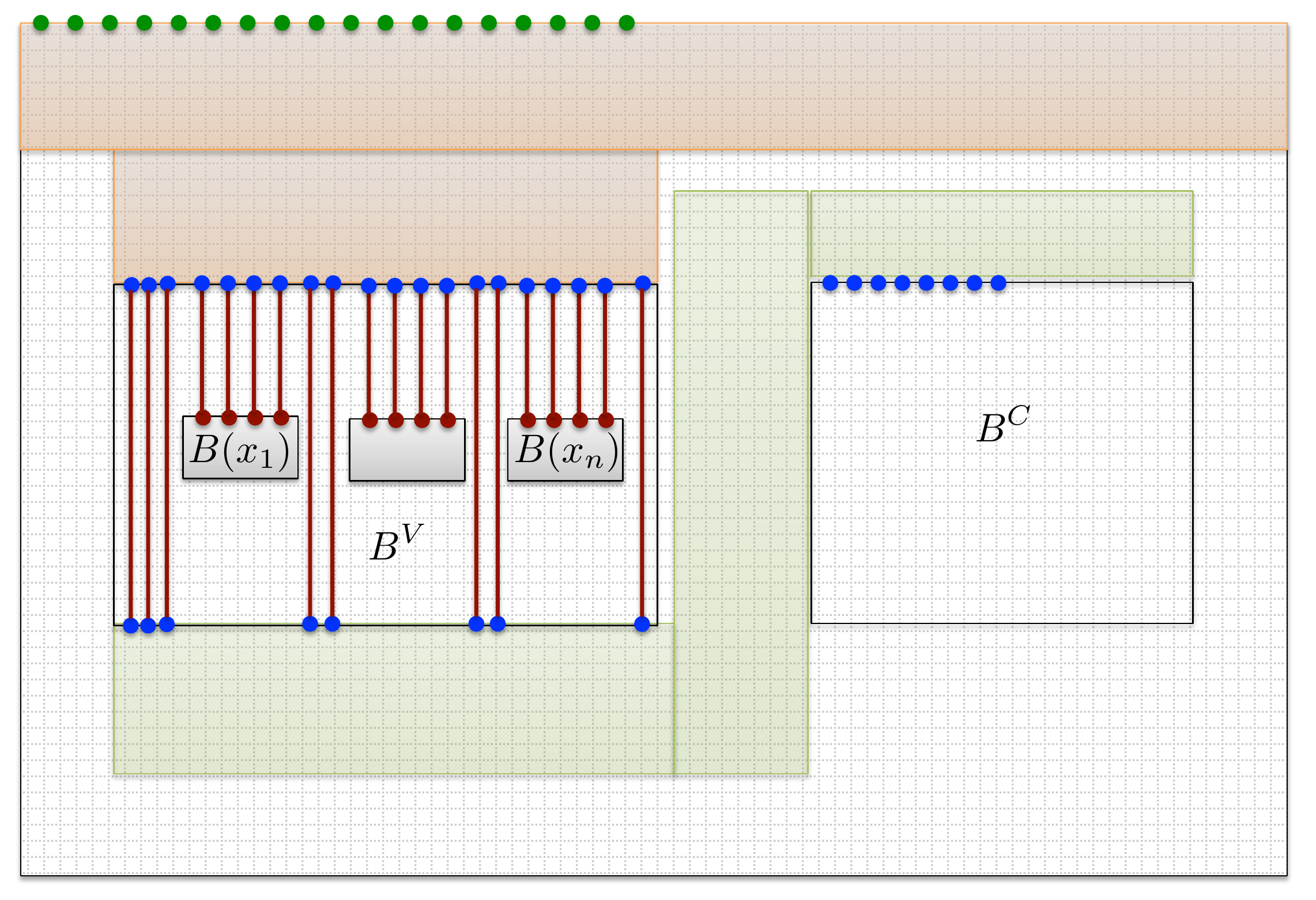}}
 \caption{Routing the sets $\pset_1,\pset_2$ and $\pset_3$ of paths inside $B(\iset)$. The paths in $\pset_2$ are shown in red; the paths of $\pset_1$ are routed inside the orange snake, and the paths of $\pset_3$ are routed inside the green snake.\label{fig: yi-routing-appendix}}
 \end{figure}
 
 \paragraph{Set $\pset_2$:} Set $\pset_2$ contains, for every vertex $v\in \Gamma$, the path $P_v$ that we have defined above. Note that for each $1\leq j\leq n$, paths in $\pset_2$ connect the vertices of $\Gamma_j$ to the vertices of $A(x_j)$, lying on the top row of $B(x_j)$, and for each $0\leq j\leq n$, paths in $\pset_2$ connect the vertices of $\Gamma_j'$ to the vertices of $\Gamma''_j$, lying on the bottom row of $B^V$. 
 
\paragraph{Set $\pset_3$:} This set contains $|\hmset^C|$ node-disjoint paths, connecting every vertex of $\Gamma''$ to a distinct vertex of $\Gamma'''$. The paths in $\pset_3$ will be internally disjoint from $B^V$ and $B^C$, and order-preserving. In order to construct the set $\pset_3$ of paths, we construct a snake $\yset^2$. Let $\wset^V$ and $\wset^C$ be the sets of columns of $G_{i+1}$  that intersect $B^V$ and $B^C$, respectively, and let $\wset^M$ be the set of columns lying between these two sets. Let $R'$ be the row of $G_{i+1}$ containing the bottommost row of $B^V$, and let $\rset^V$ be the set of $N_{i+1}-2$ consecutive rows of $G_{i+1}$ lying below $R'$, including $R'$. Let $R''$ be the row of $G_{i+1}$ containing the topmost row of $B^C$, and let $\rset^C$ be the set of   $N_{i+1}-2$ consecutive rows of $G_{i+1}$ lying above $R^C$, including $R^C$. Finally, let $\rset^M$ be the set of rows lying between the top row of $\rset^C$ and the bottom row of $\rset^V$, including these two rows. The snake $\yset^2$ consists of three corridors (see Figure~\ref{fig: yi-routing-appendix}). The first corridor, $\Y_1^2$, is spanned by the set $\rset^V$ of rows and the set $\wset^V$ of columns. The second corridor, $\Y_2^2$, is spanned by the set $\rset^M$ of rows and the set $\wset^M$ of columns. The third and the final corridor, $\Y_3^2$, is spanned by the set $\rset^C$ of rows and the set $\wset^C$ of columns. We extend corridor $\Y_1^2$ by one column to the right, so that it intersects the boundary of $\Y_2^2$, and similarly we extend corridor $\Y_3^2$ by one column to the left.  It is now easy to verify that we obtain a snake of width at least $N_{i+1}-2$, and so from Claim~\ref{claim: routing in a snake}, there is a set $\pset_3$ of node-disjoint paths, connecting every vertex of $\Gamma''$ to a distinct vertex of $\Gamma'''$.
It is immediate to verify that the paths in $\pset_3$ are order-preserving.

 By combining the paths of $\pset_0,\ldots,\pset_3$, we obtain a set $\pset''$ of node-disjoint paths, that are internally disjoint from $B^C$. For every variable $x_j$ of $\phi$, set $\pset''$ contains a set $\pset(x_j)$ of paths connecting every vertex of $S(\hmset(x_j))$ to a vertex of $A(x_j)$, so that the paths in $\pset(x_j)$ are order-preserving. For every clause $C_q\in \cset$, set $\pset''$ contains a set $\pset'(C_q)$ of node-disjoint order-preserving paths, connecting every vertex of $S(\hmset(C_q))$ to a distinct vertex of $\Gamma'''$. We denote by $\Gamma(C_q)\subseteq\Gamma'''$ the set of vertices that serve as endpoints of the paths in $\pset'(C_q)$. Note that for each clause $C_q\in \cset$, the vertices of $\Gamma(C_q)$ appear consecutively in $\Gamma'''$.
 We define an ordering $\oset'$ of the clauses in $\cset$ as follows: clause $C_q$ appears before clause $C_{q'}$ in this ordering iff the vertices of $\Gamma(C_q)$ appear to the left of the vertices of $\Gamma(C_{q'})$ on the top row of $B^C$.

\paragraph{Set $\pset_4$:}  In our final step, we construct, for each clause $C_q\in \cset$, a set $\pset''(C_q)$ of paths, connecting the vertices of $\Gamma(C_q)$ to the vertices of $A(C_q)$, so that the paths in $\pset''(C_q)$ are order-preserving and the paths in $\bigcup_{C_q\in \cset}\pset''(C_q)$ are node-disjoint. We then let $\pset_4$ be the union of these sets of paths.
 
 \begin{claim}\label{claim: snake-like routing}
 For each clause $C_q\in \cset$, there is a set $\pset''(C_q)\subseteq B^C$ of paths, connecting every vertex of $\Gamma(C_q)$ to a distinct vertex of $A(C_q)$, such that the paths in set $\pset_4=\bigcup_{C_q\in \cset}\pset''(C_q)$ are mutually node-disjoint and internally disjoint from the boxes $\set{B(C_1),\ldots,B(C_m)}$. Moreover, the paths in $\pset_4$ are internally disjoint from the top boundary of $B^C$, and for each $1\leq q\leq m$, the paths in $\pset(C_q)$ are order-preserving.
 \end{claim}
 
 We obtain the final set $\pset'$ of paths by combining the paths of $\pset_0,\ldots,\pset_4$. In order to complete the proof of Claim~\ref{claim: step 1 of routing}, it now remains to prove Claim~\ref{claim: snake-like routing}.
 
 \begin{proof}
 Consider some clause $C_q\in \cset$, and let $\wset(C_q)$ be the set of columns of $G_{i+1}$ that intersect $B(C_q)$. We let $\wset^L(C_q)$ be the set of $N_{i+1}$ columns lying immediately to the left of $\wset(C_q)$, and similarly we let $\wset^R(C_q)$ be the set of $N_{i+1}$ columns lying immediately to the right of $C_q$. Let $\rset$ be the set of all rows of $G_{i+1}$ that intersect the box $B(C_q)$ (this set is the same for all clauses $C_q$). We let $B'(C_q)$ be the sub-grid of $G_{i+1}$ spanned by the set $\wset^L(C_q)\cup \wset(C_q)\cup \wset^R(C_q)$ of columns and the set $\rset$ of rows. 
 
 The idea is to define a snake-like routing, where the paths visit the boxes $B'(C_1),\ldots,B'(C_m)$ in turn. For each clause $C_q$, only the paths corresponding to the clauses $C_q,C_{q+1},\ldots,C_m$ will visit the box $B'(C_q)$. The paths corresponding to clause $C_q$ will terminate on the top boundary of $B(C_q)$. For each $q+1\leq q'\leq m$, if $C_{q'}$ appears before $C_q$ in the ordering $\oset'$, then the paths corresponding to $C_{q'}$ will traverse the box $B'(C_q)$ to the left of $B(C_q)$; otherwise, they will traverse the box $B'(C_q)$ to the right of $B(C_q)$. In order to implement this, we select a subset of vertices on the top and the bottom boundaries of $B'(C_q)$, through which the paths will enter and leave the box (see Figure~\ref{fig: yi-clauses1}).

 \begin{figure}[h]
 \scalebox{0.6}{\includegraphics{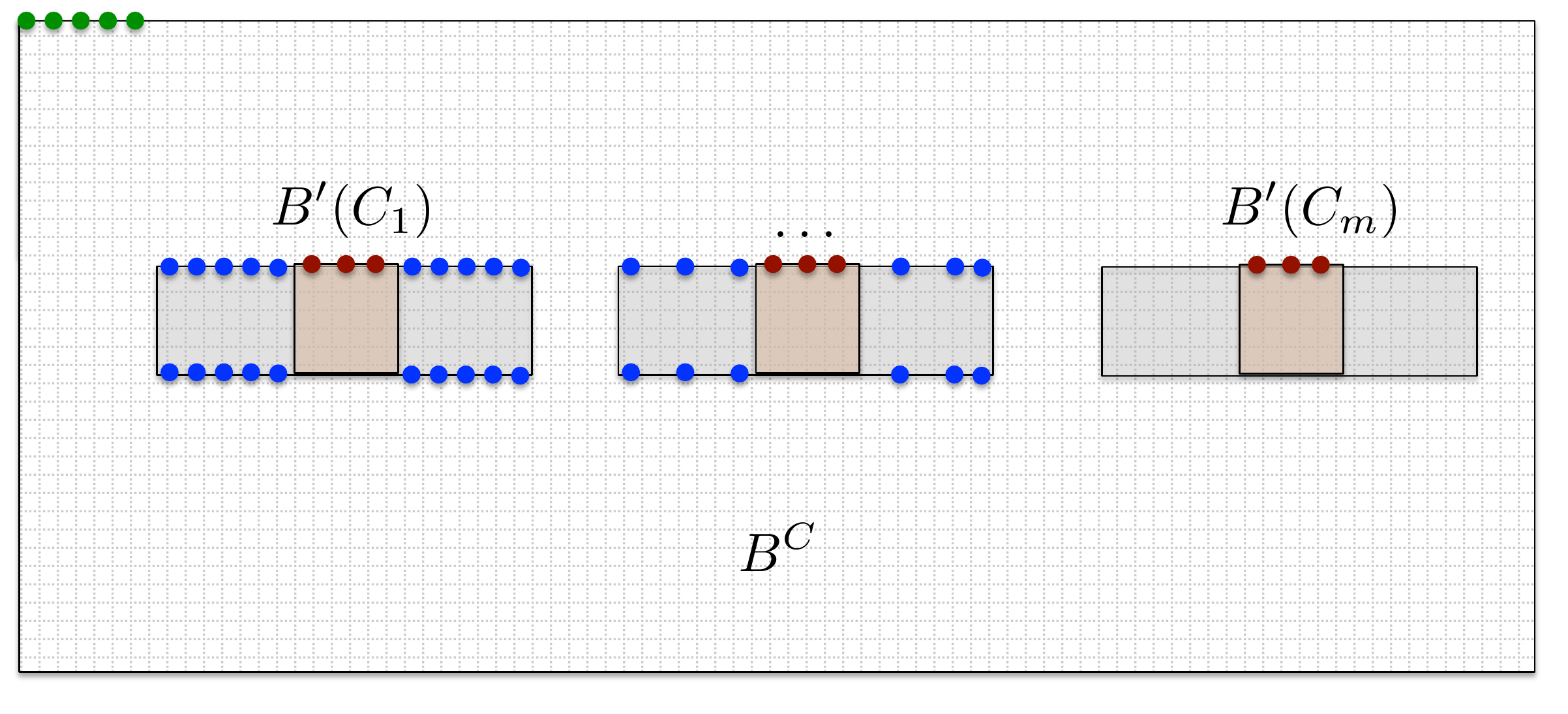}}
 \caption{Selecting the sets $\Gamma(C_q)$ and $\Gamma'(C_q)$ of vertices for the clauses $C_q$. The vertices of $\Gamma'''$ are shown in green.\label{fig: yi-clauses1}}
 \end{figure}
 
 Fix some clause $C_q$, for $1\leq q\leq m$. Let $n^L_q$ be the number of clauses $C_{q'}$, with $q<q'\leq m$, such that $C_{q'}$ appears before $C_q$ in the ordering $\oset'$, and let $n^R_q$ be the number of such clauses with $C_{q'}$ appearing after $C_q$ in $\oset'$. Recall that for each clause $C_{q'}$, $|\hmset(C_q')|=hc_{i+1}N_i$, and $mhc_{i+1}N_i<N_{i+1}-2$. We select an arbitrary set $\Gamma^L(C_q)$ of $n^L_q\cdot hc_{i+1}N_i$ vertices on the top boundary of $B'(C_q)$, strictly to the left of $B(C_q)$, and similarly, we select an arbitrary set $\Gamma^R(C_q)$ of $n^R_q\cdot hc_{i+1}N_i$ vertices on the top boundary of $B'(C_q)$, strictly to the right of $B(C_q)$. For each vertex $v\in \Gamma^L(C_q)\cup \Gamma^R(C_q)$, we let $P_v$ be the column of $B'(C_q)$ that contains $v$, and we let $v'$ be its other endpoint. For each vertex $v\in A(C_q)$, we let $P_v$ be a path consisting of a single vertex $v$. We define $\Gamma(C_q)=\Gamma^L(C_q)\cup A(C_q)\cup \Gamma^R(C_q)$, a set of vertices on the top boundary of $B'(C_q)$. We define $\Gamma'(C_q)=\set{v'\mid v\in \Gamma^L(C_q)\cup \Gamma^R(C_q)}$, a set of vertices on the bottom boundary of $B'(C_q)$. Finally, we let $\tpset_q'=\set{P_v\mid v\in \Gamma(C_q)}$. It is easy to verify that for all $1\leq q<m$, $|\Gamma'(C_q)|=|\Gamma(C_{q+1})|$.
 
 In order to compute the routing, we use the following simple observation.
 
 \begin{observation}\label{obs: snake-like routing}
 For all $1\leq q< m$, there is a set $\tpset_q\subseteq B^C$ of paths, such that:
 
 \begin{itemize}
 \item Paths in $\tpset_1$ connect every vertex in $\Gamma'''$ to a distinct vertex of $\Gamma(C_1)$ and are order-preserving;
 
 \item For $1<q< m$, paths in $\tpset_q$ connect every vertex in $\Gamma'(C_q)$ to a distinct vertex of $\Gamma(C_{q+1})$ and are order-preserving;
 
 \item The paths in $\bigcup_q\tpset_q$ are mutually node-disjoint, and they are internally disjoint from all boxes in $\set{B'(C_1),\ldots,B'(C_m)}$.
 \end{itemize}
 \end{observation}
 
 Notice that combining the paths in sets $\bigcup_{q=1}^{m-1}\tpset_q$ and $\bigcup_{q=1}^m\tpset'_q$ finishes the proof of Claim~\ref{claim: snake-like routing} and Claim~\ref{claim: step 1 of routing}. We now prove the observation.
 
 \begin{proof}
 We construct, for each $1\leq q\leq m$, a snake $\yset(C_q)$, and route the set $\tpset_q$ of paths inside it (see Figure~\ref{fig: yi-clauses2} for an illustration). Since $|\hmset^C|<N_{i+1}-2$, it is enough to ensure that each snake has width at least $N_{i+1}$.

 \begin{figure}[h]
 \scalebox{0.6}{\includegraphics{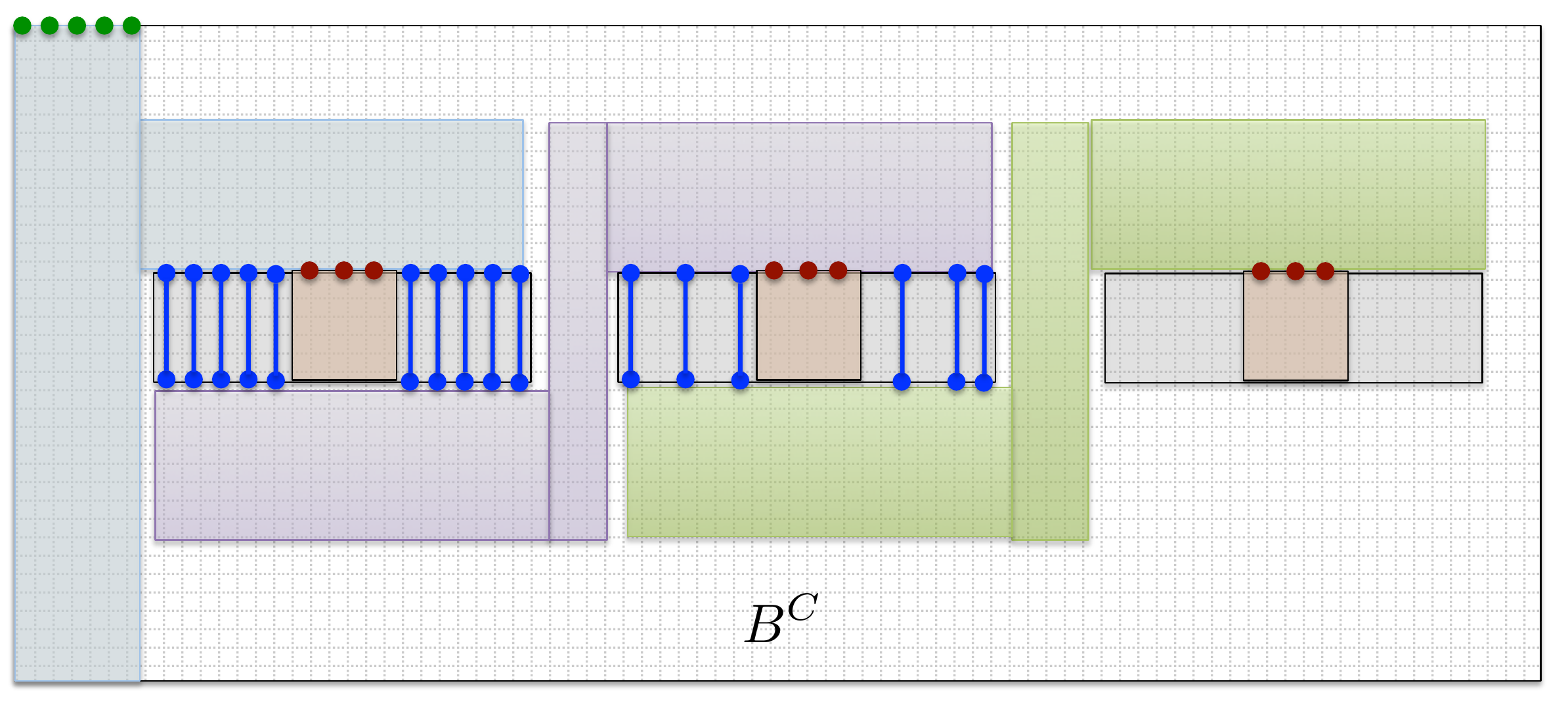}}
 \caption{The snakes used to route the sets $\tilde \pset_q$ of paths.\label{fig: yi-clauses2}}
 \end{figure}

 Recall that $\rset$ is the set of all rows of $G_{i+1}$ that intersect the boxes $B(C_q)$. Let $R'$ and $R''$ be the topmost and the bottommost rows of $\rset$, respectively. Let $\rset'$ be the set of $N_{i+1}$ consecutive rows lying above 
 $R'$ including $R'$, and let $\rset''$ be the set of $N_{i+1}$ consecutive rows lying below $R''$, including $R''$.
 
 The first snake, $\yset(C_1)$ consists of two corridors. The first corridor, $\Y_1(C_1)$, is spanned by the first $N_{i+1}$ columns of $B^C$. For the second corridor, $\Y_2(C_1)$, we let $W$ be the last column of $\Y_1(C_1)$, $W'$ the last column of $\wset^R(C_1)$, and $\wset'$ the set of all columns lying between $W$ and $W'$, including these two columns. We then let $\Y_2(C_1)$ be the corridor spanned by the columns of $\wset'$ and the rows of $\rset'$.
 
For each $1<q\leq m$, we construct a snake $\yset(C_q)$, that consists of three corridors. The first corridor, $\Y_1(C_q)$ is spanned by the rows of $\rset''$ and the columns of $\wset^L(C_{q-1})\cup \wset(C_{q-1})\cup \wset^R(C_{q-1})$. The third corridor, $\Y_3(C_q)$ is spanned by the rows of $\rset'$ and the columns of $\wset^L(C_{q})\cup \wset(C_{q})\cup \wset^R(C_{q})$. For the second corridor, $\Y_2(C_q)$, we let $\rset^*$ be the set of consecutive rows starting from the top row of $\rset'$ and terminating at the bottom row of $\rset''$, including these two rows, and we let $\wset'(C_q)$ be the set of all columns lying between $B'(C_{q-1})$ and $B'(C_q)$. We then let $\Y_2(C_q)$ be the corridor spanned by the rows in $\rset^*$ and the columns in $\wset'(C_q)$. In order to obtain a valid snake, we extend $\Y_1(C_q)$ by one column to the right and $\Y_3(C_q)$ by one column to the left (see Figure~\ref{fig: yi-clauses2}).   
 
It is immediate to verify that for each $1\leq q\leq m$, we obtain a valid snake $\yset(C_q)$ of width at least $N_{i+1}$; the resulting snakes are mutually disjoint from each other, and each snake $\yset(C_q)$ intersects the boxes in $\set{B(C_1),\ldots,B(C_m)}$ only on their boundaries. For each $1\leq q< m$, we can now find the desired set $\tpset(C_q)$ of paths inside the snake $\yset(C_q)$, using Claim~\ref{claim: routing in a snake}. It is immediate to verify that the paths are order-preserving.
 \end{proof}
\end{proof} 



\label{-------------------------------------------sec: EDP-----------------------------------}
\section{Reducing the Maximum Vertex Degree and Hardness of \EDP in Sub-Cubic Planar Graphs}\label{sec: EDP}

In this section we prove Theorem \ref{thm: main-EDP}, and show that Theorem~\ref{thm: main} holds for sub-cubic planar graphs. We start with proving Theorem~\ref{thm: main-EDP}.
Let $G=G^{\ell,h}$ be a grid of length $\ell$ and height $h$, where $\ell>0$ is an even integer, and $h>0$. For every column $W_j$ of the grid, let $e^j_1,\ldots,e^j_{h-1}$ be the edges of $W_j$ indexed in their top-to-bottom order. Let $E^*(G)\subseteq E(G)$ contain all edges $e^j_z$, where $z\neq j \mod 2$, and let $\hat G$ be the graph obtained from $G\setminus E^*(G)$, by deleting all degree-$1$ vertices from it. The resulting graph is called a \emph{wall of length $\ell/2$ and height $h$} (see Figure~\ref{fig: wall}). Consider the subgraph of $\hat G$ induced by all horizontal edges of the grid $G$ that belong to $\hat G$. This graph is a collection of $h$ node-disjoint paths, that we refer to as the \emph{rows} of $\hat G$, and denote them by $R_1,\ldots,R_h$ in this top-to-bottom order; notice that $R_j$ is a sub-path of the $j$th row of $G$ for all $j$. Graph $\hat G$ contains a unique collection $\wset$ of $\ell/2$ node-disjoint paths that connect vertices of $R_1$ to vertices of $R_h$ and are internally disjoint from $R_1$ and $R_h$. We refer to the paths in $\wset$ as the \emph{columns} of $\hat G$, and denote them by $W_1,\ldots,W_{\ell/2}$ in this left-to-right order.
Paths $W_1, W_{\ell/2},R_1$ and $R_h$ are called the left, right, top and bottom boundary edges of $\hat G$, respectively, and the union of these paths is the boundary of $\hat G$.

\begin{figure}[h]
\scalebox{0.4}{\includegraphics{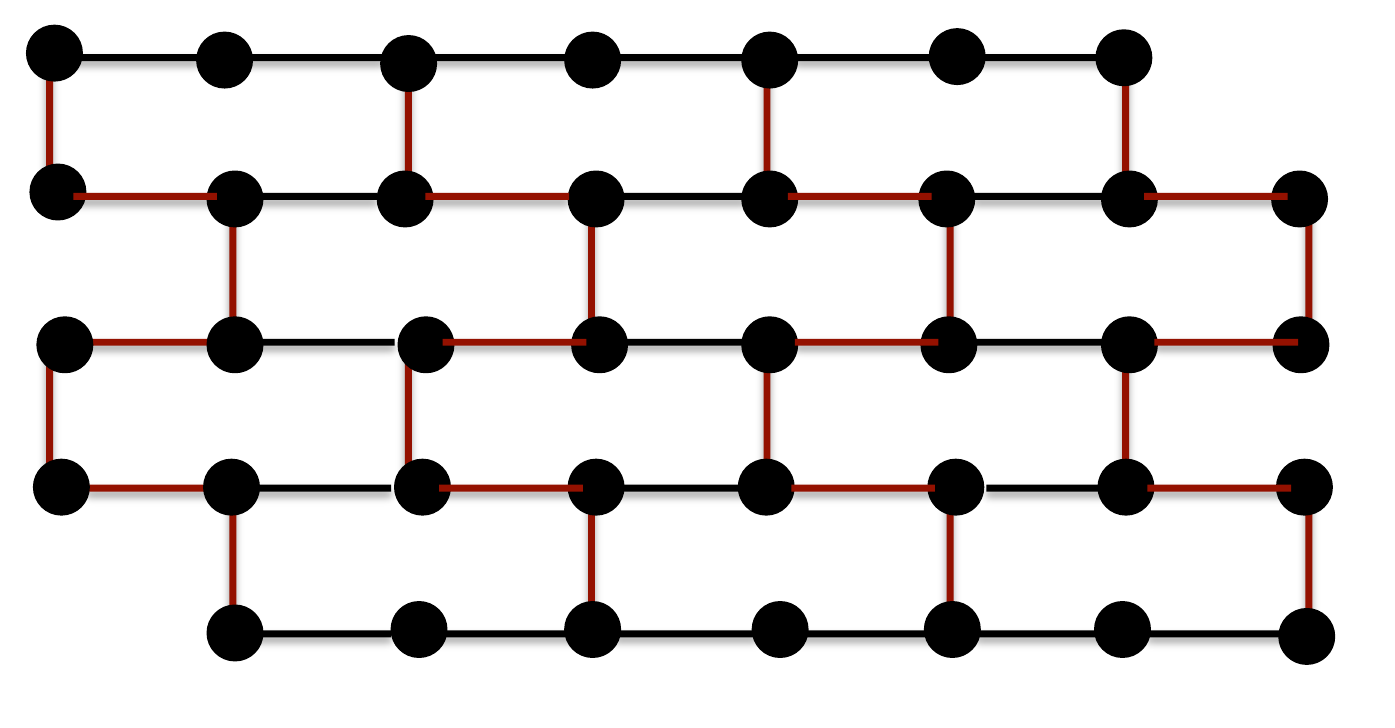}}
\caption {A wall of height $5$ and length $4$; the columns of the wall are shown in red. \label{fig: wall}}
\end{figure}

Given a wall $\hat G$, a consecutive subset $\rset'$ of its rows, and a consecutive subset $\wset'$ of its columns, the sub-wall of $\hat G$ spanned by the rows of $\rset'$ and the columns of $\wset'$ is the subgraph of $\hat G$ induced by the set $\set{v\mid \exists R\in \rset',W\in \wset':  v \in R\cap W }$ of vertices. The first and the last columns of $\wset'$ serve as the left and the right boundary edges of the sub-wall, and the top and the bottom rows of $\rset'$ serve as its top and bottom boundary edges.

We perform a reduction from the 3SAT(5) problem. Assume that we are given an instance $\phi$ of 3SAT(5) on $n$ variables and $m=5n/3$ clauses. As before, our construction has $\Theta(\log n)$ levels. For every level $i\geq 0$, we define a family of instances of \EDP. In order to construct a level-$i$ instance $\iset$, we define the parameters $H_i,L_i$ and $L'_i$ exactly as before, a path $Z(\iset)$ and a box $B(\iset)$ (which is a subgraph of the grid $G'_i$ of length $L'_i$ and height $H_i$), a collection $\mset$ of demand pairs, and the mappings of the vertices of $S(\mset)$ to the vertices of $Z(\iset)$, and of the vertices of $T(\mset)$ to distinct vertices of the middle row of $B(\iset)$ exactly as before. In order to instantiate this instance, we select an arbitrary grid $G_i$ of length $\ell\geq 2L_i+2L'_i+4H_i$, where $\ell$ is an even integer, and height $h\geq 3H_i$, map the vertices of $Z(\iset)$ to the vertices of the first row of $G_i$, and map the vertices of $B(\iset)$ to the vertices of a sub-grid $G''_i$ of $G_i$ exactly as before, obtaining an instantiation $(G,\mset)$ of the level-$i$ instance $\iset$ of \NDP. Our final step is to delete from $G$ every edge of $E^*(G_i)\cap E(G)$, and then to delete all vertices that have degree $1$ in the resulting graph. We also delete every other edge on the top row of $G_i$, and all horizontal edges that are incident to the vertices of $T(\mset)$, to ensure that the degree of every terminal is at most $2$. The final graph, denoted by $\hat G$, is a subgraph of a wall of length $\ell/2$ and height $h$. We denote by $\hat B(\iset)$ the intersection of the image of $B(\iset)$ in $G_i$ and the graph $\hat G$.
This concludes the definition of the reduction. Since the resulting graph $\hat G$ is a subgraph of $G$, the following observation is immediate.

\begin{observation}\label{obs: EDP-NI}
If $\phi$ is a \ni, then for every level $i$, for every instantiation $(G,\mset)$ of the level-$i$ instance $\iset$ of \NDP, the corresponding instance $(\hat G,\mset)$ of \EDP has optimal solution value at most $N'_i$.
\end{observation}

\begin{proof}
Assume otherwise. Consider an instantiation $(G,\mset)$ of a level-$i$ instance $\iset$ of \NDP, for some $i\geq 0$, and let $(\hat G,\mset)$ be the corresponding instance of \EDP. Let $\pset$ be a set of edge-disjoint paths in $\hat G$, routing more than $N'_i$ demand pairs. Since graph $\hat G$ is sub-cubic, and the degree of every terminal in $\hat G$ is at most $2$, it is immediate to verify that the paths in $\pset$ are also vertex-disjoint. Since $\hat G\subseteq G$, there is a set $\pset$ of node-disjoint paths routing more than $N'_i$ demand pairs in the instance $(G,\mset)$ of \NDP, contradicting Theorem~\ref{thm: ni}.
\end{proof}

It is now enough to show that, if $\phi$ is a \yi, then for every level $i$, for every instantiation $(G,\mset)$ of the level-$i$ instance $\iset$ of \NDP, the corresponding instance $(\hat G, \mset)$ of \EDP has a solution of value at least $N_i/2$. Before we do so, we need several definitions.

Suppose we are given some set $\pset$ of node-disjoint paths in some wall $\hat H$, and assume that every path in $\pset$ connects some vertex on a row $R$ of $\hat H$ to a vertex on a row $R'$ of $\hat H$, where $R\neq R'$. As before, we say that the paths in $\pset$ are \emph{order-preserving} iff the left-to-right ordering of their endpoints on $R$ is the same as the left-to-right ordering of their endpoints on $R'$.

A subset $U$ of the vertices lying on a row $R$ of a wall $\hat H$ is called \emph{well-spread} iff $U$ does not contain a pair of vertices connected by an edge in $\hat H$. Notice that if $U$ is well-spread, then no two vertices of $U$ may lie on the same column of $\hat H$.

We now define an analogue of box-respecting paths. Consider some level-$i$ instance $\iset$ of \NDP, for $i\geq 0$, and an instantiation $(G,\mset)$ of this instance. Let $(\hat G,\mset)$ be the corresponding instance of $\EDP$, and let $\pset$ be a set of node-disjoint paths routing some subset of the demand pairs in $\hat G$. Let $A$ be the top boundary of $\hat B(\iset)$ (that is, $A$ is the subgraph of $G''_i$, that contains all the edges and the vertices of its top row that belong to $\hat B(\iset)$). We say that set $\pset$ of paths is \emph{canonical with respect to the box $\hat B(\iset)$} iff for every path $P\in \pset$, $P\cap A$ is a single edge, and the following holds. Denote $\pset=(P_1,\ldots,P_r)$, and denote, for every path $P_i$, its source vertex by $s_i$, and the unique edge of $P_i\cap A$ by $e_i$, such that $s_1,\ldots,s_r$ appear on the top row of $\hat G$ in this left-to-right order. Then the edges of $e_1,\ldots,e_r$ must appear in this left-to-right order on $A$, and for each $1\leq j\leq r$, $e_j$ is the $(2j)$th edge of $A$ from the left.

Assume that $\phi$ is a \yi.
Recall that for each $i\geq 0$, for every level-$i$ instance $\iset$ of \NDP, we have defined a collection $\mset^*(\iset)$ of demand pairs, such that for every instantiation of $\iset$, there is a set $\pset$ of node-disjoint paths, that respect the box $B(\iset)$, and route the set $\mset^*(\iset)$ of demand pairs. Denote the pairs in $\mset^*(\iset)$ by $(s_1,t_1),\ldots,(s_p,t_p)$, and assume that the vertices $s_1,\ldots, s_p$ appear in this left-to-right order on $Z(\iset)$. We partition $\mset^*(\iset)$ into two subsets: set $\mset^*_1(\iset)$ contains all demand pairs $(s_j,t_j)$ where $j$ is odd, and set $\mset^*_2(\iset)$ contains all remaining demand pairs. Notice that each of the sets $S(\mset^*_1(\iset)),S(\mset^*_2(\iset))$ is well-spread, for any instantiation of $\iset$ and its corresponding instance of \EDP. It is now enough to prove the following lemma.

\begin{lemma}\label{lem: EDP-YI}
For all $i\geq 0$, for every level-$i$ instance $\iset$ of \NDP, for every instance $(\hat G, \mset)$ of \EDP corresponding to an instantiation $(G,\mset)$ of $\iset$, there is a set of node-disjoint paths in $\hat G$ that route all demand pairs in $\mset^*_1(\iset)$, such that the paths are canonical with respect to $\hat B(\iset)$, and the same holds for $\mset^*_2(\iset)$.
\end{lemma}

\begin{proof}
The proof is by induction on $i$. The lemma is clearly true for $i=0$. We now assume that it holds for some $i\geq 0$ and prove it for a level-$(i+1)$ instance $\iset$ of \NDP. Let $(G,\mset)$ be any instantiation of $\iset$, and let $(\hat G,\mset)$ be the corresponding instance of \EDP. We show that the set  $\mset^*_1(\iset)$ of demand pairs can be routed in $\hat G$ via node-disjoint paths, that are canonical with respect to $\hat B(\iset)$; the proof for set $\mset^*_2(\iset)$ is identical.

Consider some level-$i$ instance $\iset'$ of \NDP that was used in the construction of the level-$(i+1)$ instance $\iset$. The current instantiation of $\iset$ also defines an instantiation of $\iset'$. From the induction hypothesis, for all $z\in \set{1,2}$, there is a set $\pset_z(\iset')$ of node-disjoint paths in $\hat G$, routing the demand pairs in $\mset^*_z(\iset')$, that respect the box $\hat B(\iset')$. Moreover, the subset of all demand pairs in $\mset^*_1(\iset)$ whose destinations lie in $\hat B(\iset')$ is either equal to $\mset^*_1(\iset')$, or to $\mset^*_2(\iset')$.

The routing that we define for the set $\mset^*_1(\iset)$ of demand pairs in graph $\hat G$ is very similar to the one used in the proof of Lemma~\ref{lem: routing yi}. However, we need to define corridors and snakes slightly differently.

Recall that the graph $\hat G$ is a subgraph of a wall $\hat G_i$, that is obtained from $G_i$, by deleting the set $E^*(G_i)$ of its edges, the vertices of $U(B(\iset))$, and all vertices whose degree in the resulting graph becomes $1$. Given a sub-wall $\Y$ of $\hat G_i$, spanned by a subset $\rset$ of its rows and a subset $\wset$ of its columns, we say that $\Y$ is a \emph{corridor} iff $\Y\subseteq \hat G$. We say that two corridors $\Y,\Y'$ are \emph{internally disjoint} iff every vertex in $\Y\cap \Y'$ belongs to a single boundary edge of each corridor, and this boundary edge must be either the top or the bottom edge. We say that $\Y$ and $\Y'$ are \emph{neighbors} iff they are internally disjoint and $\Y\cap \Y'\neq \emptyset$.
As before, a snake of length $\ell$ is a sequence $\Y_1,\ldots,\Y_{\ell}$ of corridors that are internally disjoint, such that for all $1\leq j,j'\leq \ell$, corridors $\Y_j$ and $\Y_{j'}$ are neighbors iff $|j-j'|=1$. We say that the width of the snake is $w$ iff each of its corridors is spanned by at least $w$ rows and at least $w$ columns of $\hat G_i$, and for all $1\leq j<\ell$, $\Y_j\cap \Y_{j+1}$ contains at least $2w$ vertices. Following is an analogue of Claim~\ref{claim: routing in a snake} for wall graphs; its proof is almost identical and is omitted here.

 \begin{claim}\label{claim: routing in a wall-snake}
 Let $\yset=(\Y_1,\ldots,\Y_{\ell})$ be a snake of width $w$, and let $A,A'$ be two sets of vertices with $|A|=|A'|\leq w-2$, such that the vertices of $A$ either all lie on the top boundary edge of $\Y_1$, or they all lie on the bottom boundary edge of $\Y_1$, and the vertices of $A'$ either all lie on the top boundary edge of $\Y_{\ell}$, or they all lie on the bottom boundary edge of $\Y_{\ell}$. Assume further that both $A$ and $A'$ are well-spread.
 Then there is a set $\qset$ of node-disjoint paths contained in $\bigcup_{\ell'=1}^{\ell}\Y_{\ell'}$, that  connect every vertex of $A$ to a distinct vertex of $A'$. 
 \end{claim}
We fix some assignment $\aset$ to the variables of $\phi$ that satisfies all clauses.
For every clause $C_q\in \cset$, we select a set $A(C_q)$ of vertices on the top boundary of $B(C_q)$ exactly as before, except that now we discard every other vertex from this set, in order to guarantee that these vertices are well-spread. Similarly, for every variable $x$ of $\phi$, we select the set $A(x)$ of vertices on the top boundary of $B(x)$ exactly as before, and then discard every other vertex from this set. 

We select the sets $\Gamma$ and $\Gamma''$ of vertices on the top and the bottom boundaries of $B^V$ exactly as before, but now we are guaranteed that the vertices in each set are well-spread. The vertices of $Z'$ and $\Gamma'''$ are selected like before, except that now discard every other vertex in each set to ensure that the resulting sets are well-spread.

The remainder of the routing in Lemma~\ref{lem: routing yi} relies on the constructions of various snakes. For every snake $\yset=(\Y_1,\ldots,\Y_r)$ that the construction uses, it is easy to modify the corresponding corridors, such that for all $1\leq r'<r$, $\Y_{r'}\cap \Y_{r'+1}$
is contained in either the top or the bottom boundary edge of each corridor. This allows us to use Claim~\ref{claim: routing in a wall-snake} instead of the original Claim~\ref{claim: routing in a snake} in order to route inside the snakes. Even though each resulting snake in the wall graph will only be able to route half the number of paths it routed before, the amount of horizontal space that we left between the various boxes in the construction of the level-$(i+1)$ instance is sufficient to ensure that all demand pairs are routed. As observed above, for each level-$i$ instance $\iset'$ used in the construction of $\iset$, the subset of all demand pairs in $\mset^*_1(\iset)$, whose destinations lie in $\hat B(\iset')$, is either equal to $\mset^*_1(\iset')$ or to $\mset^*_2(\iset')$, and from the induction hypothesis, each one of these sets of demand pairs has a set of node-disjoint paths routing it in $\hat G$, that is canonical with respect to $\hat B(\iset')$. We use this fact to complete the routing inside the boxes $\hat B(\iset')$ of the corresponding level-$i$ instances $\iset'$.
\end{proof}

\paragraph{Hardness of \NDP on Sub-Cubic Planar Graphs} Consider the instances of \EDP constructed above. Each such instance is defined on a sub-cubic planar graph, where the degree of every terminal is at most $2$. It is easy to see that, if we are given a graph $G$ with the above properties, and any set $\pset$ of paths whose endpoints are distinct terminals, the paths in $\pset$ are mutually edge-disjoint iff they are mutually node-disjoint. Therefore, the number of the demand pairs that can be routed in the \yi and the \ni via node-disjoint paths remains the same as for edge-disjoint paths. This completes the proof of Theorem~\ref{thm: main}.

\label{--------------------------------------sec: large degree NDP---------------------------------------}
\section{Maximum Vertex Degree in \NDP Instances}\label{sec: large-degree-NDP}


 Given an instance $(G,\mset)$ of the \NDP problem, we denote by $\tset(\mset)$ the set of all vertices that participate in the demand pairs in $\mset$, and we refer to them as \emph{terminals}. We start by defining equivalence between \NDP instances.

\begin{definition}
We say that two instances $(G,\mset)$, $(G',\mset')$ of the \NDP problem are \emph{equivalent} iff there is a bijection $f: \tset(\mset) \rightarrow \tset(\mset')$, such that, for every subset $\tmset\subseteq \mset$ of demand pairs, there is a set $\pset$ of node-disjoint paths in graph $G$ routing all  pairs in $\tmset$ iff there is a set $\pset'$ of node-disjoint paths in graph $G'$ routing all demand pairs in set $\tmset'=\set{(f(s),f(t))\mid (s,t)\in \tmset}$.
\end{definition}

The goal of this section is to prove the following theorem.

\begin{theorem} \label{thm: planar neq subgraph of grids}
For any integer $d>3$, there is an instance $(G,\mset)$ of the \NDP problem, where $G$ is a planar graph, such that for every instance $(G',\mset')$ of \NDP that is equivalent to $(G,\mset)$, there is a vertex in $G'$ of degree at least $d$.
\end{theorem}

In a sense, the above theorem shows that the class of all planar graphs is strictly more general than the class of all planar graphs with maximum vertex degree at most $3$ in the context of the \NDP problem, and so our hardness result in Theorem~\ref{thm: main} holds even for a restricted family of planar graphs.

\begin{proof}
Instance $(G,\mset)$ of \NDP is defined as follows. Graph $G$ is simply a star graph with $d+2$ vertices (see Figure~\ref{fig: star}), so $V(G)=\set{v_1,\ldots,v_{d+2}}$, and $E(G)=\set{e_1,\ldots,e_{d+1}}$, where for $1\leq i\leq d+1$, $e_i=(v_i,v_{d+2})$. We let $\mset=\set{(v_i,v_j)\mid 1\leq i<j\leq d+2}$. It is easy to verify that a subset $\mset'\subseteq \mset$ of demand pairs can be routed via node-disjoint paths in $G$ if and only if $|\mset'|=1$. Note that $G$ is planar.

\begin{figure}[h!]
\centering
\subfigure[Graph $G$ for the case where $d=5$]{\scalebox{0.45}{\includegraphics{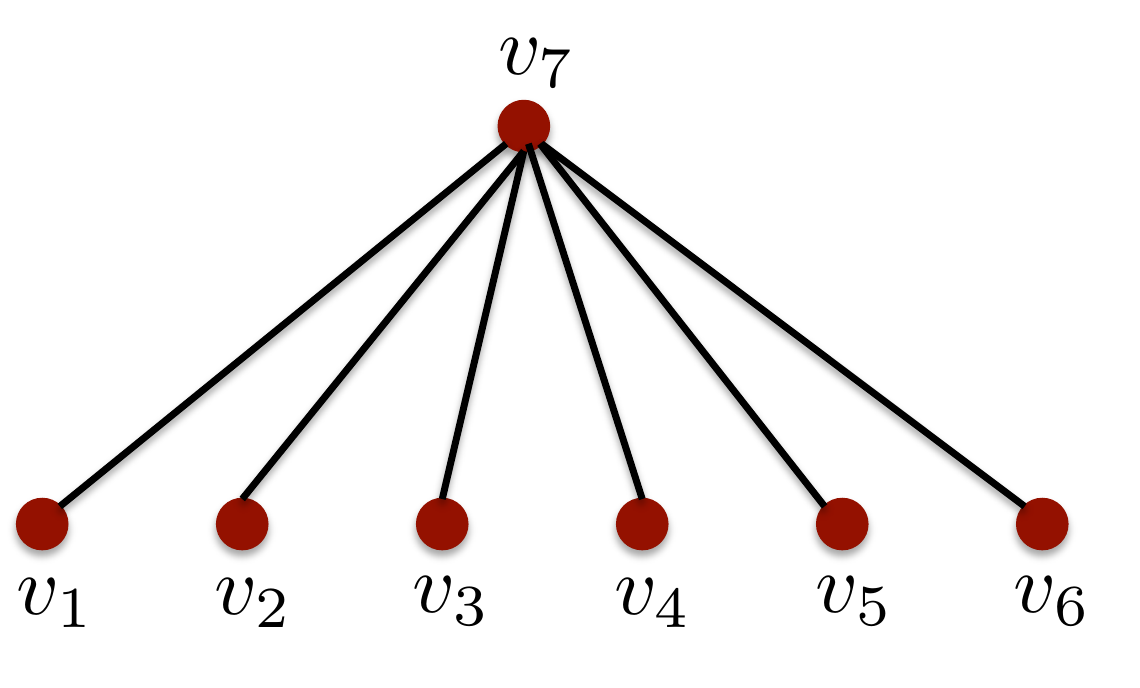}}\label{fig: star}}
\hspace{1cm}
\subfigure[Constructing two paths in the case where $u_3\neq u_4$.]{\scalebox{0.4}{\includegraphics{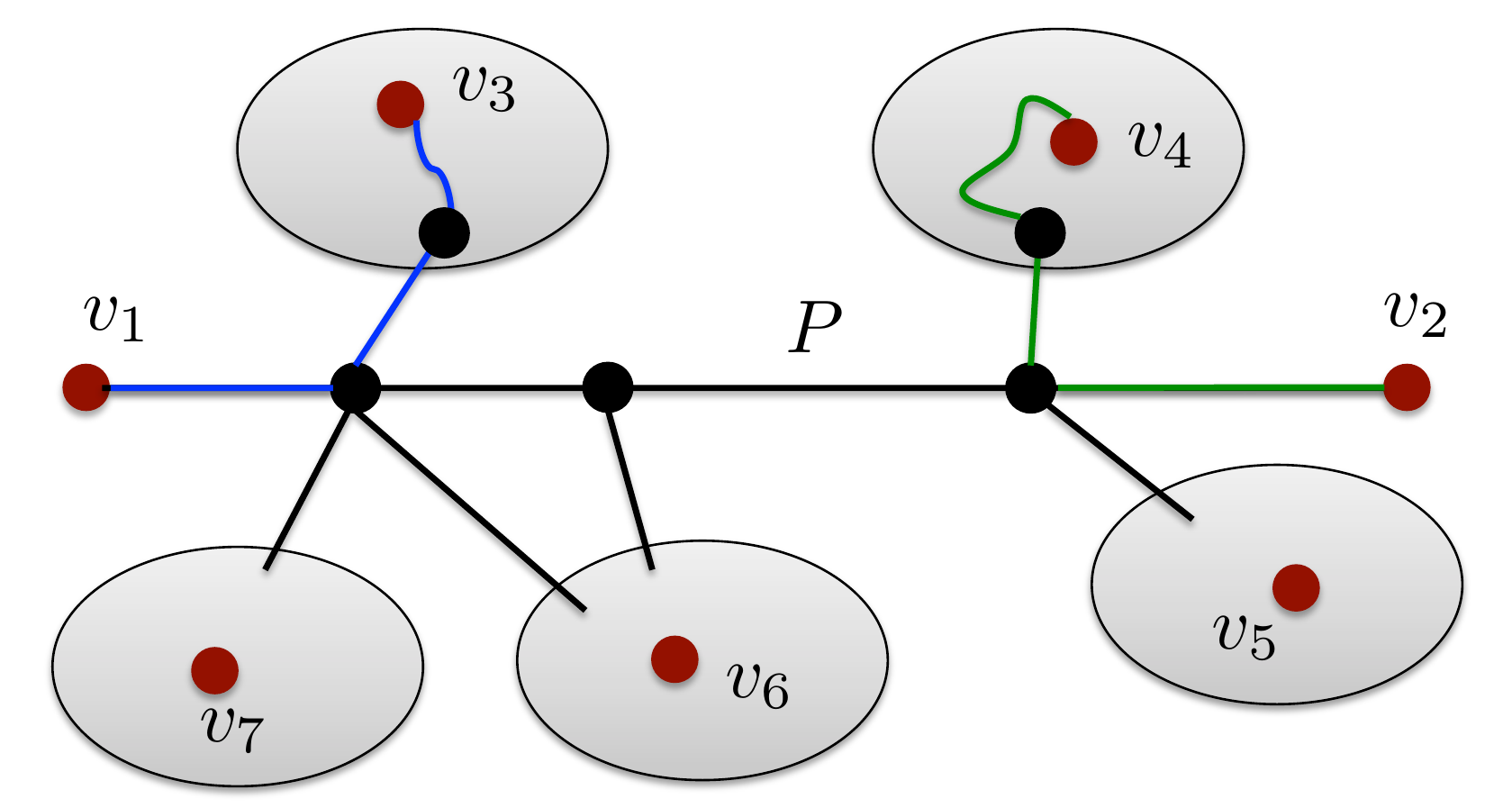}}\label{fig: two paths}}
\caption{Illustration to the proof of Theorem~\ref{thm: planar neq subgraph of grids}}
\end{figure}

Assume for contradiction that there is an instance $(G',\mset')$ of \NDP that is equivalent to $(G,\mset)$, such that the maximum vertex degree in $G'$ is less than $d$. For convenience, for $1\leq i\leq d+2$, we denote the vertex $f(v_i)$ of $G'$ by $v_i$.
Let $P$ be the shortest path in $G'$ connecting a pair of terminals. We assume without loss of generality that $P$ connects $v_1$ to $v_2$; so $P$ does not contain any terminals as inner vertices. Consider the graph $H=G'\setminus P$. We claim that every vertex of $\set{v_3,\ldots,v_{d+2}}$ must belong to a distinct connected component of $H$. Indeed, if, for example, $v_3$ and $v_4$ belong to the same connected component, then we can simultaneously route $(v_3,v_4)$ and $(v_1,v_2)$ in $G'$. Since we cannot do the same in $G$, the two instances are not equivalent. For $3\leq i\leq d+2$, let $C_i$ be the connected component of $H$ containing $v_i$, and let $e_i$ be any edge connecting a vertex of $C_i$ to a vertex of $P$ (such an edge must exist as there must be a path connecting $v_i$ to $v_1$ in $G$). Let $u_i$ denote the endpoint of $e_i$ that lies on $P$.

Assume first that for some $3\leq i\neq j\leq d+2$, $u_i\neq u_j$. Assume without loss of generality that $u_i$ lies closer to $v_1$ on $P$ than $u_j$. Then we can simultaneously route two pairs $(v_i,v_1)$ and $(v_j,v_2)$ via node-disjoint paths, as follows (see Figure~\ref{fig: two paths}). The first path uses the segment of $P$ from $v_1$ to $u_i$, the edge $e_i$, and some path connecting an endpoint of $e_i$ to $v_i$ inside $C_i$. The second path similarly uses the segment of $P$ from $v_2$ to $u_j$, the edge $e_j$, and some path connecting an endpoint of $e_j$ to $v_j$ inside $C_j$. Clearly, these two paths are disjoint. But we cannot route the pairs $(v_i,v_1)$ and $(v_j,v_2)$ simultaneously in $G$, and so the two instances are not equivalent.

We conclude that $u_3=u_4=\cdots=u_{d+2}$ must hold. But then $G'$ must contain a vertex of degree at least $d$, a contradiction.
\end{proof}

\label{--------------------------------------sec: large degree EDP---------------------------------------}
\section{Maximum Vertex Degree in Planar \EDP Instances}\label{sec: EDP-degree-reduction}

In this section we show that there is an instance $(G, \mset)$ of the \EDP problem, where $G$ is a planar graph, such that for every instance $(G',\mset')$ of \EDP that is equivalent to $(G,\mset)$, where $G'$ is a planar graph, at least one vertex of $G'$ must have degree at least $4$. As before, given an instance $(G,\mset)$ of the \EDP problem, the set of all vertices participating in the demand pairs is denoted by $\tset(\mset)$, and the vertices in this set are called terminals. We say that a subset $\tmset\subseteq \mset$ of demand pairs is \emph{routable} in $G$ iff there is a set $\pset$ of edge-disjoint paths in $G$, routing every demand pair in $\tmset$.

The equivalence between the \EDP instances is defined similarly to the equivalence between the \NDP instances, except that now we only consider restricted subsets of demand pairs: a subset $\tmset\subseteq \mset$ of demand is called \emph{restricted} iff every terminal participates in at most one demand pair of $\tmset$. We note that in the \NDP problem, if a set $\tmset$ of demand pairs is routable via node-disjoint paths in a graph $G$, then $\tmset$ must be a restricted set; this is not necessarily true for \EDP instances. We are now ready to define the equivalence between the instances of \EDP.

\begin{definition}
Two instances $(G,\mset)$ and $(G',\mset')$ of the \EDP problem are equivalent iff there is a bijection $f : \tset(\mset) \rightarrow \tset(\mset')$, such that for every restricted subset $\tilde {\mset} \subseteq \mset$ of demand pairs, $\tmset$ is routable in $G$ iff the set $\tmset'=\{(f(s),f(t)) \mid (s,t) \in \tilde \mset\}$ of demand pairs is routable in $G'$.
\end{definition}

The goal of this section is to prove the following theorem.

\begin{theorem}
There is an instance $(G,\mset)$ of the \EDP problem, where $G$ is a planar graph, such that for every instance $(G',\mset')$ of \EDP that is equivalent to $(G,\mset)$, where $G'$ is planar, some vertex of $G'$ has degree at least 4.
\end{theorem}

 We note that if we did not require that the sets $\tmset$ are restricted in the definition of the equivalence between \EDP instances, then the theorem would be much easier to prove --- in fact one could prove that for every $d$, there is an instance $(G,\mset)$ of \EDP where $G$ is planar, such that for every instance $(G',\mset')$ of \EDP equivalent to $(G,\mset)$, the maximum vertex degree in $G'$ must be at least $d$. This can be done by letting the demand pairs in $\mset$ correspond to the edges of the star graph with $d$ leaves. 

\begin{proof}
We construct an instance $(G,\mset)$ of \EDP with a set $\tset$ of terminals that consists of $25$ vertices, partitioned into $5$ subsets $\tset_1,\ldots,\tset_5$ of $5$ vertices each. For each $1\leq i\leq 5$, we let $\tset_i = \set{t^i_{j} \mid 1\leq j\leq 5}$, and we let $\tset=\bigcup_{i=1}^5\tset_i$. In order to construct the graph $G$, we start with the set $\tset$ of terminals. For each subset $\tset_i$, for $1\leq i\leq 5$, we add a vertex $a_i$, that connects to every terminal in $\tset_i$ with an edge. For all $1\leq i<i'\leq 5$, we also add the edge $(a_i,a_{i'})$ to the graph. To summarize, so far, $V(G)=\tset\cup\set{a_1,\ldots,a_5}$, and $E(G)=\set{(t, a_i) \mid 1\leq i\leq 5, t \in \tset_i} \cup \set{(a_i, a_j)\mid 1\leq i<j\leq 5}$.
We fix a drawing of $G$ in the plane, such that there is exactly one crossing of the edges in $\set{(a_i, a_j)\mid 1\leq i<j\leq 5}$ and no other crossings. We then replace this crossing with a vertex $b$, so that $G$ becomes a planar graph (see Figure~\ref{fig: edp-hard-instance}). Let $\mset$ contain all pairs $(t,t')$ of terminals, where $t\neq t'$. This finishes the definition of the instance $(G,\mset)$.

\begin{figure}[h]
\centering
\scalebox{0.33}{\includegraphics{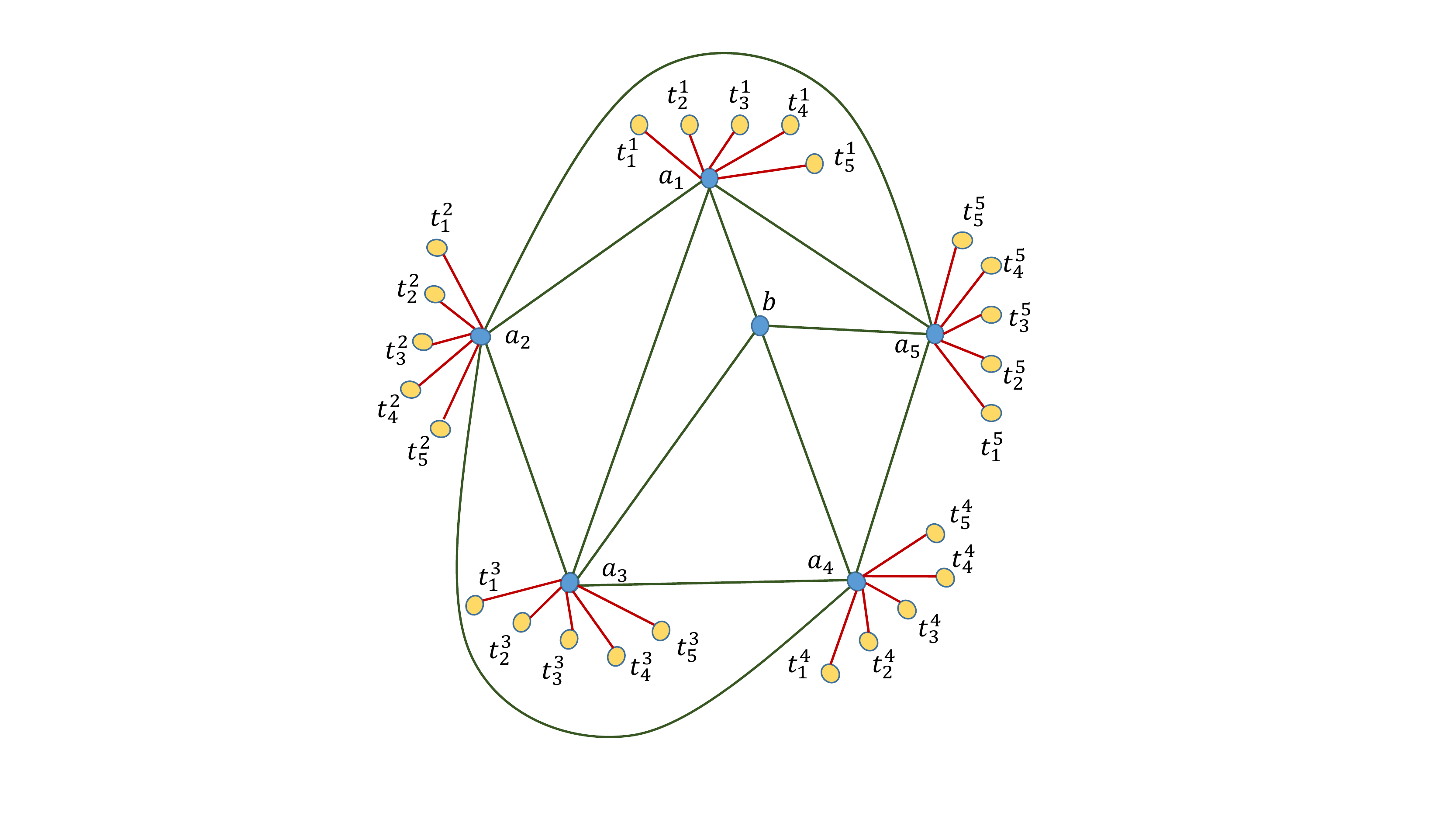}}
\caption{Graph $G$. \label{fig: edp-hard-instance}}
\end{figure}

The following two observations are immediate.

\begin{observation}\label{obs: not routable}
Let $\hmset\subseteq \mset$ be a restricted set of demand pairs, such that for some $1\leq i\leq 5$, each demand pair in $\hmset$ contains a single distinct terminal of $\tset_i$ (and another terminal in $\tset\setminus \tset_i$). Then $\hmset$ is routable in $G$ iff $|\hmset|\leq 4$.
\end{observation}

\begin{observation}\label{obs: routable}
Let $\hmset\subseteq \mset$ be a set of demand pairs, defined as follows. Consider the graph $H=K_5$ (a clique on $5$ vertices), and denote its vertices by $v_1,\ldots,v_5$. Set $\hmset$ of demand pairs contains, for every edge $e=(v_i,v_j)\in E(H)$, a pair $(s_e,t_e)$, such that $s_e\in \tset_i$ and $t_e\in \tset_j$. The vertices $s_e,t_e$ for all $e\in E(H)$ are chosen in way that ensures that $\hmset$ is a restricted set of demand pairs. Then $\hmset$ is routable in $G$.
\end{observation}


Assume now for contradiction that there is an instance $(G',\mset')$ of \EDP that is equivalent to $(G,\mset)$, such that $G'$ is planar, and the maximum vertex degree in $G'$ is at most $3$. 

For convenience, for every terminal $t\in \tset$, we denote the corresponding terminal $f(t)$ of $G'$ by $t$. We also define the partition $\set{\tset_i}_{i=1}^5$ of the new terminals as before. 
Fix some index $1\leq i\leq 5$, and let $(A_i,B_i)$ be a partition of $V(G')$ with $\tset_i\subseteq A_i$ and $\tset\setminus\tset_i\subseteq B_i$ that minimizes $|E(A_i,B_i)|$. Among all such partitions, choose the one where $|A_i|$ is the smallest. Since instances $(G,\mset)$, $(G',\mset')$ are equivalent, from Observation~\ref{obs: not routable} and the max-flow min-cut theorem, $|E(A_i,B_i)|= 4$. We now claim that all sets $A_i$ are mutually disjoint.

\begin{claim}\label{claim: disj}
For all $1\leq i<j\leq 5$, $A_i\cap A_j=\emptyset$.
\end{claim}
\begin{proof}
For every set $U\subseteq V(G')$ of vertices, let $\delta(U)=|E(U,V(G')\setminus U)|$. From the sub-modularity of cuts:

\[\delta(A_i)+\delta(A_j)\geq \delta(A_i\setminus A_j)+\delta(A_j\setminus A_i).\]

Recall that $\delta(A_i)=\delta(A_j)=4$. From our definition of the sets $\set{A_k}_{k=1}^5$, $\tset_i\subseteq A_i\setminus A_j$ and $\tset_j\subseteq A_j\setminus A_i$,  so $\delta(A_i\setminus A_j),\delta(A_j\setminus A_i)\geq 4$. This can only happen if both values are equal to $4$, contradicting our choice of $A_i$ and $A_j$.
\end{proof}

We will show that $G'$ contains the graph $H=K_5$ as a minor, which will lead to a contradiction, since $G'$ is a planar graph. We will embed every vertex $v_i\in V(H)$ into the subgraph $G'[A_i]$ of $G'$. In order to be able to do so, we need to show each such subgraph is connected.


\begin{observation}
For all $1\leq i\leq 5$, $G'[A_i]$ is connected.
\end{observation}
\begin{proof}
Assume otherwise, and let $\cset$ be the set of all connected components of $G'[A_i]$. By the definition of $A_i$, each such connected component must contain at least one terminal of $\tset_i$. Assume for contradiction that $|\cset|>1$. Then there are two terminals of $\tset_i$ that lie in different connected components of $\tset_i$. Assume w.l.o.g. that these terminals are $t_1^i$ and $t_2^i$. Consider the following set $\mset''$ of demand pairs: $\mset''$ contains the pair $(t_1^i,t_2^i)$, and for all $3\leq j\leq 5$, it contains a pair $(t_j^i,t_{j}^{i'})$ for some index $i'\neq i$. It is easy to verify that the set $\mset''$ of demand pairs is restricted and it is routable in $G$, and so it must be routable in $G'$. Each of the pairs $\set{(t_j^i,t_{j}^{i'})}_{j=3}^5$ is routed on a path that must use at least one edge of $E(A_i,B_i)$. But the path routing the pair $(t_1^i,t_2^i)$ must use two edges of $E(A_i,B_i)$, which is impossible since $|E(A_i,B_i)|=4$.\end{proof}

Let $H=K_5$, and denote $V(H)=\set{v_1,\ldots,v_5}$. We define a set $\hmset\subseteq\mset'$ of demand pairs exactly as in Observation~\ref{obs: routable}. Since the set $\hmset$ of demand pairs is routable in $G$, it must also be routable in $G'$. Let $\pset$ be the set of edge-disjoint paths routing the set $\hmset$ of demand pairs in $G'$. For each edge $e=(v_i,v_{i'})\in E(H)$, we let $P_{i,i'}$ be the path in $\pset$, routing the pair $(s_e,t_e)$. Assume that $i<i'$, and think of the path $P_{i,i'}$ as directed from a terminal in $\tset_i$ to a terminal in $\tset_{i'}$. Let $u$ be the last vertex of $P_{i,i'}$ that belongs to $A_i$, and let $u'$ be the first vertex of $P_{i,i'}$ that appears after $u$ on the path and belongs to $A_{i'}$. We then let $P'_{i,i'}$ be the sub-path of $P_{i,i'}$ from $u$ to $u'$. Note that $P'_{i,i'}$ does not contain any vertex of $A_i\cup A_{i'}$ as an inner vertex. We need the following stronger claim.

\begin{observation}
Path $P_{i,i'}$ does not contain any vertex of $\bigcup_{k=1}^5A_k$ as an inner vertex.
\end{observation}
\begin{proof}
Assume otherwise, and let $w\in \bigcup_{k=1}^5A_k$ be some vertex that belongs to $P_{i,i'}$ as an inner vertex. Assume that $w\in A_k$. From the above discussion, $k\not\in \set{i,i'}$. But there are four paths in $\pset$ that terminate at the vertices of $\tset_k$: the paths routing the demand pairs corresponding to the edges of $H$ that are incident to $v_k$. Each such path contains an edge of $E(A_k,B_k)$, and $|E(A_k,B_k)|=4$. Therefore, path $P_{i,i'}$ cannot contain an edge of $E(A_k,B_k)$ and thus it cannot contain a vertex of $A_k$.
\end{proof}

Let $\pset'=\set{P'_{i,i'}\mid P_{i,i'}\in \pset}$.
Since the maximum vertex degree in $G'$ is $3$ and the paths in $\pset$ are edge-disjoint, they are almost vertex-disjoint: the only way for two paths in $\pset$ to share a vertex is when that vertex is an endpoint of one of these paths. It is then easy to see that whenever two paths in $\pset'$ share a vertex, that vertex must be an endpoint of each of these paths.

It is now immediate to show that $H$ is a minor of $G'$: we map every vertex $v_i$ of $H$ to the connected subgraph $G'[A_i]$ of $G'$, and every edge $e$ of $H$ to the path $P'_e$ of $\pset'$. But $G'$ is a planar graph and cannot contain a $K_5$-minor, a contradiction.
\end{proof}


\bibliography{NDP-hardness.v4}
 \bibliographystyle{alpha}
\end{document}